\documentclass[11pt]{article}

\RequirePackage{fix-cm}

\usepackage{setspace}

\usepackage{etoolbox}

\usepackage{anyfontsize}

\usepackage[letterpaper,margin=1in]{geometry}

\usepackage[T1]{fontenc}

\usepackage{bbold}
\usepackage{amsmath}

\usepackage{amsthm}

\usepackage{amssymb}
\usepackage{pgffor}
\usepackage[inline]{enumitem}

\setlist[enumerate]{nosep,topsep=0.3em}
\setlist[enumerate,1]{label=(\roman*)}
\setlist[enumerate,2]{label=(\alph*)}

\setlist[itemize]{nosep,topsep=0.1em}

\newlist{algostepsarabic}{enumerate}{1}
\setlist[algostepsarabic]{label=\arabic*., ref=\arabic*, nosep, itemsep=.3em, labelwidth=0.9em, labelindent=-.8em, leftmargin=!, rightmargin=8mm}

\usepackage{mathtools}

\usepackage{bm}
\usepackage{microtype}
\usepackage{xspace}
\usepackage{xfrac}
\ExplSyntaxOn
\DeclareRestrictedTemplate { xfrac } { text } { math }
  {
    numerator-font      = \number \fam ,
    slash-symbol        = /            ,
    slash-symbol-font   = \number \fam ,
    denominator-font    = \number \fam ,
    scale-factor        = 0.7          ,
    scale-relative      = false        ,
    scaling             = true         ,
    denominator-bot-sep = 0 pt         ,
    math-mode           = true         ,
    phantom             = ( %
  }
\DeclareInstance { xfrac } { mathdefault } { math }
  { numerator-top-sep = 0pt }
\ExplSyntaxOff

\usepackage{times}

\usepackage{array}
\usepackage[table]{xcolor}
\usepackage{multirow}
\usepackage{makecell}
\usepackage{booktabs}

\usepackage{changepage}

\usepackage[bottom]{footmisc}

\usepackage{url}
\usepackage{hyperref}

\hypersetup{colorlinks, linkcolor=darkblue, citecolor=darkgreen, urlcolor=darkblue}

\usepackage[
backend=biber,
style=alphabetic,
citestyle=alphabetic,
maxalphanames=6,
maxcitenames=99,
mincitenames=98,
maxbibnames=99,
giveninits=true,
url=false,
doi=true,
isbn=true,
backref=true
]{biblatex}

\usepackage{xpatch}

\makeatletter
\patchcmd\blx@bblinput{\blx@blxinit}
                      {\blx@blxinit
                      }{}{\fail}
\makeatother
\makeatletter
\DeclareFieldFormat{eprint:arxiv}{%
  arXiv\addcolon\space
  \ifhyperref
    {\href{https://arxiv.org/abs/#1}{%
       \nolinkurl{#1}%
       \iffieldundef{eprintclass}
     {}
     {\addspace\mkbibbrackets{\thefield{eprintclass}}}}}
    {\nolinkurl{#1}
     \iffieldundef{eprintclass}
       {}
       {\addspace\mkbibbrackets{\thefield{eprintclass}}}}}
\makeatother

\usepackage{float}
\usepackage{graphicx}

\usepackage[margin=10pt,font=small,labelfont=bf]{caption}
\usepackage{subcaption}

\graphicspath{{graphics/}}

\usepackage{tikz}

\usetikzlibrary{calc}
\usetikzlibrary{decorations.pathreplacing}
\usetikzlibrary{decorations.pathmorphing}
\usetikzlibrary{shapes.geometric}
\usetikzlibrary{fadings,decorations,fit}
\usetikzlibrary{positioning}

\pgfdeclarelayer{bg}
\pgfdeclarelayer{fg}
\pgfsetlayers{bg,main,fg}

\tikzstyle{guessedIn}+=[line width=3pt, blue]
\tikzstyle{guessedOut}+=[densely dotted, red]
\tikzstyle{ns}+=[thick,draw=black,fill=white,circle,minimum size=1.5mm, inner sep=0pt]

\definecolor{darkblue}{rgb}{0,0,0.38}
\definecolor{darkred}{rgb}{0.6,0,0}
\definecolor{darkgreen}{rgb}{0.1,0.35,0}
\colorlet{cutcol}{white!50!black}

\usepackage{wrapfig}

\usepackage[vlined,algoruled,algo2e]{algorithm2e}
\SetAlgoSkip{medskip}

\usepackage{mdframed}

\makeatletter
\newcommand{\labeltarget}[1]{\Hy@raisedlink{\hypertarget{#1}{}}}
\makeatother

\usepackage[capitalize, nameinlink, sort]{cleveref}
\crefname{theorem}{Theorem}{Theorems}
\crefname{lemma}{Lemma}{Lemmas}
\Crefname{claim}{Claim}{Claims}
\Crefname{fact}{Fact}{Facts}
\Crefname{remark}{Remark}{Remarks}
\Crefname{observation}{Observation}{Observations}
\Crefname{algocf}{Algorithm}{Algorithms}
\Crefname{property}{Property}{Properties}
\Crefname{line}{Line}{Lines}
\crefalias{AlgoLine}{line}
\Crefname{figure}{Figure}{Figures}
\Crefname{algostepsarabici}{Step}{Steps}

\newtheorem{theorem}{Theorem}
\newtheorem{lemma}[theorem]{Lemma}

\newtheorem{proposition}[theorem]{Proposition}
\newtheorem{definition}[theorem]{Definition}

\newtheorem{observation}[theorem]{Observation}

\newtheorem{property}[theorem]{Property}

\newlength\bxheight
\newcommand\TJdomRaw[1]{\ensuremath{P_{#1\textrm{-join}}^\uparrow}}
\newcommand\TJdom[1]{\raisebox{0pt}[\bxheight]{\TJdomRaw{#1}}}

\newcommand\PST{\ensuremath{P_{\mathrm{ST}}}\xspace}

\newcommand\OPT{\ensuremath{\mathrm{OPT}}\xspace}
\renewcommand\P{\ensuremath{\mathrm{P}}\xspace}
\newcommand\NP{\ensuremath{\mathrm{NP}}\xspace}
\newcommand\Oh{\mathcal{O}}

\newcommand\sfc[1][]{\ensuremath{(S_{#1},F_{#1},\mathcal{C}_{#1})}}
\newcommand\sfcbar[1][]{\ensuremath{(\smash{\overline{S}}_{#1},\smash{\overline{F}}_{#1},\smash{\overline{\mathcal{C}}}_{#1})}}
\newcommand\sfcprime{\ensuremath{(S',F',\mathcal{C}')}}

\DeclareMathOperator{\conv}{conv}
\DeclareMathOperator{\cone}{cone}
\DeclareMathOperator{\supp}{supp}

\DeclareMathOperator{\argmin}{argmin}
\DeclareMathOperator{\odd}{odd}
\newcommand{\symdiff}{\mathbin{\bigtriangleup}}

\makeatletter
\newcommand{\manuallabel}[1]{\def\@currentlabel{#1}\label{#1}}
\makeatother

\usepackage{xspace}

\newcommand{\MCCST}{\hyperlink{prb:MCCST}{MCCST}\xspace}
\newcommand{\MLCST}{\hyperlink{prb:MLCST}{MLCST}\xspace}
\newcommand{\pathTSP}{\hyperlink{prb:pathTSP}{Path TSP}\xspace}
\newcommand{\MSCJ}[1][T]{\hyperlink{prb:MSCJ}{MSCJ\textsubscript{$#1$}}\xspace}
\newcommand{\MBDST}{\hyperlink{prb:MBDST}{MBDST}\xspace}
\newcommand{\TSP}{\hyperlink{prb:TSP}{TSP}\xspace}
\newcommand{\ATSP}{\hyperlink{prb:ATSP}{ATSP}\xspace}
\newcommand{\bATSP}{\hyperlink{prb:bATSP}{Bottleneck ATSP}\xspace}

\usepackage{needspace}

\title{A New Dynamic Programming Approach for Spanning Trees with Chain Constraints and Beyond%
\thanks{%
Funded through the Swiss National Science Foundation grants 200021\_184622 and P500PT\_206742, the European Research Council (ERC) under the European Union's Horizon 2020 research and innovation programme (grant agreement No 817750), and the Deutsche Forschungsgemeinschaft (DFG, German Research Foundation) under Germany's Excellence Strategy~--~EXC~2047/1~--~390685813.

A short version of this work appeared in the proceedings of the 30th annual ACM-SIAM Symposium on Discrete Algorithms (SODA 2019)~\cite{nagele_2019_new}.
}%
}

\author{%
Martin N\"agele\thanks{%
Research Institute for Discrete Mathematics and Hausdorff Center for Mathematics, University of Bonn, Bonn, Germany.
Email: \href{mailto:mnaegele@uni-bonn.de}%
{mnaegele@uni-bonn.de}.
Most of this work was done while the author was employed at ETH Zurich.
}
\and
Rico Zenklusen\thanks{%
Department of Mathematics, ETH Zurich, Zurich, Switzerland.
Email: \href{mailto:ricoz@math.ethz.ch}%
{ricoz@math.ethz.ch}.
}
}

\date{}

\begin{document}

\maketitle

\begin{abstract}
Short spanning trees subject to additional constraints are important building blocks in various approximation algorithms, and, moreover, they capture interesting problem settings on their own.
Especially in the context of the Traveling Salesman Problem (TSP), new techniques for finding spanning trees with well-defined properties have been crucial in recent progress.
We consider the problem of finding a spanning tree subject to constraints on the edges in a family of cuts forming a laminar family of small width.
Our main contribution is a new dynamic programming approach where the value of a table entry does not only depend on the values of previous table entries, as it is usually the case, but also on a specific representative solution saved together with each table entry.
This allows for handling a broad range of constraint types.

In combination with other techniques---including negatively correlated rounding and a polyhedral approach that, in the problems we consider, allows for avoiding potential losses in the objective through the randomized rounding---we obtain several new results.
We first present a quasi-polynomial time algorithm for the Minimum Chain-Constrained Spanning Tree Problem with an essentially optimal guarantee.
More precisely, each chain constraint is violated by a factor of at most $1+\varepsilon$, and the cost is no larger than that of an optimal solution not violating any chain constraint.
The best previous procedure is a bicriteria approximation violating each chain constraint by up to a constant factor and losing another factor in the objective.
Moreover, our approach can naturally handle lower bounds on the chain constraints, and it can be extended to constraints on cuts forming a laminar family of constant width.

Furthermore, we show how our approach can also handle parity constraints (or, more precisely, a proxy thereof) as used in the context of (Path) TSP and one of its generalizations, and discuss implications in this context.
\end{abstract}

\begin{tikzpicture}[overlay, remember picture, shift = {(current page.south east)}]
\coordinate (anchor) at (0,0);
\node[anchor=south east, outer sep=5mm] at (anchor) {
\begin{tikzpicture}[outer sep=0] %
\node (ERC) {\includegraphics[height=13mm]{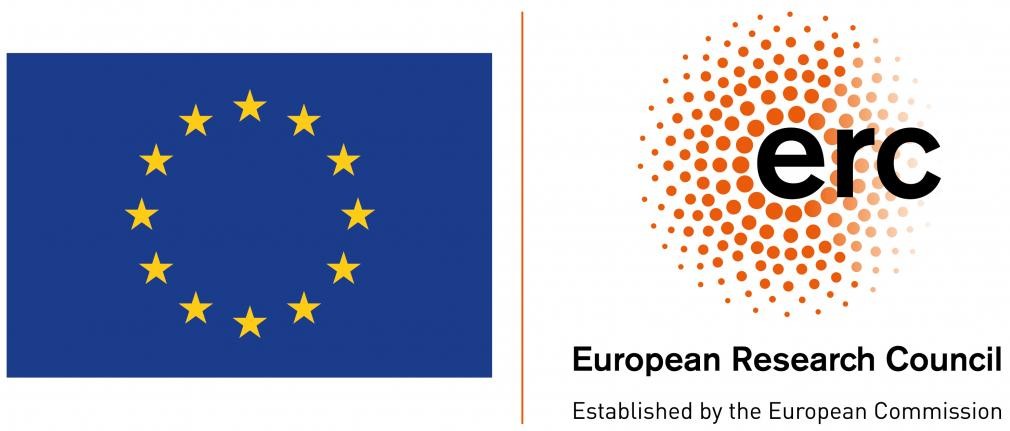}};
\node[left=5mm of ERC] (SNSF) {\includegraphics[height=7mm]{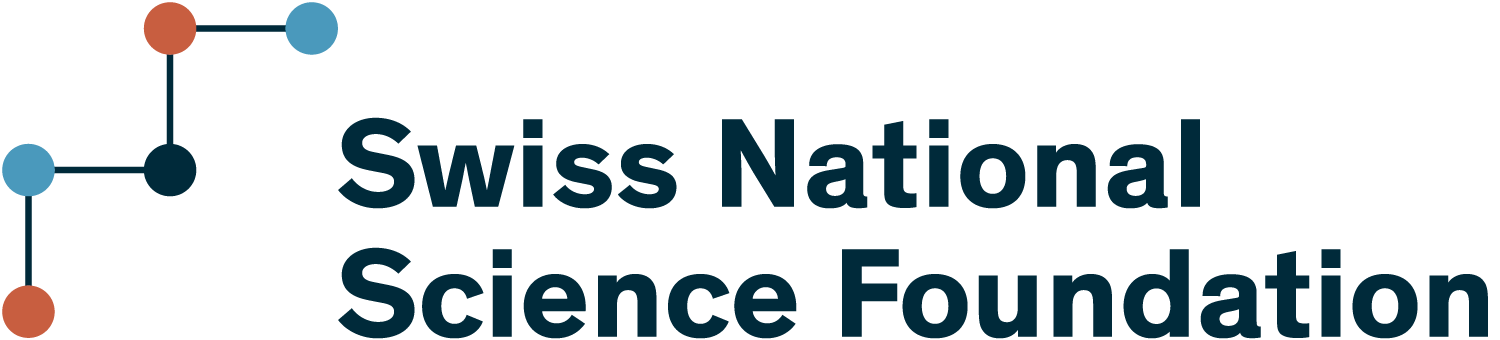}};
\node[right=5mm of ERC] (DFG) {\includegraphics[height=5mm]{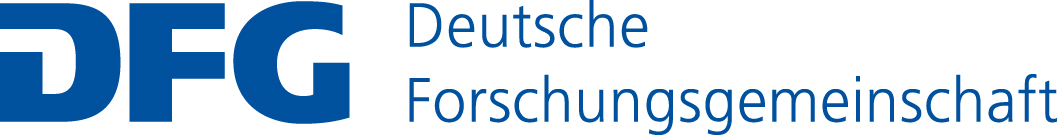}};%
\end{tikzpicture}
};
\end{tikzpicture}

\thispagestyle{empty}
\addtocounter{page}{-1}
\newpage

\section{Introduction}\label{sec:introduction}

Given a graph $G=(V,E)$ and edge costs $c\colon E\to\mathbb{R}_{\geqslant 0}$, the problem of finding a minimum cost spanning tree in $G$ with respect to $c$ is one of the most classical network design problems.
A variety of applications in areas like chip design, vehicle routing, and telecommunication networks triggered interest in constrained spanning tree problems.
Moreover, such problems are regularly used as building blocks in the design of approximation algorithms.
In particular, many approaches used in recent progress on the Traveling Salesman Problem (\labeltarget{prb:TSP}{\TSP}), where the goal is to find a shortest cycle in $G$ covering all vertices, and its path version have a constrained spanning tree problem as a key component.%
\footnote{
We remark that when referring to \TSP and its variants, we always assume that the involved edge lengths $c$ are metric.
}

The arguably most classical example of a constrained spanning tree problem is the minimum bounded degree spanning tree problem (\labeltarget{prb:MBDST}{\MBDST}).
Here, the goal is to find a spanning tree $T\subseteq E$ in $G$ of minimum cost subject to $T$ satisfying a degree constraint $|T\cap\delta(v)|\leqslant d(v)$ at every vertex $v$, where $d\colon V\to\mathbb{Z}_{>0}$ are given degree bounds.
Already just finding a feasible solution for \MBDST can easily be seen to be \NP-hard, even in the special case where $d(v)=2$ for all $v$, as this captures the Hamiltonian path problem.
This is typical for most constrained spanning tree problems.
The focus has therefore been on approximation algorithms that allow for a slight violation of the additional constraints.
This led to algorithms with various trade-offs between cost and constraint violation.
After a series of papers with progress on the approximation guarantees (see~\cite{konemann2000matter, konemann2003primaldual, goemans2006mbdst, chaudhuri2009pushrelabel, chaudhuri2009what} and references therein), an essentially best possible approximation algorithm for \MBDST was given by \textcite{singh2007approximating}.
Using iterative relaxation, they return a spanning tree violating each degree constraint by at most $1$ unit, and of cost no more than that of an optimal solution not violating the degree constraints.
\textcite{bansal2009additive} presented an elegant generalization of this result to upper bounds on the number of edges picked in a family of arbitrary edge sets $E_1,\ldots, E_k\subseteq E$.
More precisely, they show that a spanning tree violating each constraint by at most $\max_{e\in E}|\{i\in [k]\mid e\in E_i\}|-1$ and with cost no more than that of an optimal solution can be found.
If each edge is only contained in a constant number of constraints, this still leads to a constraint violation by only an additive constant.
Whereas iterative relaxation is undoubtedly a very strong tool to find constrained spanning trees, it is difficult to obtain constraint violations of at most a constant (either additively or multiplicatively) through this technique when edges can be in a super-constant number of constraints (see~\cite{zenklusen2012matroidal} for one rare example of this type).

However, constrained spanning tree problems appearing in the design of approximation algorithms, especially within problems related to \TSP, are often of this type.
For example, \textcite{asadpour_2017_atsp} established a beautiful connection between the asymmetric version of \TSP (\labeltarget{prb:ATSP}{\ATSP}), where the edge lengths may be asymmetric but are still assumed to satisfy the triangle inequality, and so-called thin trees, which are trees with constraints on all cut sets.
More precisely, if there is a constant $c$ such that for any $k\in \mathbb{Z}_{>0}$, one can efficiently find in any $k$-edge-connected graph $G=(V,E)$ a spanning tree $T\subseteq E$ with $|T\cap \delta(S)| \leqslant \sfrac{c}{k}\cdot |\delta(S)|$ for all $S\subseteq V$, then this can be transformed into an $\Oh(1)$-approximation for \ATSP.
Such trees are sometimes referred to as constantly-thin trees.
The existence of a weaker version of constantly-thin trees was conjectured by \textcite{goddyn_2004_open} and remains open.
We highlight that recently, \textcite{svensson_2020_constant-factor} obtained a $506$-approximation for \ATSP through different techniques, which has subsequently been improved to a $(22+\varepsilon)$-approximation by \textcite{traub_2022_improved}.
These algorithms are LP-based with respect to the Held-Karp relaxation, whose integrality gap is known to have a lower bound of $2$.
Finding constantly-thin spanning trees may be one path to advance on the approximability of \ATSP and the integrality gap of the Held-Karp relaxation, and they are also a natural path to obtain a first $O(1)$-approximation for \labeltarget{prb:bATSP}{\bATSP}, where the goal is find a Hamiltonian cycle where the edge of largest length is as small as possible.
(See~\cite{an_2021_approximation}, which also presents the currently best $O(\sfrac{\log |V|}{\log \log |V|})$-approximation for \bATSP.)
Moreover, especially for \pathTSP, where the task is to find a shortest Hamiltonian $s$-$t$ path in a complete graph with metric lengths, finding spanning trees with various additional constraints/properties has been crucial in recent progress~\cite{an_2015_improving,sebo_2013_eight-fifth,vygen_2016_reassembling,gottschalk_2018_better,sebo_2019_salesman,traub_2019_approaching,zenklusen2018tsp}.
Interestingly, the type of tree properties considered for \pathTSP are often on the edges contained in a family of $s$-$t$ cuts that form a chain.\footnote{Throughout this paper, a \emph{cut} of a vertex set $V$ is a non\-empty set $S\subsetneq V$.
An edge $e$ \emph{lies in} a cut $S$ if $e\in\delta(S)$.} More generally, the metric shortest connected $T$-join problem (\MSCJ), which generalizes both \pathTSP and classical \TSP, naturally leads to a laminar family of cuts to be considered~\cite{cheriyan_2015_approximating}.\footnote{For some even cardinality vertex set $T\subseteq V$ in a graph $G=(V,E)$, a $T$-join is an edge set $U\subseteq E$ such that the vertices of odd degree in the subgraph $(V,U)$ are precisely $T$.
Moreover, in \MSCJ, one is allowed to choose as $U$ a multiset of edges in $E$.}
The appearance of cut families with laminar or chain structure in this context stems from the use of combinatorial uncrossing arguments, which are ubiquitous in the context of \TSP, and is thus not surprising.
Clearly, when constraints are imposed on the edges in a family of cuts that are laminar, or even just a chain, then edges can appear in a large number of constraints.

The arguably most canonical constrained spanning tree problem with constraints on a laminar family of cuts is when there are upper bounds on the number of edges in each cut.
This setting was considered by \textcite{bansal2013generalizations}, who designed an iterative relaxation approach for returning a spanning tree violating each constraint by at most $\Oh(\log |V|)$ units and being of cost no more than the cost of an optimal solution not violating the constraints.
As later shown by \textcite{olver_2018_chain-constrained}, this is almost optimal because an additive violation of $\sfrac{c \log |V|}{\log\log |V|}$ units, for some constant $c$, cannot be achieved unless $\P=\NP$.
It remains open whether $\Oh(1)$-multiplicative violations are possible.

In summary, constrained spanning tree problems where edges can appear in a large number of constraints are still badly understood, and new approaches and techniques are needed.

\medskip

The goal of this paper is to introduce a versatile dynamic programming type approach to deal with a variety of constraint types on laminar cut families of small width, with applications to chain-constrained spanning trees, \pathTSP and beyond.
Dynamic programming did not play a crucial role in the above-mentioned problems until a very recent breakthrough result by \textcite{traub_2019_approaching} in the context of \pathTSP, which inspired this work, and later results~\cite{zenklusen2018tsp,traub_2021_reducing} in the context of \pathTSP and variants thereof.
A key new technical ingredient in our approach is to introduce a generalized form of dynamic programming, where the value of a table entry does not only depend on the values of previous table entries, as it is usually the case, but also on a fixed representative solution saved together with each table entry.
This leads to the peculiar situation that it is hard to define upfront the solution set over which our dynamic program optimizes.
However, we can show that it optimizes over a relaxation of the problems we are interested in, and returns solutions with well-defined properties to be exploited later on, which is all we need.
For chain-constrained problems, our dynamic program can be leveraged to return a fractional point in the spanning tree polytope, which can then be rounded to an actual spanning tree.
We show that good spanning trees can be obtained by using negatively correlated rounding procedures together with an alteration procedure that may be of independent interest, and which we therefore present in a more general context.

\subsection{Our results}

Here, we provide a summary of the results that we obtain by combining our dynamic programming approach with various other techniques.
We start with a natural special case of laminarly constrained spanning trees that has been studied previously, namely the \emph{minimum chain-constrained spanning tree} problem (\MCCST), where upper bounds are imposed on the number of edges that can be chosen in a family of cuts that form a chain.
Opposed to previous results, we also allow for lower bounds on the number of edges in the cuts, which can be handled with our methods without additional complications.

\begin{mdframed}[leftmargin=0.025\linewidth,rightmargin=0.025\linewidth]%
{\textbf{Minimum Chain-Constrained Spanning Tree Problem (\labeltarget{prb:MCCST}{\MCCST}):}}
Let $G=(V,E)$ be a graph with edge costs $c\colon E\to\mathbb{R}_{\geqslant 0}$, and let $\emptyset\subsetneq S_1\subsetneq S_2\subsetneq \ldots\subsetneq S_k\subsetneq V$ and $a_1,\ldots,a_k,b_1,\ldots,b_k\in\mathbb{Z}_{\geqslant0}$.
Find a spanning tree $T\subseteq E$ minimizing $c(T)\coloneqq\sum_{e\in T}c(e)$ among all trees satisfying
\vspace{-0.5em}
\begin{equation*}
a_i \leqslant |T\cap \delta(S_i)| \leqslant b_i \quad\text{for all $i\in[k]\coloneqq \{1,\ldots, k\}$.}
\end{equation*}
\end{mdframed}
For $\alpha,\beta \geqslant 1$, we say that an algorithm returning a spanning tree $T$ is an \emph{$(\alpha,\beta)$-approximation} for \MCCST if $\frac{1}{\beta} \cdot a_i \leqslant |T\cap \delta(S_i)| \leqslant \beta\cdot b_i$ for all $i\in [k]$, and $c(T)\leqslant \alpha\cdot c(\OPT)$, where $\OPT$ is a spanning tree of minimum cost among all spanning trees not violating the chain constraints.
For \MCCST without lower bounds, i.e., $a_1=\ldots=a_k=0$, \textcite{linhares_2018_reduction} recently presented an efficient $(\frac{\lambda}{\lambda-1},9\lambda)$-approximation for any $\lambda >1$ by extending a prior approach of \textcite{olver_2018_chain-constrained} that did not handle costs.
Our main result is a quasi-polynomial algorithm for \MCCST (with lower bounds) with essentially best possible guarantees.
\begin{theorem}\label{thm:MCCST}
For every $\varepsilon>0$, there is a randomized $(1,1+\varepsilon)$-approximation algorithm for \MCCST with running time $|V|^{\Oh(\sfrac{\log |V|}{\varepsilon^2})}$.
\end{theorem}
The approximation guarantee is essentially best possible in the sense that finding a tree that fulfills all chain constraints is \NP-hard as shown in~\cite{olver_2018_chain-constrained}.
Our randomized algorithm returns a $(1,1+\varepsilon)$-approximation with high probability, and can also be transformed into a Las Vegas algorithm.
Moreover, \cref{thm:MCCST} gives hopes that such best possible guarantees may be achievable with an efficient procedure.
We prove the theorem through a combination of our new dynamic programming approach, which we will introduce in this context, together with negatively correlated rounding and an alteration step to improve the value of the final solution.
As we discuss in \cref{sec:localCorrections}, in this context, the alteration step we use could also be replaced by a technique introduced by \textcite{linhares_2018_reduction}.

\medskip

It turns out that with some modifications, the technical insights outlined above are enough to push our results beyond pure chain constraints towards the more general problem of \emph{minimum laminarly constrained spanning trees} (\MLCST), which is defined as follows.

\begin{mdframed}[leftmargin=0.025\linewidth,rightmargin=0.025\linewidth]%
{\textbf{Minimum Laminarly Constrained Spanning Tree Problem (\labeltarget{prb:MLCST}{\MLCST}):}}
Let $G=(V,E)$ be a graph with edge costs $c\colon E\to\mathbb{R}_{\geqslant 0}$, let $\mathcal{L}\subseteq 2^V\setminus\{\emptyset,V\}$ be a laminar family, and $a_S,b_S\in\mathbb{Z}_{\geqslant0}$ for $S\in\mathcal{L}$.
Find a spanning tree $T\subseteq E$ minimizing $c(T)\coloneqq\sum_{e\in T}c(e)$ among all trees satisfying
\vspace{-0.5em}
\begin{equation*}
a_S \leqslant |T\cap \delta(S)| \leqslant b_S \quad\text{for all $S\in\mathcal{L}$.}
\end{equation*}
\end{mdframed}
Note that if $\mathcal{L}$ contains precisely all singletons, then the above problem setting reduces to the minimum bounded degree spanning tree problem (\MBDST) mentioned in the introduction.
From a structural point of view, constraint types in the special cases \MCCST and \MBDST are ``orthogonal'' in the sense that the role of $\mathcal{L}$ is taken by a chain in one case and by an antichain in the other.
Currently, there is no efficient approach covering both cases.
For a step towards \MLCST using our dynamic programming framework, we parametrize laminar families by their \emph{width}, which is the smallest integer $k$ such that the laminar family does not contain any $k+1$ disjoint sets.
We denote the width of a laminar family $\mathcal{L}$ by $\operatorname{width}(\mathcal{L}$.
Using this notion, we can generalize \cref{thm:MCCST} to obtain the following result.
\begin{theorem}\label{thm:MLCST}
For every $\varepsilon>0$, there is a randomized $(1,1+\varepsilon)$-approximation algorithm for \MLCST with running time $|V|^{\Oh(\sfrac{k\log |V|}{\varepsilon^2})}$, where $k=\operatorname{width}(\mathcal{L})$.
\end{theorem}
Observe that the running time in the above result depends exponentially on the width of the laminar family.
For width $k$ up to the order $\mathcal{O}(\log|V|)$, we thus still achieve quasi-polynomial running time.
Unfortunately, the exponential dependence on $k$ seems to be intrinsic to our approach.

\medskip

Our dynamic programming approach is very versatile in terms of constraint types that can be handled.
To highlight this fact, we show how it can be employed in the context of the Traveling Salesman Problem (TSP), where one is often interested in finding spanning trees with parity constraints on a laminar family of cuts.
More precisely we consider the \emph{metric shortest connected $T$-join problem} (\MSCJ), which is a generalization of \pathTSP defined as follows.

\begin{mdframed}[leftmargin=0.025\linewidth,rightmargin=0.025\linewidth]%
{\boldmath\textbf{Metric Shortest Connected $T$-Join Problem (\labeltarget{prb:MSCJ}{\MSCJ}):}}
Let $G=(V,E)$ be a complete graph with metric edge lengths $\ell\colon E\to\mathbb{R}_{\geqslant 0}$, and let $T\subseteq V$ be nonempty with $|T|$ even.
Find a $T$-join $J\subseteq E$ minimizing $\ell(J)\coloneqq\sum_{e\in J}\ell(e)$ among all $T$-joins $J$ such that $(V,J)$ is connected.
\end{mdframed}
We show how a slight adaptation of our DP approach leading to \cref{thm:MLCST} allows for finding a spanning tree that readily leads to a $(1.5+\varepsilon)$-approximation for \MSCJ when $|T|$ is constant.
\begin{theorem}\label{thm:TJoins}
For every $\varepsilon>0$, there is a $(1.5+\varepsilon)$-approximation algorithm for \MSCJ with running time $|V|^{\Oh(\sfrac{|T|}{\varepsilon})}$.
\end{theorem}
We remark that, similar to the way we adapt our technique to \MSCJ, both the algorithms by \textcite{traub_2019_approaching,zenklusen2018tsp}, introduced in the context of \pathTSP (giving $(1.5+\varepsilon)$- and $1.5$-approximation algorithms, respectively), can be generalized to \MSCJ at no loss in the approximation guarantee.
Hence, for constant $|T|$, all three approaches imply an efficient method improving on a prior $1.6$-approximation by \textcite{sebo_2013_eight-fifth} and a more recent $\tfrac{11}{7}\approx 1.571$-approximation by \textcite{traub_2020_improving}´.
Nevertheless, we expand on our approach here in order to highlight another quite direct implication of our new techniques and showcase how different constraint types can be handled, in the hope that this may be of interest for possible future applications.

\paragraph{Organization of the paper}
We start by introducing our techniques in the context of \MCCST.
\cref{sec:approachMCCST} provides a clear outline of what we want to achieve with our dynamic program, and why this implies \cref{thm:MCCST} together with negatively correlated rounding procedures and the solution alteration technique mentioned ealier.
\cref{sec:approachDPMCCST} then provides a thorough discussion of the key aspects of our dynamic programming technique.
\cref{sec:localCorrections} contains additional details on the local alteration approach that we use to obtain a unicriteria approximation for \MCCST, and shows a further application of this technique to turn bicriteria approximations into unicriteria ones.
In \cref{sec:extensionMLCST}, we discuss in detail why the natural generalization of our techniques to laminar constraint families fails, and how these difficulties can be overcome to obtain results for \MLCST and a proof of \cref{thm:MLCST}.
\cref{sec:approachPathTSP} shows how our technique can be used in the context of \TSP, in particular for \MSCJ, leading to \cref{thm:TJoins}.
\cref{sec:relaxationWeak} discusses why, for \MCCST, the natural LP relaxation is not strong enough to obtain results with guarantees as in \cref{thm:MCCST}, thus further motivating the use of a dynamic programming approach to strengthen the relaxation.
Finally, \cref{sec:exampleBacktracingOPT} presents an example showing that a classical analysis of our DP, namely by backtracing an optimal solution, is impossible in the laminarly constrained setting.%

\section[Overview of our approach for MCCST]{Overview of our approach for \MCCST}\label{sec:approachMCCST}

The first step of our approach for \MCCST relies on finding a solution to a suitable polyhedral relaxation.
The canonical relaxation, which was also used in prior results on chain-constrained trees~\cite{linhares_2018_reduction,olver_2018_chain-constrained}, enhances the spanning tree polytope $\PST$ with cut constraints.
We recall that $\PST$ is the convex hull of all characteristic vectors of spanning trees in $G=(V,E)$, and, by a seminal result of \textcite{edmonds_1971_matroids}, can be described by
\begin{equation*}%
\PST \coloneqq \left\{
x\in \mathbb{R}^E_{\geqslant 0} \,\middle|\, \begin{aligned}
x(E)    &= |V|-1\\
x(E[S]) &  \leqslant |S|-1 & \forall S\subsetneq V,\  |S|\geqslant 2
\end{aligned}
\right\}\enspace,
\end{equation*}
where $E[S]\subseteq E$ are all edges with both endpoints in $S$.
The polytope $Q\subseteq \mathbb{R}^E$ below describes the natural relaxation of \MCCST:
\begin{equation*}
Q \coloneqq \left\{ x\in \PST \,\middle|\, a_i \leqslant x(\delta(S_i)) \leqslant b_i\ \forall i\in [k] \right\}\enspace.
\end{equation*}

Unfortunately, solutions of the linear programming relaxation $\min\{c^\top x \mid x\in Q\}$ are too weak for our purposes.
In particular, there are instances where there exists a solution $y\in Q$ fulfilling the chain-constraints, even though any spanning tree must violate at least one chain constraint by a factor of at least $2$.
(We provide such an example in \cref{sec:relaxationWeak}.)
Hence, when comparing any integral solution to $y$, it will be impossible to stay within a factor of $1+\varepsilon$ regarding the violation of constraints---but this is precisely what we want to achieve.
This also shows a hard limit for prior approaches, which are all based on $Q$.

We therefore aim for a stronger relaxation.
It turns out that the reason why $Q$ can be a bad relaxation is the potential existence of small bounds $a_i, b_i$.
Indeed, assume that all $a_i,b_i$ for $i\in [k]$ were at least $c\cdot \log k$ for a sufficiently large constant $c$ (depending on $\varepsilon$).
Then one could first find an optimal solution $x^*$ to $\min\{c^\top x \mid x\in Q\}$, and then round $x^*$ to a spanning tree by using one of several negatively correlated rounding procedure (see~\cite{asadpour_2017_atsp,chekuri2010dependent}), which lead to Chernoff-type concentration bounds.
The theorem below summarizes a simplified form of the properties obtained by those procedures.\footnote{More generally, randomized rounding procedures with these properties can be obtained for any matroid base polytope and Chernoff-type concentration holds for any linear function with small non-negative coefficients (see~\cite{chekuri2010dependent}).}
\begin{theorem}[see~\cite{asadpour_2017_atsp,chekuri2010dependent}]\label{thm:rounding}
Let $y\in \PST$.
There exists an efficient randomized rounding scheme for rounding $y$ to a random spanning tree $T$ in $G$ such that the following holds.
\begin{enumerate}
\item\label{item:marginals} $\Pr[e\in T] = y_e$ for all $e\in E$.

\item\label{item:chernoff} For any $\lambda >0$ and $U\subseteq E$, we have
\begin{equation*}\label{eq:cernoffUpper}
\Pr\big[(1-\lambda)y(U) \leqslant |T\cap U| \leqslant (1+\lambda) y(U) \big] \geqslant 1-2e^{-\sfrac{y(U)\lambda^2}{3}}\enspace.
\end{equation*}
\end{enumerate}
\end{theorem}
Consider applying \cref{thm:rounding} to $x^*$ to obtain a spanning tree $T$.
By choosing $U=\delta(S_i)$ in the above theorem for any $i\in [k]$, one obtains that $|T\cap \delta(S_i)|$ is within a $(1\pm \varepsilon)$-factor of $x^*(\delta(S_i))$ with probability $1-k^{\Omega(1)}$ if $x^*(\delta(S_i)) \geqslant c \cdot \log k$.
Moreover, $x^*\in Q$ implies $a_i \leqslant x^*(\delta(S_i)) \leqslant b_i$ for $i\in [k]$.
Hence, a union bound over all chain constraints shows that $T$ is unlikely to violate any chain constraint by a large factor.

Motivated by this observation, we design a dynamic programming approach to find points $y\in Q$ of small cost that, for each $i\in [k]$, are either integral on the edges $\delta(S_i)$, or have a large value $y(\delta(S_i))$.
To formalize this idea, we introduce the notion of $\tau$-integral solutions.

\begin{definition}[$\tau$-integral]\label{def:tauIntegral}
For $\tau \in \mathbb{Z}_{\geqslant 0}$, we say that a point $y\in \mathbb{R}^E$ is \emph{$\tau$-integral} (with respect to the cuts $S_1,\ldots, S_k$) if for each $i\in [k]$, either
\begin{enumerate}
\item\label{item:smallCuts} $y(\delta(S_i))\leqslant \tau$ and $y$ is integral on the edges in $\delta(S_i)$, or
\item\label{item:largeCuts} $y(\delta(S_i))\geqslant \tau+1$.
\end{enumerate}
We call the cuts $S_i$ satisfying \cref{item:smallCuts,item:largeCuts} the $y$-\emph{small} and $y$-\emph{large} cuts, respectively.
\end{definition}

Clearly, every integral point is $\tau$-integral for any $\tau\in \mathbb{Z}_{\geqslant 0}$.
The key implication of our dynamic programming approach in the context of \MCCST is the following.
\begin{theorem}\label{thm:MCCSTdPguarantee}
For any $\tau\in \mathbb{Z}_{\geqslant 0}$, there is an algorithm that returns in $|V|^{\Oh(\tau)}$ time a $\tau$-integral point $y\in Q$ with $c^\top y \leqslant c(\OPT)$, where $\OPT$ is an optimal solution to \MCCST.
\end{theorem}
Not surprisingly, to obtain \cref{thm:MCCSTdPguarantee}, we want our dynamic program to guess edges in the cuts that will later be $y$-small.
However, this simple high-level plan comes with some important technical hurdles.
In particular, even if we knew the edges used in some cut $\delta(S_i)$, completing the two parts of the spanning tree on the left-hand side of the cut (on the vertices $S_i$) and on its right-hand side (on $V\setminus S_i$), respectively, are two highly dependent subproblems.
Interestingly, it is not easy to separate them into independent ones by guessing further structure, like the connectedness on each side, without creating \NP-hard subproblems.
We expand on these, and further issues, in \cref{sec:approachDPMCCST}, and show how one can address them.
A key difference between classical dynamic programs and our approach is that our propagation step requires a fractional solution of a previous subproblem, and not just a small fingerprint of previously obtained solutions.

The issue of small cuts is now resolved through \cref{thm:MCCSTdPguarantee} by setting $\tau=\Theta(k)$ and rounding a $\tau$-integral point $y\in Q$ using a randomized rounding procedure with the guarantees stated in \cref{thm:rounding}: Because $y$ is integral on $y$-small cuts, the rounding procedure will return a tree $T$ such that $\chi^T$ coincides with $y$ on all $y$-small cuts, because it is marginal-preserving (\cref{item:marginals} in \cref{thm:rounding}).

One last technical hurdle to overcome to obtain a $(1,1+\varepsilon)$-approximation for \MCCST is that the properties of a negatively correlated rounding procedure, as stated in \cref{thm:rounding}, are not enough to get a spanning tree that both
\begin{enumerate*}
\item violates chain constraints at most slightly, and
\item has cost no more than $c^\top y$.
\end{enumerate*}
Indeed, typical applications of such rounding procedures only lead to $(1+\varepsilon)$-approximations in terms of the objective (see~\cite{chekuri2010dependent,chekuri2009dependent_arxiv} for examples).
We show that this loss in the objective is avoidable in \MCCST, and other settings, by using a simple alteration step that modifies the obtained spanning tree by swapping one edge.
\begin{theorem}\label{thm:neighborhood}
Let $y\in\PST$ and $c\in\mathbb{R}^E$.
Let $T$ be a random spanning tree in $G=(V,E)$ drawn from a distribution satisfying $\Pr[e\in T]=y(e)$ for all $e\in E$.
Let $\overline{T}$ be a spanning tree minimizing $c(U)$ among all spanning trees $U$ whose symmetric difference $U\symdiff T \coloneqq (U\setminus T) \cup (T \setminus U)$ with $T$ satisfies $|U\symdiff T| \leqslant 2$ and such that $y(e)\in (0,1)$ for $e\in U\symdiff T$.
Then
\vspace{-0.3em}
\begin{equation*}
\Pr\big[c(\overline{T}) \leqslant c^\top y \big] \geqslant (|V|-1)^{-1}\enspace.
\vspace{-0.2em}
\end{equation*}
\end{theorem}
In \cref{sec:localCorrections}, we show that \cref{thm:neighborhood} holds even in a much more general context and has implications outside \MCCST.
For the specific setting of \MCCST, we observe in \cref{sec:localCorrections} that also a method introduced by \textcite{linhares_2018_reduction} can be adapted to avoid the $(1+\varepsilon)$-factor loss in the objective.

We can now put together the above ingredients to obtain our quasi-polynomial $(1,1+\varepsilon)$-approximation for \MCCST, stated as \cref{alg:MCCST} below.

\begin{algorithm2e}[H]
\caption{Quasi-polynomial $(1,1+\varepsilon)$-approximation for MCCST\strut}\label{alg:MCCST}

\begin{algostepsarabic}
 \item\label{algostep:gety} Let $\tau\coloneqq \lfloor\sfrac{96\ln(2|V|)}{\varepsilon^2}\rfloor$, and use \cref{thm:MCCSTdPguarantee} to find a $\tau$-integral point $y\in Q$ with $c^\top y \leqslant c(\OPT)$.\strut

 \item\label{algostep:round} Let $\ell\coloneqq\lceil 2|V|\ln|V|\rceil$, and randomly round $y$ with a rounding procedure as guaranteed by \cref{thm:rounding}, $\ell$ times independently, to obtain spanning trees $T_1,\ldots, T_{\ell}$.

\item\label{algostep:alterate} For each $j\in [\ell]$, find a minimum cost spanning tree $\overline{T}_j$ among all spanning trees $T$ with $|T\symdiff T_j| \leqslant 2$ and such that $y(e)\in (0,1)$ for all $e\in T \symdiff T_j$.

\item\label{algostep:return} Among all $\overline{T}_j$ for $j\in [\ell]$ with $\frac{a_i}{1+\varepsilon} \leqslant |\overline{T}_j \cap \delta(S_i)| \leqslant (1+\varepsilon) b_i$ for all $i\in [k]$, return one of smallest cost.

\end{algostepsarabic}
\end{algorithm2e}

We now show that the above results---in particular \cref{thm:MCCSTdPguarantee}, which follows from our dynamic program, and \cref{thm:neighborhood}---imply that \cref{alg:MCCST} is a quasi-polynomial $(1,1+\varepsilon)$-approximation for \MCCST.

\begin{proof}[Proof of \cref{thm:MCCST}]
We will show that with probability at least $1-\sfrac{1}{|V|}$, there is one spanning tree $\overline{T}_j$ among the trees $\overline{T}_1,\ldots, \overline{T}_\ell$ computed by \cref{alg:MCCST} that satisfies both
\begin{enumerate}
\item\label{item:goalPropTbj1} $\frac{1}{1+\varepsilon}\cdot y(\delta(S_i)) \leqslant |\overline{T}_j \cap \delta(S_i)| \leqslant (1+\varepsilon) \cdot y(\delta(S_i))$ for all $i\in [k]$, and
\item\label{item:goalPropTbj2}
$c(\overline{T}_j)\leqslant c^\top y$,
\end{enumerate}
which indeed implies that the returned solution is $(1,1+\varepsilon)$-approximate because $a_i\leqslant y(\delta(S_i)) \leqslant b_i$ for all $i\in [k]$ due to $y\in Q$, and $y$ satisfies $c^\top y \leqslant c(\OPT)$ as guaranteed by \cref{algostep:gety} of the algorithm.

We first analyze a single random spanning tree among the spanning trees $T_1,\ldots, T_\ell$ determined in \cref{algostep:round} of the algorithm.
We denote by $T$ such a spanning tree that was obtained by randomly rounding $y$ with a randomized rounding procedure as guaranteed by \cref{thm:rounding}.
Observe that because the rounding is marginal-preserving, $\chi^T$ coincides with $y$ on any edge $e\in E$ with $y(e)\in \{0,1\}$.
As $y$ is $\tau$-integral, all edges within $y$-small cuts are of this type and $T$ thus fulfills all chain constraints corresponding to $y$-small cuts.

Together with Chernoff-type concentration bounds guaranteed by \cref{thm:rounding}, applied with $\lambda = \sfrac{\varepsilon}{4}$ and using $y(\delta(S_i))\geqslant \tau + 1$ for $y$-large cuts, we have%
\begin{equation}\label{eq:oneChainTGood}
\Pr\left[
\left(1-\frac{\varepsilon}{4}\right)y(\delta(S_i)) \leqslant |T\cap\delta(S_i)| \leqslant \left(1+\frac{\varepsilon}{4}\right)y(\delta(S_i))
\right] \geqslant 1 - \frac{1}{2|V|^2} \quad \forall i\in [k]\enspace.
\end{equation}%
Now let $\overline{T}$ be a spanning tree of minimum cost among all spanning trees $U\subseteq E$ with $|U\symdiff T| \leqslant 2$ and $y(e)\in (0,1)$ for $e\in U\symdiff T$.
The cost of this spanning tree has the same distribution as the cost of the spanning trees $\overline{T}_1,\ldots, \overline{T}_\ell$ computed in \cref{algostep:alterate} of \cref{alg:MCCST}.
By \cref{thm:neighborhood}, we have
\begin{equation}\label{eq:costTbGood}
\Pr\left[
c(\overline{T}) \leqslant c^\top y
\right] \geqslant \frac{1}{|V|-1}\enspace.
\end{equation}
Using a union bound over the $k$ events described in~\eqref{eq:oneChainTGood} and the one described in~\eqref{eq:costTbGood}, we obtain that $T$ and $\overline{T}$ simultaneously fulfill
\begin{enumerate}[label=(\alph*)]
\item\label{item:tChainGood}
$\left(1-\frac{\varepsilon}{4}\right)y(\delta(S_i)) \leqslant |T\cap\delta(S_i)| \leqslant \left(1+\frac{\varepsilon}{4}\right)y(\delta(S_i))$ for all $i\in [k]$, and

\item\label{item:costTbGood}
$c(\overline{T}) \leqslant c^\top y$,
\end{enumerate}
with probability at least
\begin{equation*}
1 - \left( k \cdot \frac{1}{2|V|^2} + \left(1-\frac{1}{|V|-1}\right) \right) \geqslant \frac{1}{2|V|}\enspace,
\end{equation*}
where we used $k\leqslant |V|$ in the above inequality.
Next, we show that \cref{item:tChainGood} above implies
\begin{equation}\label{eq:tbChainGood}
\left(1-\frac{\varepsilon}{2}\right)y(\delta(S_i))
\leqslant |\overline{T}\cap\delta(S_i)|
\leqslant \left(1+\frac{\varepsilon}{2}\right)y(\delta(S_i))
\quad\forall i\in [k]\enspace,
\end{equation}%
which in turn implies
$\frac{1}{1+\varepsilon}\cdot y(\delta(S_i)) \leqslant |\overline{T}\cap\delta(S_i)| \leqslant \left(1+\varepsilon\right)\cdot y(\delta(S_i))$ for all $i\in [k],$
providing the property that we seek as highlighted in \cref{item:goalPropTbj1,item:goalPropTbj2} above.
To see that~\eqref{eq:tbChainGood} holds for any $i\in [k]$ that corresponds to a $y$-small cut, notice that for such $i$ we have $\overline{T}\cap \delta(S_i) = T\cap \delta(S_i)$, as $\overline{T}$ and $T$ only differ on edges on which $y$ has a fractional value, and, due to $\tau$-integrality of $y$, small cuts do not contain such edges.
Hence, consider $i\in [k]$ with $y(\delta(S_i))\geqslant \tau + 1$.
Because $|T\symdiff \overline{T}|\leqslant 2$, $\overline{T}$ is either the same as $T$ or obtained from $T$ by replacing one edge by a different one, so
\begin{equation*}
|T\cap \delta(S_i)| - 1  \leqslant
|\overline{T} \cap \delta(S_i)| \leqslant
|T\cap \delta(S_i)| + 1\enspace.
\end{equation*}
The relation~\eqref{eq:tbChainGood} for $y$-large cuts now follows from the inequality in \cref{item:tChainGood}:
\begin{align*}
|\overline{T}\cap \delta(S_i)| \geqslant \left(1-\frac{\varepsilon}{4}\right)\cdot y(\delta(S_i)) -1 &\geqslant \left(1-\frac{\varepsilon}{2}\right) \cdot y(\delta(S_i))\enspace,\quad\text{and}\\
|\overline{T}\cap \delta(S_i)| \leqslant \left(1+\frac{\varepsilon}{4}\right)\cdot y(\delta(S_i)) +1 &\leqslant \left(1+\frac{\varepsilon}{2}\right)\cdot y(\delta(S_i))\enspace,
\end{align*}%
where the second inequality in each of the two above lines follows from $y(\delta(S_i))\geqslant \tau+1 \geqslant \sfrac{96\ln(2|V|)}{\varepsilon^2}$, because $S_i$ is $y$-large.

In summary, the tree $\overline{T}$ satisfies the two desired properties highlighted in \cref{item:goalPropTbj1,item:goalPropTbj2} with probability at least $(2|V|)^{-1}$.
Because the algorithm computes $\ell=\lceil 2|V| \ln |V|\rceil$ independent random trees $\overline{T}_1,\ldots, \overline{T}_j$ with the same distribution as $\overline{T}$, the probability that at least one of them fulfills the properties in~\cref{item:goalPropTbj1,item:goalPropTbj2} is at least
\begin{equation*}
1 - \left(1-\frac{1}{2|V|} \right)^{\ell} \geqslant 1 - e^{-\frac{\ell}{2|V|}} \geqslant 1 - \frac{1}{|V|}\enspace,
\end{equation*}
as desired.
Finally, the running time is dominated by the quasi-polynomial time dynamic programming approach used to find a cheap $\tau$-integral point $y\in Q$ in \cref{algostep:gety} of \cref{alg:MCCST}.
(All other steps of the algorithm can be performed efficiently.)
By \cref{thm:MCCSTdPguarantee}, we thus get a running time bound $|V|^{\Oh(\tau)} = |V|^{\Oh(\sfrac{\log |V|}{\varepsilon^2})}$.
\end{proof}

\section[The dynamic programming approach for MCCST]{The dynamic programming approach for \MCCST}\label{sec:approachDPMCCST}

First observe that to prove \cref{thm:MCCSTdPguarantee}, it suffices to consider $\tau\leqslant |V|-1$, because any $\tau$-integral point in $y\in Q$ for $\tau \geqslant |V|-1$ is integral as $Q\subseteq \PST$, and the $y$-value on any cut is at most $|V|-1$ for any point in $\PST$.
Thus, any $\tau \geqslant |V|$ can be replaced by $\tau = |V|-1$, so we assume $\tau\leqslant |V|-1$ in what follows.

Our dynamic program to find a cheap $\tau$-integral point in $Q$ is inspired by recent dynamic programming approaches in the context of \pathTSP~\cite{traub_2019_approaching,zenklusen2018tsp}, but faces important new technical challenges that require novel conceptual insights.
To highlight this point, let us first consider the significantly simpler special case of $\tau=1$.
The dynamic programming approaches for \pathTSP are essentially algorithms for this case.\footnote{More precisely, dynamic programs for \pathTSP are looking for points in the Held-Karp relaxation of \pathTSP instead of the spanning tree polytope, but this is only a minor technical difference without significant impact on the dynamic program.}

\subsection[Brief overview to find cheap one-integral solution following prior techniques]{Brief overview to find cheap \boldmath$1$\unboldmath-integral solution following prior techniques}

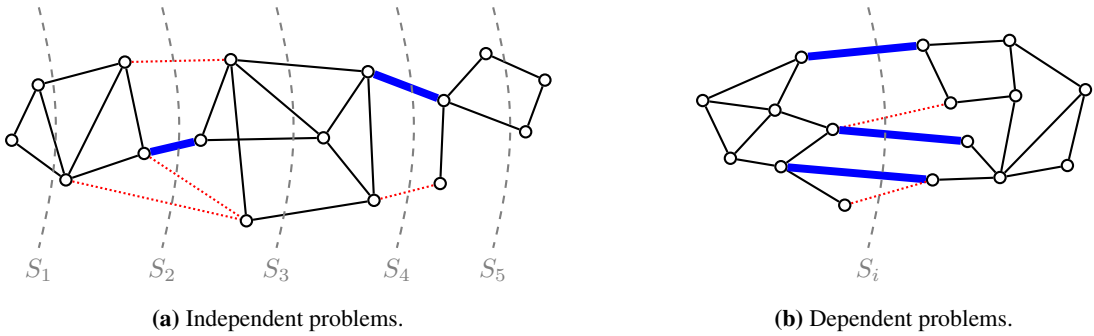
\begin{figure}[b] %
\begin{subfigure}[b]{.49\textwidth}
\centering
\begin{tikzpicture}[scale=0.8, font=\small]

\begin{scope}[every node/.style={ns}]
\node  (1) at (2.10,-4.71) {};
\node  (2) at (2.54,-3.79) {};
\node  (3) at (3.00,-5.37) {};
\node  (4) at (3.98,-3.41) {};
\node  (5) at (4.30,-4.93) {};
\node  (6) at (5.74,-3.37) {};
\node  (7) at (5.24,-4.71) {};
\node  (8) at (6.00,-6.05) {};
\node  (9) at (8.02,-3.57) {};
\node (10) at (7.28,-4.67) {};
\node (11) at (8.12,-5.71) {};
\node (12) at (9.28,-4.05) {};
\node (13) at (9.22,-5.43) {};
\node (14) at (9.98,-3.27) {};
\node (15) at (10.96,-3.71) {};
\node (16) at (10.64,-4.57) {};
\end{scope}

\begin{scope}[thick]
\draw  (1) --  (2);
\draw  (1) --  (3);
\draw  (2) --  (3);
\draw  (2) --  (4);
\draw  (3) --  (4);
\draw  (3) --  (5);
\draw[guessedOut]  (3) --  (8);
\draw  (4) --  (5);
\draw[guessedOut]  (4) --  (6);
\draw[guessedIn]  (5) --  (7);
\draw[guessedOut]  (5) --  (8);
\draw  (6) --  (8);
\draw  (6) --  (9);
\draw  (6) -- (10);
\draw  (7) --  (6);
\draw  (7) -- (10);
\draw  (8) -- (11);
\draw  (9) -- (10);
\draw  (9) -- (11);
\draw[guessedIn]  (9) -- (12);
\draw (10) -- (11);
\draw[guessedOut] (11) -- (13);
\draw (12) -- (13);
\draw (12) -- (14);
\draw (12) -- (16);
\draw (14) -- (15);
\draw (15) -- (16);
\end{scope}

\begin{scope}[shift={(2,-2.5)},thick,cutcol,dashed,bend left=14pt, every node/.style={below}]
\draw (0.55,0) to ++(0,-4) node {$S_1$};
\draw (2.6,0) to ++(0,-4) node {$S_2$};
\draw (4.5,0) to ++(0,-4) node {$S_3$};
\draw (6.5,0) to ++(0,-4) node {$S_4$};
\draw (8.1,0) to ++(0,-4) node {$S_5$};
\end{scope}

\end{tikzpicture}
 \caption{Independent problems.}\label{fig:oneEdgeIndepDecomp}
\label{fig:indepProblems}
\end{subfigure}
\begin{subfigure}[b]{.49\textwidth}
\centering
\begin{tikzpicture}[scale=0.8,font=\small]

\begin{scope}[every node/.style={ns}]
\node  (1) at (2.55,-1.85) {};
\node  (2) at (4.19,-1.13) {};
\node  (3) at (3.75,-2.01) {};
\node  (4) at (3.01,-2.81) {};
\node  (5) at (4.71,-2.33) {};
\node  (6) at (3.85,-2.95) {};
\node  (7) at (4.91,-3.59) {};
\node  (8) at (6.21,-0.93) {};
\node  (9) at (6.67,-1.89) {};
\node (10) at (6.95,-2.53) {};
\node (11) at (6.37,-3.17) {};
\node (12) at (7.65,-0.85) {};
\node (13) at (7.75,-1.77) {};
\node (14) at (7.49,-3.13) {};
\node (15) at (8.61,-2.93) {};
\node (16) at (8.91,-1.67) {};
\end{scope}

\begin{scope}[thick]
\draw  (1) --  (2);
\draw  (1) --  (3);
\draw  (1) --  (4);
\draw  (2) --  (3);
\draw[guessedIn]  (2) --  (8);
\draw  (3) --  (4);
\draw  (3) --  (5);
\draw  (4) --  (6);
\draw  (5) --  (6);
\draw[guessedOut]  (5) --  (9);
\draw[guessedIn]  (5) -- (10);
\draw  (6) --  (7);
\draw[guessedIn]  (6) -- (11);
\draw[guessedOut]  (7) -- (11);
\draw  (8) --  (9);
\draw  (8) -- (12);
\draw  (9) -- (13);
\draw (10) -- (14);
\draw (11) -- (14);
\draw (12) -- (13);
\draw (12) -- (16);
\draw (13) -- (14);
\draw (14) -- (15);
\draw (14) -- (16);
\draw (15) -- (16);
\end{scope}

\begin{scope}[shift={(5.3,-0.3)},thick,cutcol,dashed,bend left=14pt, every node/.style={below}]
\draw (0,0) to ++(0,-4) node {$S_i$};
\end{scope}

\end{tikzpicture}
 \caption{Dependent problems.}\label{dependendDecomp}
\label{fig:depProblems}
\end{subfigure}
\caption{\textbf{(a)} Here, the small cuts are $S_{i_1}=S_2$ and $S_{i_2}=S_4$, and we assume that only a single edge, drawn in thick and blue, crosses each small cut.
Finding a $\tau$-integral point $y\in \PST$ boils down to solving independent subproblems in $G[S_{2}]$, $G[S_{4}\setminus S_{2}]$, and $G[V\setminus S_{4}]$.
\textbf{(b)} If more than one edge is in the small cut $S_i$, the problem does not decompose into independent subproblems on the left-hand side and right-hand side of $S_i$.}
\end{figure}

To gain intuition for this special case, which nicely allows for showcasing later on the added difficulty faced for general $\tau$, assume that we knew upfront the small cuts with respect to an optimal solution $\OPT\subseteq E$, i.e., the cuts among $S_1,\ldots, S_k$ in which $\OPT$ contains a single edge.
Let $S_{i_1}, \ldots, S_{i_\ell}$ for $1\leqslant i_1 < \ldots < i_\ell \leqslant k$ be these small cuts.
For notational convenience, we set $S_{i_0}\coloneqq \emptyset$ and $S_{i_{\ell+1}}\coloneqq V$.
Now consider the $\ell+1$ induced subgraphs $G[S_{i_1}]$, $G[S_{i_2}\setminus S_{i_1}]$,\ldots, $G[S_{i_\ell}\setminus S_{i_{\ell-1}}]$, $G[V\setminus S_{i_\ell}]$.\footnote{For $W\subseteq V$, $G[W]$ is the subgraph of $G$ induced by $W$.}
It is not hard to observe that the edges of $\OPT$ within each of these subgraphs must form a spanning tree in that subgraph.
Moreover, for $j\in [\ell]$, the single edge $e_j\in \OPT \cap \delta(S_{i_{j}})$ must go from $S_{i_j}\setminus S_{i_{j-1}}$ to $S_{i_{j+1}}\setminus S_{i_j}$ for $\OPT$ to be a spanning tree (see \cref{fig:indepProblems}).
If, moreover, we even knew the single edge $e_j\in \delta(S_{i_j})\cap \OPT$ for each $j\in [\ell]$, then the problem of finding a \emph{corresponding} cheapest $1$-integral point $y\in \PST$---i.e., with $y$-small cuts $S_{i_1}, \ldots, S_{i_\ell}$ and edges $e_1,\ldots, e_\ell$ contained in them---decomposes into $\ell+1$ independent linear programs, one within each of the above-mentioned induced subgraphs.
More precisely, one has to find, for $j\in [\ell+1]$, a cheapest point $y^j$ in the spanning tree polytope of $G[S_{i_{j}}\setminus S_{i_{j-1}}]$ with lower bounds on each cut $S_i$ with $S_{i_j}\subsetneq S_i \subsetneq S_{i_{j-1}}$ to make sure that $y^j$, together with the guessed edges in small cuts, has a load $y^j(\delta(S_i))$ of at least $2$ on these cuts.

The above observations now naturally lead to a dynamic programming approach that extends solutions from left to right, i.e., a $1$-integral solution in some subgraph $G[S_i]$ for some $i$ is extended to one on $G[S_j]$ for $j>i$.
This way, one can use a dynamic program to optimize over all possibilities of small cuts and edges contained in them (see~\cite{traub_2019_approaching,zenklusen2018tsp} for more details of this approach in the context of \pathTSP).

\subsection[Toward general tau with connectivity patterns and resulting challenges]{Toward general \boldmath$\tau$\unboldmath\ with connectivity patterns and resulting challenges}

However, if $\tau\geqslant 2$, i.e., if there are two or more edges in small cuts, splitting the problem into independent ones along small cuts comes with significant additional challenges linked to obtaining connectivity and acyclicity globally from independent solutions of the subproblems (see \cref{fig:depProblems}).%

\begin{figure}[ht]
\begin{center}
\begin{tikzpicture}[scale=0.8, font=\small]

\begin{scope}[every node/.style={ns}]
\node  (1) at (1.37,-1.61) {};
\node  (2) at (1.41,-2.95) {};
\node  (3) at (2.23,-2.27) {};
\node  (4) at (3.11,-0.77) {};
\node  (5) at (3.23,-2.05) {};
\node  (6) at (3.77,-3.23) {};
\node  (7) at (3.47,-4.13) {};
\node  (8) at (4.65,-0.63) {};
\node  (9) at (4.81,-1.67) {};
\node (10) at (5.05,-2.55) {};
\node (11) at (5.01,-3.33) {};
\node (12) at (5.05,-4.23) {};
\node (13) at (6.65,-0.75) {};
\node (14) at (6.75,-1.53) {};
\node (15) at (7.37,-2.27) {};
\node (16) at (6.85,-3.19) {};
\node (17) at (6.89,-4.05) {};
\node (18) at (6.91,-4.57) {};
\node (19) at (8.31,-0.85) {};
\node (20) at (8.57,-1.77) {};
\node (21) at (8.27,-2.59) {};
\node (22) at (8.85,-3.47) {};
\node (23) at (9.01,-4.69) {};
\node (24) at (9.69,-1.43) {};
\node (25) at (9.89,-2.57) {};
\node (26) at (9.65,-3.61) {};
\end{scope}

\begin{scope}
\node at (8)[above] {$u_1$};
\node at (13)[above] {$v_1$};
\node at (9)[above left=-2pt] {$u_2$};
\node at (14)[above] {$v_2$};
\node at (10)[above] {$u_3\!=\!u_4$};
\node at (15)[above] {$v_3$};
\node at (16)[above]{$v_4$};
\node at (11)[above]{$u_5$};
\node at (12)[below]{$u_6$};
\node at (17)[above,xshift=2mm]{$v_5\!=\!v_6$};
\end{scope}

\begin{scope}[thick]
\draw  (1) --  (2);
\draw  (1) --  (3);
\draw  (1) --  (4);
\draw  (2) --  (3);
\draw  (2) --  (7);
\draw  (3) --  (6);
\draw  (4) --  (5);
\draw  (4) --  (8);
\draw  (5) --  (9);
\draw  (5) -- (10);
\draw  (6) -- (10);
\draw  (7) -- (11);
\draw  (7) -- (12);
\draw  (8) --  (9);
\draw[guessedIn]  (8) -- (13);
\draw[guessedIn]  (9) -- (14);
\draw[guessedOut]  (9) -- (15);
\draw[guessedIn] (10) -- (15);
\draw[guessedIn] (10) -- (16);
\draw[guessedOut] (10) -- (21);
\draw[guessedIn] (11) -- (17);
\draw[guessedIn] (12) -- (17);
\draw[guessedOut] (12) -- (18);
\draw (13) -- (19);
\draw (13) -- (20);
\draw (14) -- (20);
\draw (15) -- (21);
\draw (16) -- (21);
\draw (16) -- (22);
\draw (17) -- (23);
\draw (18) -- (23);
\draw (19) -- (20);
\draw (19) -- (24);
\draw (20) -- (21);
\draw (20) -- (25);
\draw (22) -- (23);
\draw (22) -- (25);
\draw (23) -- (26);
\draw (24) -- (25);
\draw (25) -- (26);
\end{scope}

\begin{pgfonlayer}{bg}
\def\d{0.2}
\begin{scope}[densely dashed,draw=darkgreen, rounded corners, thick,fill=darkgreen!20]
\filldraw ($(13)+(-1.5*\d,3*\d)$) coordinate (tl1) rectangle  ($(14)+1.3*(\d,-0.8*\d)$) coordinate (br1);

\filldraw ($(15)+(-1.5*\d,2.8*\d)$) coordinate (tl2) rectangle ($(15)+1.3*(\d,-0.8*\d)$) coordinate (br2);

\filldraw ($(16)+(-1.9*\d,2.8*\d)$) coordinate (tl3) rectangle ($(17)+1.3*(3.7*\d,-0.8*\d)$) coordinate (br3);
\end{scope}
\end{pgfonlayer}

\begin{scope}[darkgreen]
\node at ($(tl1 -| br1)+(0.35,-0.3)$) {$C_1$};
\node at ($(tl2 -| br2)+(0.35,-0.3)$) {$C_2$};
\node at ($(tl3 -| br3)+(0.35,-0.4)$) {$C_3$};
\end{scope}

\begin{scope}[shift={(5.6,-0.25)},thick,cutcol,dashed,bend left=14pt, every node/.style={below}]
\draw (0,0) to ++(0,-4.5) node {$S_i$};
\end{scope}

\begin{scope}
\node[blue] at ($(13)+(-1.3,-0.4)$) {$F$};
\end{scope}

\end{tikzpicture}
 \end{center}
 \caption{A connectivity triple $(S_i,F,\mathcal{C})$ with connectivity pattern $\mathcal{C}=\{C_1,C_2,C_3\}$.}
\label{fig:components}
\end{figure}
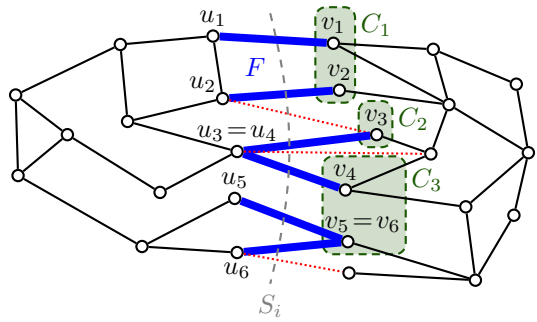
One natural approach to try to address such challenges is to maintain more structure in the dynamic program, by for example also enumerating over potential \emph{connectivity patterns} of edges in small cuts, i.e., ways of how the edges in a small cut could be connected on one or either side of the cut.
For ease of presentation, consider a situation with a single small cut $S_i$, and a given selection of $t\leqslant \tau$ many edges $F \coloneqq\{\{u_j,v_j\}\mid j\in[t]\}\subseteq\delta(S_i)$ with $u_j\in S_i$ and $v_j\notin S_i$ for all $j\in[t]$.
A \emph{connectivity pattern} for the right-hand side of the cut is a partition $\mathcal{C}$ of $\{v_1,\ldots,v_t\}$, where a set $C\in\mathcal{C}$ indicates that the vertices in $C$ shall be connected in $G[V\setminus S_i]$ .
We call the triple $(S_i,F,\mathcal{C})$ a \emph{connectivity triple} (see \cref{fig:components}).

\begin{definition}[Compatibility with a connectivity triple]\leavevmode
\begin{enumerate}
\item A spanning tree $T\subseteq E$ is \emph{compatible} with the connectivity triple~$(S_i,F,\mathcal{C})$ if $T\cap \delta(S_i) = F$ and the partition on $\{v_1,\ldots, v_t\}$ induced by the connected components of $T\cap E[V\setminus S_i]$ equals $\mathcal{C}$.
\item A set $R\subseteq \binom{V\setminus S_i}{2}$ is \emph{right-compatible} with $(S_i,F,\mathcal{C})$ if $R$ is a forest and the partition on $\{v_1,\ldots, v_t\}$ induced by the connected components of $R$ equals $\mathcal{C}$.
\item Let $U\subseteq E[S_i]$, and let $R\subseteq \binom{V\setminus S_i}{2}$ be right-compatible with $(S_i,F,\mathcal{C})$.
Then $U$ is \emph{left-compatible} with $(S_i,F,\mathcal{C})$ if $U\cup F\cup R$ is a spanning tree.
\item Let $x\in \mathbb{R}^E$ with $\supp(x)\subseteq E[S_i]$, and let $R\subseteq \binom{V \setminus S_i}{2}$ be right-compatible with $(S_i,F,\mathcal{C})$.
Then $x$ is \emph{left-compatible} with $(S_i,F,\mathcal{C})$ if $x + \chi^F + \chi^R$ is in the spanning tree polytope of $(V,E\cup R)$.
\end{enumerate}
\end{definition}
We highlight that right-compatible sets $R$ are not required to be a subset of the edges of $G$, but can contain any pairs of vertices within $V\setminus S_i$.
This makes sure that right-compatible sets exist for any connectivity triple, which simplifies the exposition.
Moreover, one can observe that the above definitions of left-compatibility do not depend on which right-compatible set $R$ is chosen, and are thus well-defined.

Knowing the correct connectivity triple, the desired separation into independent subproblems can actually be achieved.
However, this comes at the cost that the subproblem on the side where we guessed the connectivity pattern becomes substantially harder than in the simple case $\tau=1$.
This is nicely highlighted by a simple connectivity pattern $\mathcal{C}$: Assume that all right endpoints $\{v_1,\ldots, v_t\}$ of $F$ are distinct, $t$ is even, and let $\mathcal{C}=\{\{v_1,v_2\}, \{v_3,v_4\}, \ldots, \{v_{t-1}, v_t\}\}$ be a grouping of the endpoints into pairs.
To pinpoint the difficulties, consider the question of whether there exists a right-compatible edge set $U\subseteq E[V\setminus S_i]$.
For this to be the case, $U$ must be a forest with $\sfrac{t}{2}$ components, one for each pair $v_j,v_{j+1}$ that connects that pair.
Such a set $U$ exists if and only if there are $\sfrac{t}2$ vertex-disjoint paths in $G[V\setminus S_i]$, one between $v_j$ and $v_{j+1}$ for each $j\in \{1,3,\ldots, t-1\}$.
Hence, just determining whether there exists a right-compatible solution is at least as difficult as the vertex-disjoint paths problem.
Though this problem is efficiently solvable for a constant number of paths $\sfrac{t}{2}$, due to the seminal results by Robertson and Seymour in the context of the Graph Minor Project (also see~\cite{kawarabayashi_2012_disjoint} for a faster procedure), the known techniques for disjoint paths are highly non-trivial, non-polyhedral, cannot handle costs, and, last but not least, the connectivity pattern $\mathcal{C}$ leading to the disjoint paths problem remains a very special case of connectivity patterns we have to deal with.

\subsection{Efficiently extending subsolutions through relaxed connectivity requirements}

To overcome this issue in our approach, we still enumerate over connectivity triples with connectivity patterns on the right-hand side of small cuts as described above, but will later not require that right-hand side solutions are right-compatible with it; they only have to properly complete an existing left-hand side solution.
The left-hand side solutions, however, will be left-compatible with the guessed connectivity triples.

We start by observing a simple way to describe left-compatibility.
Let $(S_i,F,\mathcal{C})$ be a connectivity triple.
Then, left-compatible edge sets $U\subseteq E[S_i]$ are simply spanning trees in an auxiliary graph $G(S_i, F, \mathcal{C})$ which we obtain from $G[S_i]$ through the following operations:
\begin{enumerate*}
\item add the edges $F$ and their endpoints to $G[S_i]$,
\item contract the vertex sets in $\mathcal{C}$, and
\item contract the edges in $F$.
\end{enumerate*}
There is a canonical one-to-one relation between the edges in $G(S_i,F,\mathcal{C})$ and $E[S_i]$, and we therefore treat them as the same edge set.
By $\PST(S_i,F,\mathcal{C})$, we denote the spanning tree polytope of $G(S_i,F,\mathcal{C})$.
Hence, $\PST(S_i,F,\mathcal{C})$ are all points that are left-compatible with $(S_i,F,\mathcal{C})$.

We are now ready to describe our dynamic programming approach.
For ease of notation, we set $S_0\coloneqq\emptyset$ and $S_{k+1}\coloneqq V$.
Consider the set $\mathcal{K}$ of all connectivity triples $(S_i,F,\mathcal{C})$, where $i\in \{0,\ldots, k+1\}$ and the crossing edges $F\subseteq \delta(S_i)$ satisfy $a_i\leqslant |F| \leqslant \min\{\tau,b_i\}$.
In particular, there are no connectivity triples for $i\in\{0,\ldots,k+1\}$ for which $\tau < a_i$.
For each $(S_i,F,\mathcal{C})\in \mathcal{K}$ we determine via our dynamic program a point (partial solution) $y_{(S_i,F,\mathcal{C})}\in \mathbb{R}^{E[S_i]}$ with the following property.
\begin{property}\label{prop:yPartialProps}\leavevmode
\begin{enumerate}
\item\label{item:yIsInPST} $y_{(S_i, F, \mathcal{C})}\in \PST(S_i,F,\mathcal{C})$.

\item\label{item:yIsAlphInt} $y_{(S_i,F,\mathcal{C})}+\chi^F$ is $\tau$-integral on $S_1,\ldots, S_{i-1}$.

\item\label{item:yFulfillsBounds} $a_h \leqslant y_{(S_i,F,\mathcal{C})}(\delta(S_h))+|F\cap\delta(S_h)|\leqslant b_h$ for all $h\in [i-1]$.

\item\label{item:yIsCheap} $c^\top y_{(S_i,F,\mathcal{C})}$ is at most the cost of a cheapest edge set $U\subseteq E[S_i]$ such that $\chi^U$ fulfills \cref{item:yIsInPST,item:yIsAlphInt,item:yFulfillsBounds}.
\end{enumerate}
\end{property}
For simplicity of notation, we consider all vectors defined on some subset $U\subseteq E$ of the edges, like $y_{(S_i,F,\mathcal{C})}$ above, to be vectors in $\mathbb{R}^E$, where all entries on which the vector was not defined are set to~$0$.\footnote{This makes sure that expressions like $c^\top y_{(S_i,F,\mathcal{C})}$ or $y_{(S_i,F,\mathcal{C})}(\delta(S_h))$ are well-defined.}
Clearly, if we can obtain such points in $|V|^{\Oh(\tau)}$ time, then we are done because $y\coloneqq y_{(V,\emptyset,\{\emptyset\})}$ is in $Q$ because of \cref{item:yIsInPST,item:yFulfillsBounds}; $y$ is a $\tau$-integral point due to \cref{item:yIsAlphInt}; and $y$ satisfies $c^\top y \leqslant c(\OPT)$ due to \cref{item:yIsCheap}, as desired.

To construct the vectors $y_{(S_i,F,\mathcal{C})}$, we initialize $y_{(S_0,\emptyset,\{\emptyset\})}$ to the zero vector, and consider triples $(S_i, F, \mathcal{C})\in \mathcal{K}$ in increasing order of $i$.
Let $i\in [k+1]$, let $(S_i,\overline{F},\overline{\mathcal{C}})\in \mathcal{K}$, and assume that for each $(S_j, F, \mathcal{C})\in \mathcal{K}$ with $j<i$, we already computed a vector $y_{(S_j,F,\mathcal{C})}$ satisfying \cref{prop:yPartialProps}.
To compute a vector $y_{(S_i,\overline{F},\overline{\mathcal{C}})}$ satisfying \cref{prop:yPartialProps}, we consider all $(S_j, F, \mathcal{C})\in \mathcal{K}$ with $j<i$ and for each such triple, we solve the following linear program, which finds a cheapest extension $z$ of $y_{(S_j,F,\mathcal{C})}$ that is left-compatible with $(S_i,\overline{F},\overline{\mathcal{C}})$:
\begin{equation}\tag{exLP}\label{eq:lpCompDP}
\begin{aligned}
&& \min\quad c^\top z \\
&& z &\in\PST(S_i,\overline{F},\overline{\mathcal{C}})\\
&& \max\{\tau+1,a_h\}\leqslant z(\delta(S_h))&+|\overline{F}\cap\delta(S_h)| \leqslant b_h && \forall h\in \{j+1,\ldots, i-1\}\\
&& z(e) &= y_{(S_j,F,\mathcal{C})}(e) && \forall e \in E[S_j]\\
&& z(e) &= 1 && \forall e\in F\\
&& z(e) &= 0 && \forall e\in \delta(S_j)\setminus F\enspace.
\end{aligned}
\end{equation}%
Among all linear programs of type~\eqref{eq:lpCompDP}, i.e., one for each triple $(S_j, F, \mathcal{C})$ with $j<i$, we determine the one achieving the smallest optimal value and set $y_{(S_j,\overline{F},\overline{\mathcal{C}})}$ to be an optimal solution of that linear program.%
\footnote{%
We use the usual convention that if some linear program~\eqref{eq:lpCompDP} is infeasible, then its objective value is interpreted as $\infty$ (and we will never use a solution to an infeasible linear program later on).
Infeasibility occurs, for example, for choices of $(S_j, F, \mathcal{C})$ where $F$ does not contain all edges of $\overline{F}\cap\delta(S_j)$.}
This finishes the description of our dynamic program.
The bottleneck of the running time is the repeated solving of linear programs of type~\eqref{eq:lpCompDP}.
A simple bound on the number of such LPs that we solve is $|\mathcal{K}|^2$, and the running time of $|V|^{\Oh(\tau)}$ then follows from the following bound and the fact that we can solve~\eqref{eq:lpCompDP} in strongly polynomial time through standard techniques.
A more formal treatment, including a proof of the simple statement below, is given in \cref{sec:DPdetails}.

\begin{proposition}\label{prop:sizeK}
$|\mathcal{K}| = |V|^{\Oh(\tau)}$.
\end{proposition}

We now highlight a few key aspects of our dynamic program.
First, even though~\eqref{eq:lpCompDP} seeks to complete a prior solution $y_{(S_j,F,\mathcal{C})}$ to one for the triple $(S_i,\overline{F},\overline{C})$, we do not require the completion to be right-compatible with $(S_j,F,\mathcal{C})$.\footnote{Notice that we did not formally define right-compatibility for fractional points because we do not need it, but a natural extension would be to say that it is a convex combination of right-compatible integral solutions.}
This connectivity pattern is completely disregarded in~\eqref{eq:lpCompDP} and was only used on the left-hand side of $S_j$ to construct $y_{(S_j,F,\mathcal{C})}$.
The key observation is that any integral solution that is compatible with both triples $(S_j,F,\mathcal{C})$ and $(S_i,\overline{F},\overline{\mathcal{C}})$, and has only large cuts between $S_i$ and $S_j$, provides a legal way to complete $y_{(S_j,F,\mathcal{C})}$ as formally described by the following statement.
\begin{lemma}\label{lem:completionOk}
Consider $i>j$ such that $b_h>\tau$ for all $h \in \{j+1,\ldots,i-1\}$, and let $R\subseteq \binom{V\setminus S_i}{2}$ be right-compatible with $(S_i,\overline{F},\overline{\mathcal{C}})$.
Then for any $U\subseteq E[S_i]$ such that $T\coloneqq U\cup \overline{F} \cup R$ is a spanning tree compatible with both $(S_j,F,\mathcal{C})$ and $(S_i,\overline{F},\overline{\mathcal{C}})$, and $|T\cap \delta(S_h)|\in [\max\{\tau+1,a_h\},b_h]$ for all $h\in \{j+1,\ldots, i-1\}$,
the following vector is a feasible solution to~\eqref{eq:lpCompDP}:
\begin{equation*}
z\coloneqq y_{(S_j,F,\mathcal{C})} + \chi^{F\setminus \overline{F}} + \chi^{U\cap E[S_i\setminus S_j]}\enspace.
\end{equation*}
\end{lemma}
The above lemma is crucial to make sure that our dynamic program remains a relaxation of the original \MCCST problem, even though it requires to complement a specific previously computed solution $y_{(S_j,F,\mathcal{C})}$.
Simply speaking, \cref{lem:completionOk} guarantees that the increment in cost when extending the solution $y_{(S_j,F,\mathcal{C})}$ up to the cut $S_i$ through our dynamic program is no more than the best integral extension that realizes both connectivity triples at $S_j$ and $S_i$.
This makes sure that the solutions we compute %
fulfill \cref{item:yIsCheap} of \cref{prop:yPartialProps}.
Intuitively, the reason why \cref{lem:completionOk} holds is the following: No matter what precise point $y_{(S_j,F,\mathcal{C})}$ we computed, as long as it is left-compatible with $(S_j,F,\mathcal{C})$, it can be completed by any edge set that is right-compatible with $(S_j,F,\mathcal{C})$ to obtain a point in $\PST$.

Also note that a crucial difference between our dynamic program and the classical way of using dynamic programming approaches is that we need an \emph{explicit} solution $y_{(S_j,F,\mathcal{C})}$ in our propagation/extension step.
Only knowing the connectivity triple $(S_j,F,\mathcal{C})\in \mathcal{K}$ and the value of a best point $y_{(S_j,F,\mathcal{C})}$ for that triple would not be enough.
To highlight this contrast, consider for example a classical dynamic programming approach for the (integer) knapsack problem~(see, e.g.,~\cite[Section 17.2]{korte_2018_combinatorial}).
Here, the dynamic program computes for every possible cost the smallest total weight of items realizing that cost.
(Sometimes, the dynamic program is presented with the roles of costs and weights exchanged, in which case, for every possible weight, the dynamic program computes a minimum cost solution of that weight.) In the propagation step of this classical dynamic program, to extend existing partial solutions of smaller cost to partial solutions of larger cost, one only needs to know the weight and cost of previously computed partial solutions, but not the exact partial solution.
Partial solutions are sometimes saved in classical dynamic programs to quickly retrieve a solution through backtracking, once the whole dynamic programming table is filled.
However, to just determine the optimal value of a solution, and in particular to perform the propagation steps in the dynamic program, partial solutions are not used.

\subsection[Details of the dynamic programming approach for MCCST]{Details of the dynamic programming approach for \MCCST}\label{sec:DPdetails}

In this section, we complete the proofs that are missing to formally ensure that the dynamic programming approach described above achieves the guarantees claimed by \cref{thm:MCCSTdPguarantee}.
Recall that the dynamic program initializes $y_{(S_0,\emptyset,{\emptyset})}=0$, and propagates to points $y_{(S_i,F,\mathcal{C})}\in\mathbb{R}^E_{\geqslant 0}$ for all $(S_i,F,\mathcal{C})\in\mathcal{K}$ using \cref{alg:propagationChainDP}.

\begin{algorithm2e}[H]
\caption{Propagation to $y_{(S_i,\overline{F},\overline{\mathcal{C}})}$ from all $y_{(S_j,F,\mathcal{C})}$ with $j<i$.\strut}\label{alg:propagationChainDP}

\begin{algostepsarabic}
\item\label{algitem:propagationStepLamExLP} For all $(S_j,F,\mathcal{C})\in\mathcal{K}$\strut{} with $j<i$, solve the linear program~\eqref{eq:lpCompDP} and obtain an optimal solution $y$.
\item Among all solutions $y$ found in \cref{algitem:propagationStepLamExLP}, let
$y_{(S_i,\overline{F},\overline{\mathcal{C}})}$ be one minimizing $c^\top y$.
Return $y_{(S_i,\overline{F},\overline{\mathcal{C}})}$.
\end{algostepsarabic}
\end{algorithm2e}

We start by proving \cref{lem:completionOk}, which essentially shows that the extension found by our dynamic program when solving a linear program of the form~\eqref{eq:lpCompDP} has cost no more than the best integral extension that is compatible with the connectivity triples on both sides of the extension.

\begin{proof}[Proof of \cref{lem:completionOk}]
To obtain feasibility of $z$ for~\eqref{eq:lpCompDP}, it is easy to see that the inequality constraints that are stated explicitly in~\eqref{eq:lpCompDP} follow immediately by definition of $z$.
Thus, it remains to prove that $z\in\PST(S_i,\overline{F},\overline{\mathcal{C}})$, i.e., that $z$ is left-compatible with $(S_i,\overline{F},\overline{\mathcal{C}})$.

By assumption, the tree $T$ is compatible with $(S_j,F,\mathcal{C})$, hence $T\cap \binom{V\setminus S_j}{2}$ is right-compatible with $(S_j,F,\mathcal{C})$.
Moreover, $y_{(S_j,F,\mathcal{C})}$ is left-compatible with $(S_j,F,\mathcal{C})$.
Combining these observations, we see that $y_{(S_j,F,\mathcal{C})} + \chi^{F} + \chi^{T\cap \binom{V\setminus S_j}{2}}$ is in the spanning tree polytope of $(V,E\cup R)$.
By partitioning $T\cap \binom{V\setminus S_j}{2}$ into $T\cap E[S_i\setminus S_j]$, $\overline{F}\setminus F$, and $R$, we get that, equivalently, $y_{(S_j,F,\mathcal{C})} + \chi^{F} + \chi^{T\cap E[S_i\setminus S_j]} + \chi^{\overline{F}\setminus F} + \chi^R$ is in the spanning tree polytope of $(V,E\cup R)$, which implies that $y_{(S_j,F,\mathcal{C})} + \chi^{F\setminus\overline{F}} + \chi^{T\cap E[S_i\setminus S_j]}$ is left-compatible with $(S_i,\overline{F},\overline{\mathcal{C}})$.
As $T\cap E[S_i\setminus S_j]=U\cap E[S_i\setminus S_j]$, the lemma follows.
\end{proof}

With the bound on the incremental cost of an extension from \cref{lem:completionOk}, we can show that propagating along \cref{alg:propagationChainDP}, we maintain \cref{prop:yPartialProps}.

\begin{lemma}\label{lem:propagationMCCST}
Let $i\in[k+1]$ and $(S_i,\overline{F},\overline{\mathcal{C}})\in\mathcal{K}$.
Assume that for all $(S_j,F,\mathcal{C})\in\mathcal{K}$ with $j<i$, we are given points $y_{(S_j,F,\mathcal{C})}$ satisfying \cref{prop:yPartialProps}.
Then $y_{(S_i,\overline{F},\overline{\mathcal{C}})}$ obtained by \cref{alg:propagationChainDP} also has \cref{prop:yPartialProps}.
\end{lemma}

\begin{proof}
Note that any solution of~\eqref{eq:lpCompDP} satisfies \cref{item:yIsInPST,item:yIsAlphInt,item:yFulfillsBounds} of \cref{prop:yPartialProps} with respect to $(S_i,\overline{F},\overline{\mathcal{C}})$.
\cref{item:yIsInPST} follows directly from the corresponding constraint in~\eqref{eq:lpCompDP}, while \cref{item:yIsAlphInt,item:yFulfillsBounds} are implied by the constraints in the linear program for cuts $S_h$ with $h\geqslant j$, and follow for cuts $S_h$ with $h<j$ from the fact that $y_{(S_j,F,\mathcal{C})}$ has \cref{prop:yPartialProps}, by assumption.

To see that \cref{item:yIsCheap} holds, it is enough to see that for any edge set $U\subseteq E[S_i]$ fulfilling \cref{item:yIsInPST,item:yIsAlphInt,item:yFulfillsBounds}, there is a connectivity triple $(S_j,F,\mathcal{C})$ with $j<i$ such that the corresponding solution $y$ of the linear program~\eqref{eq:lpCompDP} satisfies $c^\top y\leqslant c(U)$.
To see this, fix such an edge set $U$, which is by definition left-compatible with $(S_i,\overline{F},\overline{\mathcal{C}})$.
Thus, for any $R\subseteq\binom{V\setminus S_i}{2}$ that is right-compatible with $(S_i,\overline{F},\overline{\mathcal{C}})$, the edge set $T\coloneqq U\cup\overline{F}\cup R$ is a spanning tree in $(V,E\cup R)$ compatible with $(S_i,\overline{F},\overline{\mathcal{C}})$.
Let $j<i$ be maximal such that $S_j$ is a $\chi^T$-small cut.
Let $F=T\cap \delta(S_j)$ and let $\mathcal{C}$ be the connectivity pattern such that $T$ is compatible with $(S_j,F,\mathcal{C})$.

We claim that $(S_j,F,\mathcal{C})$ is the connectivity triple that we are looking for, i.e., if $y$ is an optimal solution of~\eqref{eq:lpCompDP} when extending $y_{(S_j,F,\mathcal{C})}$ to $(S_i,\overline{F},\overline{\mathcal{C}})$, then $c^\top y\leqslant c(U)$.
Thereto, observe that by definition, $T=U\cup\overline{F}\cup R$ is a spanning tree in $(V,E\cup R)$ compatible with both $(S_j,F,C)$ and $(S_i,\overline{F},\overline{\mathcal{C}})$, and by the choice of $j$, we also have $|T\cap\delta(S_h)|\in[\max\{\tau+1,a_h\},b_h]$ for all $h\in\{j+1,\ldots,i-1\}$.
Thus, $R$ and $U$ satisfy the assumptions of \cref{lem:completionOk}, and we get that $z\coloneqq y_{(S_j,F,\mathcal{C})}+\chi^{F\setminus\overline{F}}+\chi^{U\cap E[S_i\setminus S_j]}$ is a feasible solution to~\eqref{eq:lpCompDP}.
But $y$ is an optimal solution of the same linear program, hence
\begin{equation*}
c^\top y \leqslant c^\top y_{(S_j,F,\mathcal{C})}+c(\chi^{F\setminus\overline{F}})+c(U\cap E[S_i\setminus S_j])\enspace.
\end{equation*}
By assumption, $y_{(S_j,F,\mathcal{C})}$ satisfies \cref{prop:yPartialProps}, and applying \cref{item:yIsCheap} with the edge set $U\cap E[S_j]$, we get $c^\top y_{(S_j,F,\mathcal{C})}\leqslant c(U\cap E[S_j])$.
Combining this with the above, we obtain the desired inequality
\begin{equation*}
c^\top y \leqslant c(U\cap E[S_j])+c(\chi^{F\setminus\overline{F}})+c(U\cap E[S_i\setminus S_j])=c(U)
\enspace .\qedhere
\end{equation*}
\end{proof}

The final step before wrapping up and proving \cref{thm:MCCSTdPguarantee} is to show the bound on $|\mathcal{K}|$ from \cref{prop:sizeK}.

\begin{proof}[Proof of \cref{prop:sizeK}.]%
The set $\mathcal{K}$ consists of all connectivity triples $(S_i,F,\mathcal{C})$, where $i\in\{0,1,\ldots,k+1\}$, $F\subseteq\delta(S_i)$ has size at most $\tau$, and $\mathcal{C}$ is a corresponding connectivity pattern.
In order to build such a triple, there are at most $|V|+1$ options for choosing a set $S_i$ (one for each $i\in\{0,\ldots,k+1\}$, and we have $k\leqslant |V|-1$).
Furthermore, note that $\delta(S_i)$ can contain at most $\Oh(|V|^2)$ many edges of $G$, hence there are at most $|V|^{\Oh(\tau)}$ many choices for a subset of size at most $\tau$.
Finally, $\mathcal{C}$ is a partition of the at most $\tau$ many endpoints of the edges in $F$ that do not lie inside $S_i$, and the number of such partitions can be bounded by $|V|^{\Oh(\tau)}$ as well, where we use $\tau\leqslant |V|-1$.
Altogether, we get $|\mathcal{K}|\leqslant |V|^{\Oh(\tau)}$.
\end{proof}%

Finally, we are ready to prove \cref{thm:MCCSTdPguarantee}.

\begin{proof}[Proof of \cref{thm:MCCSTdPguarantee}]%
We calculate points $y_{(S_i,F,\mathcal{C})}$ for all $(S_i,F,\mathcal{C})\in\mathcal{K}$ by initializing $y_{(\emptyset,\emptyset,\{\emptyset\})}=0$ and calculating $y_{(S_i,F,\mathcal{C})}$ in increasing order of $i$ using the propagation step described in \cref{alg:propagationChainDP}.
Note that $y_{(\emptyset,\emptyset,\{\emptyset\})}$ has \cref{prop:yPartialProps}, and hence, by an inductive application of \cref{lem:propagationMCCST}, we obtain that all points $y_{(S_i,F,\mathcal{C})}$ for $(S_i,F,\mathcal{C})\in\mathcal{K}$ have \cref{prop:yPartialProps}.
In particular, \cref{prop:yPartialProps} for $y_{(V,\emptyset,\{\emptyset\})}$ immediately implies that this point has the properties claimed by \cref{thm:MCCSTdPguarantee}.
Note that the guarantee on $c^\top y_{(V,\emptyset,\{\emptyset\})}$ follows from the fact that a cheapest edge set $U$ such that $\chi^U$ fulfills \cref{item:yIsInPST,item:yIsAlphInt,item:yFulfillsBounds} of \cref{prop:yPartialProps} with respect to $(S_i,F,\mathcal{C})=(V,\emptyset,\{\emptyset\})$ is in fact an optimal solution.

For the running time bound, observe that the dominating operation of our dynamic programming procedure is repeatedly solving linear programs of type~\eqref{eq:lpCompDP}.
The total number of such linear programs that we solve is bounded from above by $|\mathcal{K}|^{2}$, and therefore, the running time of $|V|^{\Oh(\tau)}$ follows from \cref{prop:sizeK} and the observation that linear programs of the type~\eqref{eq:lpCompDP} can be solved in strongly polynomial time.
The latter can be achieved by using a compact extended formulation for the spanning tree polytope with small coefficients in the constraint matrix (one can, for example, use the one by \textcite{martin1991using}, which has coefficients that are bounded by $1$ in absolute value), and then applying the framework of \textcite{tardos_1986_strongly}.
\end{proof}%

\section[Local correction steps for rounding procedures in {0,1}-polytopes]{Local correction steps for rounding procedures in \boldmath$\{0,1\}$\unboldmath-polytopes}\label{sec:localCorrections}

\newcommand{\Fstar}{\overline{F}}

In this section, we discuss details of the proof of \cref{thm:neighborhood}, used to avoid a $(1+\varepsilon)$-factor loss in the objective value.
We present this result separately because it may be of independent interest, as it applies to a broad class of problem settings.
At the end of this section, in \cref{sec:linharesSwamy}, we briefly discuss how, for \MCCST, an alternative approach introduced by \textcite{linhares_2018_reduction} also allows for avoiding a loss in the objective value.

\smallskip

We show \cref{thm:neighborhood} by proving a more general statement for $\{0,1\}$-polytopes, based on polyhedral neighborhoods.
We therefore start with some basic polyhedral terminology.\footnote{We refer the interested reader to~\cite[Volume A]{schrijver2003combinatorial} for more information on polyhedral combinatorics.}
$\{0,1\}$-polytopes are a representation of \emph{set systems} $(E,\mathcal{F})$, where $E$ is a finite ground set and $\mathcal{F}\subseteq 2^E$.
One can think of $\mathcal{F}$ as the feasible sets of some combinatorial problem over $E$.
For example, $E$ may be the edge set of a graph and $\mathcal{F}$ the family of all spanning trees.
The \emph{combinatorial polytope} $P_{\mathcal{F}}$ that corresponds to $\mathcal{F}$ is defined by
\begin{equation}\label{eq:combPoly}
P_{\mathcal{F}} \coloneqq \conv(\{\chi^F \mid F\in \mathcal{F}\})\enspace,
\end{equation}
where $\conv$ denotes the convex hull.
Hence, if $\mathcal{F}$ are all spanning trees of a graph, $P_{\mathcal{F}}$ is the spanning tree polytope.
For $F_1,F_2\in \mathcal{F}$ and $q\in \mathbb{Z}_{\geqslant 0}$, we say that $F_2$ is in the $q$-neighborhood of $F_1$ on $P_{\mathcal{F}}$ if one can reach the vertex $\chi^{F_2}$ of $P_{\mathcal{F}}$ from the vertex $\chi^{F_1}$ by successively traversing at most $q$ edges of $P_{\mathcal{F}}$.
Notice that this natural notion extends the way we modify $T$ to obtain $\overline{T}$ in \cref{thm:neighborhood}: Indeed, this follows from the well-known property that two spanning trees $T,\overline{T}$ in $G$---or, more generally, any two bases $T,\overline{T}$ of a matroid---have the property that $\chi^T, \chi^{\overline{T}}$ are adjacent in $P_{\mathcal{F}}$ if and only if $|T\symdiff \overline{T}| = 2$ (see, e.g.,~\cite[Volume B]{schrijver2003combinatorial}).
Furthermore, for any $y\in P_{\mathcal{F}}$, we denote by $P_{\mathcal{F}_y}\subseteq P_{\mathcal{F}}$ the minimal face of $P_{\mathcal{F}}$ that contains $y$.
Additionally, $\mathcal{F}_y\subseteq \mathcal{F}$ denotes all sets $F\in \mathcal{F}$ such that $\chi^F\in P_{\mathcal{F}_y}$.
(Note that these definitions of $\mathcal{F}_y$ and $P_{\mathcal{F}_y}$ are consistent with~\eqref{eq:combPoly} in the sense that $P_{\mathcal{F}_y}$ is indeed the combinatorial polytope of the family $\mathcal{F}_y$.)

A key quantity in our derivations is the cardinality $\rho(\mathcal{F})$ of a largest size set in $\mathcal{F}$:
\begin{equation*}
\rho(\mathcal{F}) \coloneqq \max\{|F| \mid F \in \mathcal{F}\}\enspace.
\end{equation*}
Typically, when dealing with a set system $(E,\mathcal{F})$ where all sets have the same cardinality, i.e., $|F|=\rho(\mathcal{F})$ for all $F\in \mathcal{F}$, we can obtain slightly stronger results later.
We call such set systems \emph{equal-cardinality} systems.
Note that the family of spanning trees (or bases of any matroid) is an equal-cardinality system.
We prove the following generalization of \cref{thm:neighborhood}.
\begin{theorem}\label{thm:neighborhoodGen}
Let $(E,\mathcal{F})$ be a set system, let $y\in P_{\mathcal{F}}$, $c\in \mathbb{R}^E$, $q\in \mathbb{Z}_{\geqslant 1}$, and let $T$ be a random set in $\mathcal{F}$ drawn from a distribution satisfying $\Pr[e\in T] = y_e$ for all $e\in E$.
Let $\overline{T}\in \mathcal{F}$ be a set minimizing $c(U)$ among all $U\in \mathcal{F}$ in the $q$-neighborhood of $T$ on $P_{\mathcal{F}_y}$.
Then
\begin{equation*}
\Pr\left[ c(\overline{T}) \leqslant c^\top y \right] \geqslant \frac{q}{2\rho({\mathcal{F}})}\enspace.
\end{equation*}
Moreover, if $(E,\mathcal{F})$ is an equal-cardinality system, then
\begin{equation*}
\Pr\left[ c(\overline{T}) \leqslant c^\top y \right] \geqslant \frac{q}{\rho({\mathcal{F}})}\enspace.
\end{equation*}
\end{theorem}

First observe that \cref{thm:neighborhoodGen} indeed implies \cref{thm:neighborhood}.

\begin{proof}[Proof of \cref{thm:neighborhood}]%
We set $E$ to be the edges of $G=(V,E)$, the family $\mathcal{F}\subseteq 2^E$ to be all spanning trees in $G$, and $q=1$.
Clearly, in this case we have $\rho(\mathcal{F})=|V|-1$ and $|F|=\rho(\mathcal{F})$ for all $F\in \mathcal{F}$, because every spanning tree has precisely $|V|-1$ edges.
Hence, spanning trees form an equal-cardinality system.
By \cref{thm:neighborhoodGen} we thus obtain
\begin{equation*}
\Pr[c(W) \leqslant c^\top y] \geqslant \frac{1}{|V|-1}\enspace,
\end{equation*}
where $W\in \mathcal{F}$ is a set minimizing $c(U)$ among all $U\subseteq \mathcal{F}$ in the $1$-neighborhood of $T$ on $P_{\mathcal{F}_y}$.

For the above to imply \cref{thm:neighborhood}, it suffices to show that any $U\subseteq \mathcal{F}$ in the $1$-neighborhood of $T$ on $P_{\mathcal{F}_y}$ fulfills $y(e)\in (0,1)$ for all $e\in U\Delta T$.
Because $P_{\mathcal{F}}\subseteq [0,1]^E$---i.e., non-negativity constraints and constraints of type $x(e)\leqslant 1$ are valid for $P_{\mathcal{F}}$---all points on $P_{\mathcal{F}_y}$ coincide with $y$ on the edges where $y$ is integral, i.e.,
\begin{equation*}
P_{\mathcal{F}_y} \subseteq \left\{ x\in [0,1]^E \;\middle\vert\; x(e)=y(e) \;\forall e\in E \text{ with } y(e)\in \{0,1\} \right\}\enspace.
\end{equation*}%
This implies that any $F\subseteq \mathcal{F}_y$ fulfills $y(e)\in (0,1)$ for all $e\in F\symdiff T$.
Hence, this also holds for any $U\subseteq \mathcal{F}$ in the $1$-neighborhood of $T$ on $P_{\mathcal{F}_y}$, as desired, and finishes the proof.
\end{proof}%

\subsection[Proof of \texorpdfstring{\cref{thm:neighborhoodGen}}{Theorem 13}]{Proof of \cref{thm:neighborhoodGen}}

We show \cref{thm:neighborhoodGen} in several steps.
We first derive a bound on the cost of a well-chosen set $A$ in the $1$-neighborhood of another fixed set $F$.
The following lemma formalizes this statement.
The cost improvement is measured with respect to the cost of some target set $\Fstar$, which will later be chosen to be a set in $\mathcal{F}$ of smallest cost.

\begin{lemma}\label{lem:thereIsGoodNeighbor}
Let $(E,\mathcal{F})$ be a set system.
Let $c\in \mathbb{R}^E$, and $F,\Fstar\in \mathcal{F}$ with $F\neq \Fstar$.
Then there exists a neighbor $A\in \mathcal{F}$ of $F$ on $P_{\mathcal{F}}$ satisfying
\begin{enumerate}
\item\label{item:AIsCheap} $c(A) - c(\Fstar) \leqslant \left(1-\frac{1}{|\Fstar\symdiff F|}\right) \cdot (c(F) - c(\Fstar))$, and
\item\label{item:AReducesSymDiff} $|\Fstar\symdiff A| \leqslant |\Fstar\symdiff F| -1$.
\end{enumerate}
Moreover, if $(E,\mathcal{F})$ is an equal-cardinality system, then the above properties can be strengthened to
\begin{enumerate}[label=(\roman*')]
\item\label{item:AIsCheapPrime} $c(A) - c(\Fstar) \leqslant \left(1-\frac{2}{|\Fstar\symdiff F|}\right) \cdot (c(F) - c(\Fstar))$, and
\item\label{item:AReducesSymDiffPrime} $|\Fstar\symdiff A| \leqslant |\Fstar\symdiff F| -2$.
\end{enumerate}
\end{lemma}
\begin{proof}
Consider the vertex $\chi^F$ of $P_{\mathcal{F}}$, and the family $\{A_1,\ldots, A_\ell\}\in \mathcal{F}$ of all neighboring sets in $\mathcal{F}$ on $P_{\mathcal{F}}$.
We start with a basic polyhedral property, namely that the cone with apex $\chi^F$ spanned by all edges of $P_{\mathcal{F}}$ incident with $\chi^F$ contains the whole polytope, i.e.,
\begin{equation*}
P_{\mathcal{F}} \subseteq \chi^F + \cone\left(\left\{\chi^{A_i} - \chi^F \;\middle\vert\; i\in [\ell]\right\}\right)\enspace.
\end{equation*}
In particular, this implies that there exist coefficients $\lambda_i\geqslant 0$ for $i\in [\ell]$ such that
\begin{equation*}
\chi^{\Fstar} = \chi^F + \sum_{i=1}^{\ell} \lambda_i\cdot \left(\chi^{A_i} - \chi^F \right)\enspace.
\end{equation*}
We are only interested in strictly positive coefficients.
Let $k$ be the number of strictly positive coefficients, and assume, by renumbering the indices, that these are the coefficients $\lambda_1,\ldots,\lambda_k$.
Hence,
\begin{equation}\label{eq:FstarinConeAtF}
\chi^{\Fstar} = \chi^F + \sum_{i=1}^{k} \lambda_i\cdot \left(\chi^{A_i} - \chi^F \right)\enspace,
\end{equation}
and $\lambda_i > 0$ for all $i\in [k]$.
Let $\lambda\coloneqq \sum_{i=1}^k \lambda_i$.

\medskip
\begin{adjustwidth}{1em}{}
{\textbf{Claim.}}
We have
$\lambda \leqslant |\Fstar\symdiff F|$.
Moreover, if $(E,\mathcal{F})$ is an equal-cardinality system, then $\lambda \leqslant \frac{1}{2}|\Fstar\symdiff F|$.

\begin{proof}[Proof of claim]%
We have
\begin{align}
|\Fstar\Delta F| &= \|\chi^{\Fstar} - \chi^F \|_1 \notag\\
   &=  \left\| \sum_{i=1}^k \lambda_i \cdot \left(\chi^{A_i} - \chi^F\right) \right\|_1\notag\\
   &= \left\| \sum_{i=1}^k \lambda_i \cdot \left(\chi^{A_i\setminus F} - \chi^{F\setminus A_i}\right) \right\|_1\notag\\
   &=  \sum_{i=1}^k \lambda_i \cdot \left\| \chi^{A_i\setminus F} - \chi^{F\setminus A_i} \right\|_1\notag\\
   &=  \sum_{i=1}^k \lambda_i \cdot |A_i\symdiff F|\enspace,\label{eq:lambdaLargeAA}
\end{align}
where the second equality follows from~\eqref{eq:FstarinConeAtF}, and the forth one from the fact that all vectors $\chi^{A_i\setminus F}-\chi^{F\setminus A_i}$ in the sum are non-positive on entries corresponding to $F$ and non-negative on all other entries.
Now, because $A_i$ are neighbors of $F$ on $P_{\mathcal{F}}$, we have $A_i\neq F$ for $i\in [k]$, and hence $|A_i\symdiff F| \geqslant 1$.
This implies, together with~\eqref{eq:lambdaLargeAA}, the first statement of the claim.
Moreover, if $(E,\mathcal{F})$ is an equal-cardinality system, then $A_i\neq F$ implies $|A_i\symdiff F| \geqslant 2$, which leads to the strengthened statement of the claim for equal-cardinality systems.
\end{proof}%
\end{adjustwidth}
\medskip

\noindent Taking the scalar product of $c$ with both sides of~\eqref{eq:FstarinConeAtF}, and rearranging terms, we get
\begin{equation*}
c(F) - c(\Fstar) = \sum_{i=1}^k \lambda_i \cdot \left(c(F) - c(A_i)\right)\enspace.
\end{equation*}
Using an averaging argument, there exists an index $j\in [k]$ such that
\begin{equation*}
\frac{1}{\lambda} \cdot \left( c(F) - c(\Fstar) \right) \leqslant c(F) - c(A_j)\enspace,
\end{equation*}
which is equivalent to
\begin{equation}\label{eq:AjIsCheapLam}
c(A_j) - c(\Fstar) \leqslant  \left(1 - \frac{1}{\lambda}\right)\left(c(F)-c(\Fstar)\right)\enspace.
\end{equation}

We will show that $A\coloneqq A_j$ fulfills the statement of the lemma.
First observe that~\eqref{eq:AjIsCheapLam} together with the claim implies that $A$ fulfills \cref{item:AIsCheap,item:AIsCheapPrime}, respectively.
To show \cref{item:AReducesSymDiff,item:AReducesSymDiffPrime}, we show that any $A_i$ for $i\in [k]$ fulfills $|\Fstar\symdiff A_i| < |\Fstar\symdiff F|$.
This indeed implies both \cref{item:AReducesSymDiff,item:AReducesSymDiffPrime}, because when dealing with equal-cardinality systems, the symmetric difference between any two sets in the system has even cardinality.
Hence, it remains to show $|\Fstar\symdiff A_i| < |\Fstar\symdiff F|$.

We start by observing that equation~\eqref{eq:FstarinConeAtF} implies $A_i\subseteq \Fstar\cup F$.
Indeed, if there were any $e\in A_i$ with $e\not\in \Fstar\cup F$, then this would lead to a strictly positive entry for $e$ on the right-hand side of~\eqref{eq:FstarinConeAtF}, whereas $\chi^{\Fstar}$, which appears on the left-hand side of~\eqref{eq:FstarinConeAtF}, has a $0$-entry at $e$.
Analogously, we can derive that $\Fstar\cap F\subseteq A_i$, because if there was $e\in (\Fstar\cap F)\setminus A_i$, then this would imply that the right-hand side of~\eqref{eq:FstarinConeAtF} has as its entry at $e$ a value strictly less than $1$, contradicting that the left-hand side has a value of $1$ at entry $e$.
In summary, we have $A_i\subseteq \Fstar\cup F$ and $\Fstar\cap F \subseteq A_i$.
However, among all sets satisfying these properties, the set $F$ is the unique set that maximizes the symmetric difference with $\Fstar$.
Because $A_i\neq F$ we thus have $|\Fstar\symdiff A_i| < |\Fstar\symdiff F|$, as desired, which finishes the proof of \cref{lem:thereIsGoodNeighbor}.
\end{proof}

\cref{lem:thereIsGoodNeighbor} is a statement about finding good sets in the $1$-neighborhood of any set $F$.
By repeatedly applying the lemma, we obtain the following generalization for $q$-neighborhoods.
\begin{lemma}\label{lem:thereIsGoodQNeighbor}
Let $(E,\mathcal{F})$ be a set system.
Let $c\in \mathbb{R}^E$, and let $F,\,\Fstar\in \mathcal{F}$ with $F\neq \Fstar$.
Then, for any $q\in \{1,\ldots, |\Fstar\symdiff F|\}$, there exists a set $A\in \mathcal{F}$ in the $q$-neighborhood of $F$ on $P_{\mathcal{F}}$ satisfying
\begin{enumerate}
\item\label{item:ACheapCor} $c(A) - c(\Fstar) \leqslant \left(1-\frac{q}{|\Fstar\symdiff F|}\right) \cdot (c(F) - c(\Fstar))$, and
\item\label{item:ACloseToFstarCor} $|\Fstar\symdiff A| \leqslant |\Fstar\symdiff F| - q$.
\end{enumerate}
Moreover, if $(E,\mathcal{F})$ is an equal-cardinality system, then we obtain the following strengthening.
For any $q\in \{1,\ldots, \frac{1}{2}|\Fstar\symdiff F|\}$, there exists a set $A\in \mathcal{F}$ in the $q$-neighborhood of $F$ on $P_{\mathcal{F}}$ satisfying
\begin{enumerate}[label=(\roman*')]
\item $c(A) - c(\Fstar) \leqslant \left(1-\frac{2q}{|\Fstar\symdiff F|}\right) \cdot (c(F) - c(\Fstar))$, and
\item $|\Fstar\symdiff A| \leqslant |\Fstar\symdiff F| - 2q$.
\end{enumerate}
\end{lemma}
\begin{proof}
We prove the lemma by induction on $q$.
For $q=1$ the statement holds due to \cref{lem:thereIsGoodNeighbor}.
Now assume $q>1$, and we show the inductive step for the case where $(E,\mathcal{F})$ is not necessarily an equal-cardinality system.
The extension to equal-cardinality systems is analogous.
By the inductive hypothesis, there is a set $\overline{A}\in \mathcal{F}$ in the $(q-1)$-neighborhood of $F$ on $P_{\mathcal{F}}$ satisfying
\begin{enumerate}[label=(\alph*)]
\item\label{item:indHypACheap} $c(\overline{A}) - c(\Fstar) \leqslant \left(1-\frac{q-1}{|\Fstar\symdiff F|}\right) \cdot (c(F) - c(\Fstar))$, and
\item\label{item:indHypACloseToFstar} $|\Fstar\symdiff \overline{A}| \leqslant |\Fstar\symdiff F| - (q-1)$.
\end{enumerate}
Moreover, applying \cref{lem:thereIsGoodNeighbor} to $F=\overline{A}$, we obtain that there is a set $A\in \mathcal{F}$ in the $1$-neighborhood of $\overline{A}$ in $P_{\mathcal{F}}$---and hence, $A$ is a $q$-neighbor of $F$ in $P_{\mathcal{F}}$---such that
\begin{enumerate}[label=(\alph*)]
\setcounter{enumi}{2}
\item\label{item:ACheapOneStep} $c(A) - c(\Fstar) \leqslant \left(1-\frac{1}{|\Fstar\symdiff \overline{A}|}\right) \cdot (c(\overline{A}) - c(\Fstar))$, and
\item\label{item:ACloseToFstarOneStep} $|\Fstar\symdiff A| \leqslant |\Fstar\symdiff \overline{A}| - 1$.
\end{enumerate}
The fact that $A$ fulfills \cref{item:ACloseToFstarCor} is now an immediate consequence of \cref{item:indHypACloseToFstar,item:ACloseToFstarOneStep}.
Moreover, we have
\begin{align*}
c(A) - c(\Fstar) &\leqslant \left( 1 - \frac{1}{|\Fstar\symdiff \overline{A}|} \right) \cdot \left( 1 - \frac{q-1}{|\Fstar\symdiff F|} \right) \cdot \left( c(F) - c(\Fstar) \right)\\
 &\leqslant \left(1 - \frac{1}{|\Fstar\symdiff F| - (q-1)} \right) \cdot \left( 1 - \frac{q-1}{|\Fstar\symdiff F|} \right) \cdot \left( c(F) - c(\Fstar) \right)\\
 &= \left(1 - \frac{q}{|\Fstar\symdiff F|} \right)\cdot \left( c(F) - c(\Fstar) \right)\enspace,
\end{align*}%
where the first inequality follows from \cref{item:ACheapOneStep,item:indHypACheap}, and the second one from \cref{item:indHypACloseToFstar}.
Hence, this shows that $A$ also fulfills \cref{item:ACheapCor} and finishes the proof.
\end{proof}

Our next lemma, \cref{lem:canAlterCloseSol}, shows that if the value of a set $T\in \mathcal{F}$ is not significantly larger than $c^\top y$, then there is a good solution in its neighborhood.
Afterwards, in \cref{lem:likelyhoodToBeClose}, we provide a lower bound for the probability of this happening if $T$ has a distribution with marginals given by $y$, which is the setting of \cref{thm:neighborhoodGen}.

\begin{lemma}\label{lem:canAlterCloseSol}
Let $(E,\mathcal{F})$ be a set system, let $y\in P_{\mathcal{F}}$, $\mu \geqslant 1$, and $\eta = \min\{c(F) \mid F\in \mathcal{F}_y\}$.
Then for any $T\in \mathcal{F}_y$ with
\begin{equation*}
(\mu-1)\cdot (c(T) - \eta) \leqslant \mu \cdot (c^\top y - \eta)\enspace,
\end{equation*}
there is a set $U\in\mathcal{F}_y$ in the $\lceil \sfrac{2\rho(\mathcal{F}_y)}{\mu}\rceil$-neighborhood of $T$ on $P_{\mathcal{F}_y}$ with $c(U)\leqslant c^\top y$.
Moreover, if $(E,\mathcal{F}_y)$ is an equal-cardinality system, then such a set $U\in \mathcal{F}$ exists in the $\lceil \sfrac{\rho(\mathcal{F}_y)}{\mu} \rceil$-neighborhood of $T$ on $P_{\mathcal{F}_y}$.
\end{lemma}
\begin{proof}
The statement trivially holds for $\mu=1$.
Hence, assume $\mu>1$.
Let $q\coloneqq \lceil \sfrac{\rho(\mathcal{F}_y)}{\mu} \rceil$ if $(E,\mathcal{F})$ is an equal-cardinality system, and $q\coloneqq \lceil \sfrac{2\rho(\mathcal{F}_y)}{\mu} \rceil$ otherwise.
Furthermore, we define
\begin{equation*}
\Fstar \in \argmin\left\{c(F) \mid F\in \mathcal{F}_y \right\}\enspace,
\end{equation*}
and hence, $c(\Fstar) = \eta$.
By \cref{lem:thereIsGoodQNeighbor}, there is a set $U\in \mathcal{F}_y$ in the $q$-neighborhood of $T$ on $P_{\mathcal{F}_y}$ satisfying
\begin{align*}
c(U) - \eta &\leqslant \left(1 - \frac{2\rho(\mathcal{F}_y)}{|\Fstar\symdiff T| \cdot\mu} \right)\cdot (c(T) - \eta) \\
  &\leqslant \left(1 - \frac{1}{\mu} \right)\cdot (c(T) - \eta)\\
  &\leqslant \left(1 - \frac{1}{\mu} \right)\cdot \frac{\mu}{\mu-1} (c^\top y - \eta)\\
  &= c^\top y - \eta\enspace,
\end{align*}
where the second inequality follows from $2\rho(\mathcal{F}_y)\geqslant |\Fstar\symdiff T|$, and the third one from the inequality given in the statement of \cref{lem:canAlterCloseSol}.
Hence, $U$ fulfills the properties required by \cref{lem:canAlterCloseSol}, which finishes the proof.
\end{proof}

\begin{lemma}\label{lem:likelyhoodToBeClose}
Let $(E,\mathcal{F})$ be a set system, let $y\in P_{\mathcal{F}}$, $\mu \geqslant 1$, and $\eta = \min\{c(F) \mid F\in \mathcal{F}_y\}$.
Moreover, let $T$ be a random set in $\mathcal{F}_y$ drawn from a distribution that satisfies $\Pr[e\in T]=y_e$ for all $e\in E$.
Then
\begin{equation*}
\Pr\left[(\mu-1) \cdot (c(T) - \eta) \leqslant \mu\cdot(c^\top y - \eta) \right] \geqslant \frac{1}{\mu}\enspace.
\end{equation*}
\end{lemma}
\begin{proof}
The statement is trivial for $\mu=1$, so assume $\mu > 1$ for the rest of this proof.
First observe that $c(T) - \eta$ is a non-negative random variable with expected value $c^\top y-\eta$, as $\Pr[e\in T] = y_e$ for all $e\in E$.
If $c^\top y = \eta$, then the statement holds because $c(T)-\eta$ has expectation zero and is non-negative; thus, it is $0$ with probability~$1$.
Assume from now on $c^\top y > \eta$.
Then, the lemma is a consequence of Markov's inequality, which implies
\begin{align*}
\Pr\left[
c(T) - \eta \geqslant \frac{\mu}{\mu-1} \cdot (c^\top y - \eta)
\right]
&\leqslant \frac{\mu-1}{\mu}\enspace.
\end{align*}
Writing $\frac{\mu-1}{\mu}=1-\frac1\mu$, we see that this implies the statement of the lemma.
\end{proof}

Finally, combining \cref{lem:canAlterCloseSol} and \cref{lem:likelyhoodToBeClose}, \cref{thm:neighborhoodGen} now readily follows.

\begin{proof}[Proof of \cref{thm:neighborhoodGen}]%
By choosing $\mu=\sfrac{2\rho(\mathcal{F})}{q}$ in both \cref{lem:canAlterCloseSol} and \cref{lem:likelyhoodToBeClose} we immediately obtain the first part of \cref{thm:neighborhoodGen}.
The bound for the case of equal cardinality set systems is obtained by setting $\mu=\sfrac{\rho(\mathcal{F})}{q}$ in both \cref{lem:canAlterCloseSol} and \cref{lem:likelyhoodToBeClose}.
\end{proof}

\subsection{Further applications of alteration technique}

The presented alteration technique is a rather general approach that is not tightly linked to the \MCCST setting where we applied it.
It may thus be of independent interest.
In particular, it can be used to avoid a multiplicative loss in the objective in several contexts where (randomized) rounding approaches are used.
We briefly mention one such further application.
The following packing result was shown in~\cite{chekuri2009dependent_arxiv}.

\begin{theorem}[Theorem 6.2 in~\cite{chekuri2009dependent_arxiv}]
Let $P$ be the base polytope of a matroid on ground set $N$, let $A\in [0,1]^{m\times N}$, and let $b\in \mathbb{R}^m$.
Then there is a $(1+\varepsilon, \Oh(\sfrac{\log m}{\log \log m}))$-bicriteria approximation for the problem
\begin{equation*}
\min \left\{ c^\top x \;\middle\vert\; x\in \{0,1\}^N,\ x\in P,\ Ax\leqslant b \right\}\enspace,
\end{equation*}
where the first guarantee is w.r.t.~the cost of the solution and the second one w.r.t.~the overflow on the packing constraints.
\end{theorem}
The above approximation was obtained by rounding a point $y\in P$ that fulfills $c^\top y \leqslant c(\OPT)$ through a negatively correlated rounding procedure.
Such a procedure preserves marginals, and hence, falls into the setting of our \cref{thm:neighborhoodGen}, which allows for avoiding the loss in the objective, at an additive $+1$ cost in the second objective, which is negligible.
Through this alteration, one obtains a unicriteria $(1,\Oh(\sfrac{\log m}{\log \log m}))$-approximation.

\subsection[Alternative approach to avoid \texorpdfstring{${1+\varepsilon}$}{1+epsilon} loss via techniques of Linhares and Swamy]{\boldmath Alternative approach to avoid ${1+\varepsilon}$ loss via techniques of Linhares and Swamy\unboldmath}\label{sec:linharesSwamy}

We briefly want to highlight that a recently introduced approach of \textcite{linhares_2018_reduction} also allows for avoiding a $(1+\varepsilon)$-factor loss in the objective.
More precisely, they introduced a Lagrangian relaxation based approach to reduce certain bicriteria weighted packing problems to bicriteria unweighted packing problems.
Within this framework, they also show how it can be modified to avoid losses in the objective value under some conditions.
For this they need an LP-based rounding procedure with certain properties.

In the following, we focus on the specific problem of \MCCST to expand further on this approach and how it can be made to work in this context.
For simplicity, consider an \MCCST problem with only upper bounds on the chain constraints, which falls within the setting of packing constraints considered in~\cite{linhares_2018_reduction}.
Let $y\in Q$ be a fractional point as computed in the first step of \cref{alg:MCCST}.
The point $y$ can be interpreted as an optimal LP solution to a linear program on the minimal face of the spanning tree polytope on which $y$ lies, together with upper bounds on the chain constraints of large $y$-value.
This allows for interpreting $y$ as an LP solution as required by the framework of \citeauthor{linhares_2018_reduction}.
Additionally, the framework needs a rounding procedure that both \begin{enumerate*}
\item rounds $y$ to a spanning tree $T$ on the same minimal face of the spanning tree polytope on which $y$ lies, and
\item the spanning tree $T$ needs to satisfy that $|T\cap \delta(S_i)|$ is within a $(1\pm \varepsilon)$-factor of $b_i$ for each chain constraint corresponding to a set $S_i$ for which $y$ is tight, i.e., $y(\delta(S_i))=b_i$.
\end{enumerate*}
Notice that it is important that the rounding does not just return a tree almost fulfilling the chain constraints, but we also need that $y$-tight chain constraints remain nearly-tight after rounding.
Our alteration step does not require such a property, but has other requirements.
Hence, the two techniques are not strictly comparable.

The negatively correlated rounding procedure that we employ fulfills both requirements stated above.
It always rounds to a spanning tree on the same minimal face, because it is marginal-preserving.
Moreover, equation~\eqref{eq:oneChainTGood} shows that the load on chain constraints does not change much.

Finally, we want to mention that the framework in~\cite{linhares_2018_reduction} can also be adjusted to deal with lower bounds in our context of \MCCST.
\section[Extension to MLCST]{Extension to \MLCST}\label{sec:extensionMLCST}

The key ingredient of the approximation algorithm for \MCCST presented in the previous sections is obtaining a $\tau$-integral point $y$ that is feasible for the linear relaxation of the problem.
Ideally, we would like to find such a point in the more general case of \MLCST as well, and then apply \cref{thm:rounding,thm:neighborhood} for rounding and local corrections, as before.

Unfortunately, analyzing the natural generalization of our dynamic programming approach, where we determine partial solutions $y_{\sfc}$ for all connectivity triples $\sfc$ with $S\in\mathcal{L}$ by continuously extending previously obtained solutions, comes with obstacles even if the laminar family $\mathcal{L}$ has constant width.
To highlight some aspects thereof, consider the instance given in \cref{fig:incompatiblePatterns}, where the laminar family $\mathcal{L}$ consists of precisely two sets $S_1$ and $S_2$, and edge costs $c$ are such that edges in $E[S_1]\cup E[S_2]$ have cost $1$, and all other edges have cost $0$.

\begin{figure}[ht]
\centering
\begin{tikzpicture}[]
\small

\pgfdeclarelayer{background}
\pgfsetlayers{background,main}

\begin{scope}[every node/.style={ns}]
\node (u1) at (0,1) {};
\node (v1) at (0,0) {};
\node (w1) at (0,-1) {};
\node (u1') at (1,1) {};
\node (v1') at (1,0) {};
\node (w1') at (1,-1) {};
\node (u2') at (3,1) {};
\node (v2') at (3,0) {};
\node (w2') at (3,-1) {};
\node (u2) at (4,1) {};
\node (v2) at (4,0) {};
\node (w2) at (4,-1) {};
\end{scope}

\begin{scope}[]
\node [above left=-4pt of u1] {$u_1$};
\node [above=-2pt of v1] {$v_1$};
\node [below left=-4pt of w1] {$w_1$};
\node [above right=-4pt of u2] {$u_2$};
\node [below=-1pt of v2] {$v_2$};
\node [below right=-4pt of w2] {$w_2$};

\node [above=-2pt of u1'] {$u_1'$};
\node [above=-2pt of v1'] {$v_1'$};
\node [above=-2pt of w1'] {$w_1'$};
\node [above=-2pt of u2'] {$u_2'$};
\node [above=-2pt of v2'] {$v_2'$};
\node [above=-2pt of w2'] {$w_2'$};
\end{scope}

\begin{scope}[thick]
\draw[line width=2pt] (u1) -- (u1') -- (u2') -- (u2) -- (v2) -- (v2') -- (v1') -- (v1) -- (w1) -- (w1') -- (w2') -- (w2);
\draw [bend right] (u1) to (w1);
\draw [bend left] (u2) to (w2);
\end{scope}

\begin{scope}[rounded corners, thick]
\draw (0.5,1.5) rectangle (-0.7,-1.5) node[left=-2pt] {$S_1$};
\draw (3.5,1.5) rectangle (4.7,-1.5) node[right=-2pt] {$S_2$};
\end{scope}

\begin{pgfonlayer}{background}
\begin{scope}[fill opacity=0.5]
\fill[fill=green!50!brown] (0.5,1) ellipse [x radius=0.5cm,y radius=0.15cm];
\fill[fill=green!50!brown] (0.5,0) ellipse [x radius=0.5cm,y radius=0.15cm];
\fill[fill=green!50!brown] (0.5,-1) ellipse [x radius=0.5cm,y radius=0.15cm];
\fill[fill=red!50!brown] (3.5,1) ellipse [x radius=0.5cm,y radius=0.15cm];
\fill[fill=red!50!brown] (3.5,0) ellipse [x radius=0.5cm,y radius=0.15cm];
\fill[fill=red!50!brown] (3.5,-1) ellipse [x radius=0.5cm,y radius=0.15cm];
\end{scope}
\begin{scope}[rounded corners, thick, fill opacity=0.2]
\filldraw[fill=green!50!brown!50, draw=green!50!black] (0.7, -0.2) rectangle (1.3,1.6);
\filldraw[fill=green!50!brown, draw=green!50!black] (0.7, -1.2) rectangle (1.3,-0.4);
\filldraw[fill=red!50!brown, draw=red!50!black] (2.7, 0.8) rectangle (3.3,1.6);
\filldraw[fill=red!50!brown, draw=red!50!black] (2.7, -1.2) rectangle (3.3,0.6);
\end{scope}
\end{pgfonlayer}

\node[below=5pt of w1', green!50!black] {$\mathcal{C}_1$};
\node[below=5pt of w2', red!50!black] {$\mathcal{C}_2$};

\draw[line width=2pt] (5.1,0.7) -- (5.9,0.7) node[right] {$\OPT$};
\fill[fill=green!50!brown, fill opacity=0.5] (5.5,0) ellipse [x radius=0.4cm,y radius=0.1cm] node[right=0.5cm, black, opacity=1] {$F_1$};
\fill[fill=red!50!brown, fill opacity=0.5] (5.5,-0.7) ellipse [x radius=0.4cm,y radius=0.1cm] node[right=0.5cm, black, opacity=1] {$F_2$};

\end{tikzpicture} %
\caption{Connectivity triples $\sfc[1]$ and $\sfc[2]$ induced by an optimal solution $\OPT$.}
\label{fig:incompatiblePatterns}
\end{figure}
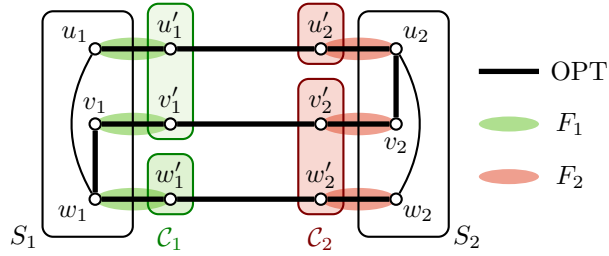

It is easy to see that every minimum spanning tree (and therefore also optimal solution) has cost $2$, as for example the one indicated in bold.
Note that in our dynamic program, we would first construct solutions $y_{(S_i,F,\mathcal{C})}$ compatible with the pattern $(S_i,F,\mathcal{S})$ for all such connectivity patterns, and then try to extend every combination of solutions on $S_1$ and $S_2$ to a global solution.
Typically, such dynamic programming approaches are analyzed by backtracing an optimal solution.
In our example, consider the optimal solution given in \cref{fig:incompatiblePatterns}.
This solution induces the connectivity triples $\sfc[1]$ and $\sfc[2]$ on $S_1$ and $S_2$, respectively, where $F_i=\{\{u_i,u_i'\},\{v_i,v_i'\},\{w_i,w_i'\}\}$ for $i\in\{1,2\}$, $\mathcal{C}_1=\{\{u_1', v_1'\},\{w_1'\}\}$ and $\mathcal{C}_2=\{\{u_2'\}, \{v_2',w_2'\}\}$, as indicated.
Ideally, we would like that the cheapest common extension of $y_{\sfc[1]}$ and $y_{\sfc[2]}$ has cost at most the cost of an optimal solution.
However, note that we could potentially have
\begin{equation*}
y_{\sfc[1]} = \chi^{\{\{u_1,w_1\}\}} \quad\text{and}\quad y_{\sfc[2]} = \chi^{\{\{u_2,w_2\}\}}\enspace,
\end{equation*}
in which case we can easily see that %
there does not even exist a feasible solution that restricts to $y_{\sfc[1]}$ and $y_{\sfc[2]}$ on $S_1$ and $S_2$, respectively.

Note that in the case of \MCCST, where we only had to consider extensions from a single partial solution on a smaller set, the analysis outlined above was enough to obtain our result (see the proof of \cref{lem:completionOk}).
In particular, it was enough to know that the dynamic program considered building a solution along the small cuts and connectivity triples induced by an optimal solution (even though these are not known upfront).
For \MLCST, we deviate from this typical analysis and exploit that the DP considers all potential connectivity triples on the small cuts induced by an optimal solution.
Let us illustrate this in the above example.
We claim that there exist connectivity patterns $\mathcal{C}_1'$ and $\mathcal{C}_2'$ (potentially different from the patterns $\mathcal{C}_1$ and $\mathcal{C}_2$ induced by the optimal solution) such that the best common extension of $y_{(S_1,F_1,\mathcal{C}_1')}$ and $y_{(S_2,F_2,\mathcal{C}_2')}$ has cost at most $2$.
To see this, we proceed iteratively, starting with $\mathcal{C}_1'=\mathcal{C}_1$.
Note that if on $E[S_1]$, we replace the edges of the optimal solution by those of $y_{\sfc[1]}$, the new point is a feasible solution and, by definition of $y_{\sfc[1]}$,  the total cost does not increase.
Observe that the new set of edges induces a different connectivity pattern on $S_2$ than $\OPT$ did, and let this pattern be $\mathcal{C}_2'$ (see \cref{fig:inducingPatterns}).
Now replacing the optimal solution on $E[S_2]$ by $y_{(S_2,F_2,\mathcal{C}_2')}$, which is $\chi^{\{\{u_2,v_2\}\}}$ in our example, we again see that feasibility is guaranteed, and the total cost does again not increase.
To finish the argument, note that we just constructed a common extension of the two partial solutions $y_{(S_1,F_1,\mathcal{C}_1')}$ and $y_{(S_2,F_2,\mathcal{C}_2')}$ of cost no more than the cost of the optimal solution---thus, the best extension will be of cost at most the cost of $\OPT$, proving the desired guarantee.

\begin{figure}[ht]
\centering
\begin{tikzpicture}[]
\small

\pgfdeclarelayer{background}
\pgfsetlayers{background,main}

\begin{scope}[every node/.style={ns}]
\node (u1) at (0,1) {};
\node (v1) at (0,0) {};
\node (w1) at (0,-1) {};
\node (u1') at (1,1) {};
\node (v1') at (1,0) {};
\node (w1') at (1,-1) {};
\node (u2') at (3,1) {};
\node (v2') at (3,0) {};
\node (w2') at (3,-1) {};
\node (u2) at (4,1) {};
\node (v2) at (4,0) {};
\node (w2) at (4,-1) {};
\end{scope}

\begin{scope}[]
\node [above left=-4pt of u1] {$u_1$};
\node [above=-2pt of v1] {$v_1$};
\node [below left=-4pt of w1] {$w_1$};
\node [above right=-4pt of u2] {$u_2$};
\node [below=-1pt of v2] {$v_2$};
\node [below right=-4pt of w2] {$w_2$};

\node [above=-2pt of u1'] {$u_1'$};
\node [above=-2pt of v1'] {$v_1'$};
\node [above=-2pt of w1'] {$w_1'$};
\node [above=-2pt of u2'] {$u_2'$};
\node [above=-2pt of v2'] {$v_2'$};
\node [above=-2pt of w2'] {$w_2'$};
\end{scope}

\begin{scope}[thick]
\draw[line width=2pt] (u1) -- (u1') -- (u2') -- (u2) -- (v2) -- (v2') -- (v1') -- (v1);
\draw[line width=2pt, dashed, gray] (v1) -- (w1);
\draw[line width=2pt] (w1) -- (w1') -- (w2') -- (w2);
\draw [bend right, line width=2pt, blue] (u1) to (w1);
\draw [bend left] (u2) to (w2);
\end{scope}

\begin{scope}[rounded corners, thick]
\draw (0.5,1.5) rectangle (-0.6,-1.5) node[left=-2pt] {$S_1$};
\draw (3.5,1.5) rectangle (4.6,-1.5) node[right=-2pt] {$S_2$};
\end{scope}

\begin{pgfonlayer}{background}
\begin{scope}[fill opacity=0.5]
\fill[fill=green!50!brown] (0.5,1) ellipse [x radius=0.5cm,y radius=0.15cm];
\fill[fill=green!50!brown] (0.5,0) ellipse [x radius=0.5cm,y radius=0.15cm];
\fill[fill=green!50!brown] (0.5,-1) ellipse [x radius=0.5cm,y radius=0.15cm];
\fill[fill=red!50!brown] (3.5,1) ellipse [x radius=0.5cm,y radius=0.15cm];
\fill[fill=red!50!brown] (3.5,0) ellipse [x radius=0.5cm,y radius=0.15cm];
\fill[fill=red!50!brown] (3.5,-1) ellipse [x radius=0.5cm,y radius=0.15cm];
\end{scope}
\begin{scope}[rounded corners, thick, fill opacity=0.2]
\filldraw[fill=green!50!brown!50, draw=green!50!black] (0.7, -0.2) rectangle (1.3,1.6);
\filldraw[fill=green!50!brown, draw=green!50!black] (0.7, -1.2) rectangle (1.3,-0.4);
\filldraw[fill=red!50!brown, draw=red!50!black] (2.3,0.8) -- (2.3,-0.4) -- (2.7,-1.2) -- (3.3,-1.2) -- (3.3,-0.4) -- (2.5,-0.4) -- (2.5, 0.8) -- (3.3,0.8) -- (3.3,1.6) -- (2.7,1.6) -- cycle;
\filldraw[fill=red!50!brown, draw=red!50!black] (2.7, -0.2) rectangle (3.3,0.6);
\end{scope}
\end{pgfonlayer}

\node[below=5pt of w1', green!50!black] {$\mathcal{C}_1$};
\node[below=5pt of w2', red!50!black] {$\mathcal{C}_2'$};

\draw[line width=2pt, blue] (4.9,0.7) -- (5.7,0.7) node[below right=-6pt and 0.1cm] {$y_{\sfc[1]}$};%
\fill[fill=green!50!brown, fill opacity=0.5] (5.3,0) ellipse [x radius=0.4cm,y radius=0.1cm] node[right=0.5cm, black, opacity=1] {$F_1$};
\fill[fill=red!50!brown, fill opacity=0.5] (5.3,-0.7) ellipse [x radius=0.4cm,y radius=0.1cm] node[right=0.5cm, black, opacity=1] {$F_2$};

\end{tikzpicture}
\caption{Patching $y_{\sfc[1]}$ induces $\mathcal{C}_2'$.}
\label{fig:inducingPatterns}
\end{figure}
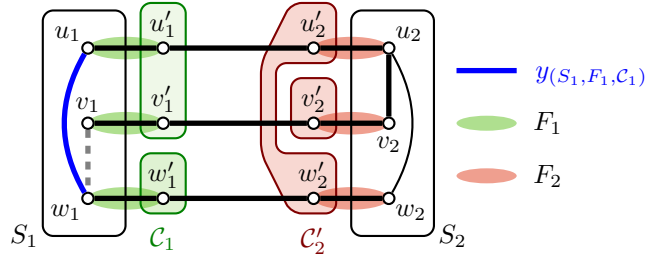

The above idea of iteratively defining suitable connectivity patterns can be generalized to an arbitrary number of sibling sets $S_1,\ldots, S_w$, and is crucial in the analysis of the propagation step of our dynamic program.
For now, we want to highlight that replacing parts of an optimal solution by a previously obtained partial solution and then inducing new connectivity patterns on other parts requires the partial solutions to be integral, as there is no notion of induced connectivity patterns for fractional solutions.
In our extension steps, however, we find common extensions of partial solutions through a linear program similar to~\eqref{eq:lpCompDP}, which will in general not be integral.
For this reason, we apply the rounding and local correction methods presented in \cref{thm:rounding,thm:neighborhood} after every single extension step, giving integral solutions at every stage and thus allowing for inducing connectivity patterns iteratively as defined above.

The example from \cref{fig:incompatiblePatterns} discussed earlier might seem contrived due to the fact that the highlighted problems could be avoided by breaking ties in the right way, for example by choosing different optimal partial solutions $y_{\sfc[1]}$ and $y_{\sfc[2]}$.
Even though this is the case here, there exist more complex examples exhibiting the same issue that do not have any tie-breaking options.
One such example is presented in \cref{sec:exampleBacktracingOPT}.

It is important to also recall that \cref{thm:rounding,thm:neighborhood} provide guarantees on constraint violation and cost of the integral solution obtained through rounding only with certain probabilities.
For a more concise analysis, our algorithm will at each step apply the rounding and local correction operations repeatedly until---in expected polynomial running time---an integral solution with the desired properties is found.
More formally, we obtain a Las Vegas algorithm with the following guarantees.

\begin{theorem}\label{thm:MLCST_randRunningTime}
For every $\varepsilon > 0$, there is a $(1,1+\varepsilon)$-approximation algorithm for \MLCST with expected running time $|V|^{\Oh(\sfrac{k\log |V|}{\varepsilon^2})}$, where $k$ is the width of the laminar family $\mathcal{L}$.
\end{theorem}

Note that any Las Vegas algorithm can easily be transformed into a randomized approximation algorithm, i.e, with deterministic polynomial running time and where the returned solution has the desired properties with high probability: By Markov's inequality, the probability that the running time of a single run of the algorithm guaranteed by \cref{thm:MLCST_randRunningTime} is less that twice the expected running time is at least $\sfrac12$.
Consequently, among $\log_2 |V|$ many independent runs, the probability that at least one run succeeds is at least $1-\sfrac1{|V|}$.
Thus, \cref{thm:MLCST} stated in the introduction is implied by \cref{thm:MLCST_randRunningTime}.

In the following two sections, we present the modifications of our dynamic programming approach for the laminar case in detail, and we expand on the ideas highlighted above, leading to a proof of \cref{thm:MLCST_randRunningTime}.%

\subsection{Dynamic programming in the laminar case}

For a formal description of the dynamic programming algorithm for \MLCST, we stick to the notation defined for \MCCST, now of course considering the laminar family $\mathcal{L}$ instead of the chain $S_1\subsetneq\ldots\subsetneq S_k$.
We denote by $\mathcal{K}$ the set of all connectivity triples $\sfc$ where $S\in\mathcal{L}\cup\{\emptyset,V\}$, with the crossing edges $F\subseteq\delta(S)$ satisfying $a_S\leqslant |F|\leqslant \min\{\tau, b_S\}$ for $S\in\mathcal{L}$, where the parameter $\tau$ is the maximal number of edges that the dynamic program ``guesses'' in small cuts.
Like in \MCCST, we again choose $\tau=\Oh(\sfrac{\log|V|}{\varepsilon^2})$.

For every $\sfc\in\mathcal{K}$, we use our dynamic program to determine a partial solution $T_{\sfc}\subseteq E[S]$ with the following properties.
For $S\in\mathcal{L}$, we let $\mathcal{L}_{S}\coloneqq\{S'\in\mathcal{L}\mid S'\subsetneq S\}$.

\needspace{5\baselineskip}
\begin{property}\label{prop:laminarDP}\leavevmode
\begin{enumerate}
\item\label{item:T_SFC_tree} $T_{\sfc}$ is a spanning tree of $G\sfc$.
\item\label{item:T_SFC_cutConstraints} $(1-\varepsilon)\cdot a_{S'}\leqslant |(T_{\sfc}\cup F)\cap\delta(S')| \leqslant (1+\varepsilon)\cdot b_{S'}$ for all $S'\in\mathcal{L}_S$.
\item\label{item:T_SFC_cost} $c(T_{\sfc})$ is at most the cost of a cheapest edge set $U\subseteq E[S]$ that forms a spanning tree of $G\sfc$ and satisfies $a_{S'}\leqslant |(U\cup F)\cap\delta(S')| \leqslant b_{S'}$ for all $S'\in\mathcal{L}_S$.
\end{enumerate}
\end{property}

It is clear that if, in expected running time $|V|^{\Oh(\sfrac{k\log |V|}{\varepsilon^2})}$, we can obtain trees $T_{\sfc}$ with the above properties, then \cref{thm:MLCST_randRunningTime} follows because the tree $T\coloneqq T_{(V,\emptyset,\emptyset)}$ has precisely the desired properties: \cref{item:T_SFC_tree,item:T_SFC_cutConstraints} state that $T$ is a spanning tree of $G$ violating the cut constraints by a factor of at most $(1\pm\varepsilon)$, and by \cref{item:T_SFC_cost}, $c(T)$ is at most the cost of an optimal solution.

To compute all $T_{\sfc}$, we initialize $T_{(\emptyset, \emptyset,\{\emptyset\})} = \emptyset$, and then propagate to $T_{\sfc}$ for all $\sfc\in\mathcal{K}$ in an order such that for all $\sfc,\sfcprime\in\mathcal{K}$, if $S'\subsetneq S$, then $T_{\sfcprime}$ is computed before $T_{\sfc}$.
A new tree $T_{\sfcbar}$ for $\sfcbar\in\mathcal{K}$ is obtained in two steps: First, for all choices of connectivity triples $\sfc[1],\ldots,\sfc[w]$ such that $S_1,\ldots,S_w\in\mathcal{L}_{\overline{S}}$ and $S_1,\ldots,S_w$ have pairwise empty intersections, we extend the partial solutions $T_{\sfc[j]}$ for $j\in[w]$ to a solution $T$ on $\sfcbar$.
In a second step, we find the best solution among all extensions obtained this way, and keep it as $T_{\sfcbar}$.
The full propagation procedure is summarized in \cref{alg:propagationLaminarDP}, and details of a single extension step are described in \cref{alg:extensionLaminarDP}.

\begin{algorithm2e}[H]
\caption{Extending partial solutions $T_{\sfc[1]}, \ldots, T_{\sfc[w]}$ to $\sfcbar$.\strut}\label{alg:extensionLaminarDP}
\begin{algostepsarabic}

\item\label{algostep:solveLP_laminar} Let $\tau\coloneqq\lfloor \sfrac{96\ln(2|V|)}{\varepsilon^2}\rfloor$, and let $y$ be a minimizer of\strut
\begin{align*}
\min\quad c^\top z \tag{lamExLP}\label{eq:lamExLP} \\
z &\in\PST\sfcbar\\
\max\{\tau+1,a_{S'}\}\leqslant z(\delta(S'))&+|F\cap\delta(S')| \leqslant b_{S'} && \forall S'\in\mathcal{L}\colon \exists j\in[w]\text{ with }S_j\subsetneq S'\subsetneq S\\
z(e) &= \chi^{T_{\sfc[j]}}(e) && \forall e \in E[S_j],\ \forall j\in[w]\\
z(e) &= \chi^{F_j}(e) && \forall e \in \delta(S_j),\ \forall j\in[w]\enspace.
\end{align*}

\item\label{algostep:roundY_laminar} Randomly round $y$ with a rounding procedure as guaranteed by \cref{thm:rounding} to obtain a spanning tree $T_0$ of $G\sfcbar$.

\item\label{algostep:alterateTree_laminar} Find a minimum cost spanning tree $T$ among all spanning trees of $G\sfcbar$ with $|T\symdiff T_0|\leqslant 2$ and such that $y(e)\in(0,1)$ for all $e\in T\symdiff T_0$.

\item\label{algostep:checkIterate_laminar} If $T$ satisfies
\begin{equation*}
(1-\varepsilon)\cdot a_{S'}\leqslant |(T_{\sfcbar}\cup F)\cap\delta(S')| \leqslant (1+\varepsilon)\cdot b_{S'}
\end{equation*}
for all $S'\in\mathcal{L}_{\overline S}$, and $c(T)\leqslant c^\top y$, output $T$.
Else, repeat from \cref{algostep:roundY_laminar}.
\end{algostepsarabic}
\end{algorithm2e}

Let us expand on the nature of the extension step given in \cref{alg:extensionLaminarDP}.
The purpose of the linear program~\eqref{eq:lamExLP} is to find a common extension $y$ of the trees $T_{\sfc[j]}$ for $j\in[w]$ that uses precisely the edges $F_j$ in the cuts $\delta(S_j)$, and is left-compatible with $\sfcbar$.
Note that the last condition appears in~\eqref{eq:lamExLP} as the constraint $z\in\PST\sfcbar$.
Additionally, we require that the partial solution $y$ together with the edges in $\overline F$ satisfy the cut constraints on cuts $S'\in\mathcal{L}$ with $S_i\subsetneq S'\subsetneq S$, with a load of at least $\tau+1$ on all those cuts.
Recall that the latter comes from the idea of finding $\tau$-integral fractional solutions (which is what we need to control cut sizes in the rounding procedure) and letting the dynamic program try all combinations of maximal small cuts $S_1,\ldots,S_w\subsetneq S$.

Once a fractional extension $y$ is found, we use a rounding procedure as guaranteed by \cref{thm:rounding} to round it to an integral solution, namely a spanning tree $T_0$.
Note that by definition, $y$ coincides with $T_{\sfc[j]}$ on $E[S_j]$ for all $j\in[w]$.
As the rounding scheme is marginal-preserving (\cref{item:marginals} in \cref{thm:rounding}), it follows that $T_0$ will coincide with these partial solutions as well, thus inheriting their properties.
On $E[S\setminus\bigcup_{j=1}^w S_j]$, the property in \cref{item:chernoff} will make sure that all cut constraints are satisfied up to small multiplicative errors.
In \cref{algostep:alterateTree_laminar}, we exploit the exchange steps described in \cref{thm:neighborhood} to regain potential loss in the objective that may have occurred in \cref{algostep:roundY_laminar} compared to $c^\top y$.
Both \cref{algostep:roundY_laminar,algostep:alterateTree_laminar} can fail with certain probabilities, hence we repeatedly apply them until an extension with the properties listed in \cref{algostep:checkIterate_laminar} is found.

We remark that the linear program~\eqref{eq:lamExLP} might be infeasible for several reasons (impossibility of completing the edge sets $T_{\sfc[j]}$ to a point in $\PST\sfcbar$, infeasibility of the lower bound constraints on cuts, inconsistencies among the edge sets $F_j$, loops generated by edges in $F_j$, etc.).
In such a case, we interpret the cost of a common integral extension $T$ to be $\infty$, which avoids using such extensions later on.

\begin{algorithm2e}[H]
\caption{Propagation to $T_{\sfcbar}$ from all partial solutions $T_{\sfcprime}$ with $S'\subsetneq \overline S$.\strut}\label{alg:propagationLaminarDP}
\begin{algostepsarabic}
\item For every choice of connectivity triples $\sfc[j]\in\mathcal{K}'$\strut{} for $j\in[w]$ where $S_1, \ldots, S_w \in\mathcal{L}$ are strict subsets of $\overline S$ with pairwise empty intersections, apply \cref{alg:extensionLaminarDP} to extend the trees $T_{\sfc[1]},\ldots,T_{\sfc[w]}$ to $\sfcbar$.
Let $\mathcal{T}$ be the set of all trees obtained this way.
\item Return $T_{\sfcbar}\in\argmin_{T\in\mathcal{T}}c(T)$.
\end{algostepsarabic}
\end{algorithm2e}

\cref{alg:propagationLaminarDP} considers all potential candidates for $T_{\sfcbar}$ that were obtained through extension steps, and returns one of minimum cost.

\subsection{Analyzing the DP}

We first show that the trees $T_{\sfcbar}$ computed by \cref{alg:extensionLaminarDP} satisfy \cref{prop:laminarDP}.
Formally, we prove this statement by induction.
Obviously, the point $T_{(\emptyset,\emptyset,\{\emptyset\})}$ has all desired properties.
The induction step to complete the proof is captured by the following lemma.

\begin{lemma}\label{lem:treesSatisfyProperty}
Let $\sfcbar\in\mathcal{K}$ with $\overline S\neq\emptyset$.
Assume that for all $\sfcprime\in\mathcal{K}$ with $S'\subsetneq \overline S$, we are given $T_{\sfcprime}$ satisfying \cref{prop:laminarDP}, and let $T_{\sfcbar}$ be obtained from \cref{alg:propagationLaminarDP}.
Then $T_{\sfcbar}$ satisfies \cref{prop:laminarDP}, as well.
\end{lemma}

\begin{proof}
Let $\mathcal{T}$ be defined as in \cref{alg:propagationLaminarDP}, namely the set of all $T\subseteq E[\overline S]$ that were obtained through \cref{alg:extensionLaminarDP}.
We already remarked earlier that any such $T$ is a spanning tree of the corresponding graph $G\sfcbar$, hence satisfying \cref{item:T_SFC_tree} in \cref{prop:laminarDP}.
 For \cref{item:T_SFC_cutConstraints}, we prove that every $T\in\mathcal{T}$ satisfies
\begin{equation}\label{eq:toProveForII}
(1-\varepsilon)\cdot a_{S'}\leqslant |(T\cup \overline F)\cap\delta(S')| \leqslant (1+\varepsilon)\cdot b_{S'}
\end{equation}
for all $S'\in\mathcal{L}_{\overline S}$.
Indeed, for cuts $S'$ with $S'\subsetneq S_j$ for some $j\in [w]$, we have $(T\cup F)\cap\delta(S') = (T_{\sfc[j]}\cup F_j)\cap\delta(S')$, hence~\eqref{eq:toProveForII} follows from the assumption that $T_{\sfc[j]}$ has \cref{prop:laminarDP}.
If $S'=S_j$ for some $j\in[w]$, then $(T\cup \overline F)\cap\delta(S')=F_j$, and we have $a_{S'}\leqslant |F_j| \leqslant b_{S'}$ by definition of the connectivity pattern $\sfc[j]$.
Finally, if $S_j\subsetneq S'\subsetneq \overline S$, then~\eqref{eq:toProveForII} is guaranteed by \cref{algostep:checkIterate_laminar} in \cref{alg:extensionLaminarDP}.

To see that point \cref{item:T_SFC_cost} of \cref{prop:laminarDP} holds, fix an edge set $U\subseteq E[\overline S]$ that forms a spanning tree of $G\sfcbar$ and satisfies $a_{S'}\leqslant |(U\cup \overline F)\cap\delta(S')| \leqslant b_{S'}$ for all $S'\in\mathcal{L}$ with $S'\subsetneq \overline S$.
We have to show that $c(T_{\sfcbar})\leqslant c(U)$.
As a first step, consider any $T\in\mathcal{T}$ and let $y$ be the solution of~\eqref{eq:lamExLP} that was used to obtain $T$.
By \cref{algostep:checkIterate_laminar} in \cref{alg:extensionLaminarDP}, we have $c(T)\leqslant c^\top y$.
It is thus enough to see that one of the linear programs~\eqref{eq:lamExLP} considered while propagating to $\sfcbar$ has a solution $y$ with $c(U)\geqslant c^\top y$.

To this end, let $R\subseteq\binom{V\setminus \overline S}{2}$ be any set of edges that is right-compatible with $\sfcbar$.
Then, $T\coloneqq U\cup \overline F\cup R$ is in the spanning tree polytope of $(V,E\cup R)$.
Let $S_1,\ldots,S_w\subsetneq S$ be the maximal $\chi^T$-small cuts in $\mathcal{L}_{\overline S}$, and let $F_j=T\cap\delta(S_j)$.
We define connectivity patterns $\mathcal{C}_1,\ldots,\mathcal{C}_w$ such that $\sfc[j]\in\mathcal{K}$ iteratively as follows, where $T_0\coloneqq T$ and $j\in[w]$:
\begin{align}\label{eq:iterativeC}
\text{Let $\mathcal{C}_{j}$ s.t.~$T_{j-1}$ is compatible with $\sfc[j]$, and let $T_j\coloneqq \big(T_{j-1}\setminus E[S_j]\big)\cup T_{\sfc[j]}$.}
\end{align}

First of all, observe that all edge sets $T_j$ are indeed spanning trees, making the above operation well-defined.
To see this, we proceed inductively and assume that $T_{j-1}$ is a spanning tree.
Compatibility of $T_{j-1}$ with $\sfc[j]$ implies that $T_{j-1}\cap E[V\setminus S_j]$ is right-compatible with $\sfc[j]$, while $T_{\sfc[j]}$ is left-compatible with $\sfc[j]$, by assumption.
Thus, $(T_{j-1}\cap E[V\setminus S_j])\cup F_j \cup T_{\sfc[j]} = (T_{j-1}\setminus E[S_j])\cup T_{\sfc[j]}=T_j$ is indeed a spanning tree.
We claim that the construction in~\eqref{eq:iterativeC} leads to a tree $T_w$ with the properties
\begin{enumerate}[label=(\alph*), itemsep=.5ex]
\item\label{item:TwCost} $c(T)\geqslant c(T_w)$, and
\item\label{item:TwFeasibility} $T_w\cap E[\overline S]$ is feasible for~\eqref{eq:lamExLP} when extending from $T_{\sfc[1]},\ldots,T_{\sfc[w]}$ to $\sfcbar$.
\end{enumerate}
These two properties are enough to conclude.
As $T$ and $T_w$ are identical outside of $E[\overline S]$, \cref{item:TwCost} implies that
\[ c(U)=c(T\cap E[\overline S]) \geqslant c(T_w\cap E[\overline S])\enspace. \]
Moreover, \cref{item:TwFeasibility} implies that $c(T_w\cap E[\overline S])\geqslant c^\top y$, which together with the previous inequality gives the desired $c(U)\geqslant c^\top y$.

To see \cref{item:TwCost}, we show that for all $j\in[w]$, we have $c(T_{j-1})\geqslant c(T_j)$.
By definition of $T_j$, the latter is equivalent to $c(T_{j-1}\cap E[S_j])\geqslant c(T_j\cap E[S_j])$.
But by construction, $T_{j-1}\cap E[S_j]=U\cap E[S_j]$, and $T_j\cap E[S_j]=T_{\sfc[j]}$.
Note that $U\cap E[S_j]$ is an integral solution of the subproblem on $G\sfc[j]$, and hence \cref{item:T_SFC_cost} of \cref{prop:laminarDP} for $T_{\sfc[j]}$ implies $c(U\cap E[S_j])\geqslant c(T_{\sfc[j]})$, which is precisely what we need.
Consequently, we have $c(T_{j-1})\geqslant c(T_j)$ for all $j\in[w]$, and hence
\[ c(T) = c(T_0) \geqslant c(T_1) \geqslant\ldots\geqslant c(T_w)\enspace, \]
as desired.
For \cref{item:TwFeasibility}, we check that the constraints in~\eqref{eq:lamExLP} hold for $T_w\cap E[\overline S]$, i.e., for all $S'\in\mathcal{L}$ such that there is an $j\in[w]$ with $S_j\subsetneq S'\subsetneq \overline S$, we prove
\begin{equation}\label{eq:cutConstraintForTw}
\max\{\tau+1, a_{S'}\} \leqslant |(T_w\cap E[\overline S])\cap\delta(S')| + |\overline F\cap\delta(S')| \leqslant b_{S'}\enspace.
\end{equation}
We have $|(T_w\cap E[\overline S])\cap\delta(S')| + |\overline F\cap\delta(S')| = |T_w\cap \delta(S')| = |(U\cup \overline F)\cap\delta(S')|$, where the last equality follows from the construction of $T_w$.
Thus, the lower bound $a_{S'}$ and the upper bound $b_{S'}$ are implied by the assumption on $U$.
Moreover, the lower bound $\tau+1$ follows from the definition of $S_1,\ldots,S_w$ as the maximal $\chi^{U\cup \overline F}$-small cuts, implying~\eqref{eq:cutConstraintForTw}.
Finally, by construction, $T_w$ equals $T_{\sfc[j]}$ and $F_j$ on $E[S_j]$ and $\delta(S_j)$, respectively, for all $j\in[w]$.
Consequently, $T_w\cap E[\overline S]$ satisfies all constraints of~\eqref{eq:lamExLP}.
This finishes the proof of \cref{lem:treesSatisfyProperty}.
\end{proof}

It remains to analyze the expected running time of an extension step as described in \cref{alg:extensionLaminarDP}.

\begin{lemma}\label{lem:extensionLaminar_expRunningTime}
\cref{alg:extensionLaminarDP} has expected running time $|V|^{\Oh(1)}$.
\end{lemma}

\begin{proof}
First of all, note that every single step of \cref{alg:extensionLaminarDP} can be implemented in running time $|V|^{\mathcal{O}(1)}$.
In particular, linear programs of the type~\eqref{eq:lamExLP} can be solved in strongly polynomial time by using a compact extended formulation for the spanning tree polytope with small coefficients in the constraint matrix (one can, for example, use the one by \textcite{martin1991using}, which has coefficients that are bounded by $1$ in absolute value), and then applying the framework of \textcite{tardos_1986_strongly}.
Consequently, the above lemma is reduced to proving a bound on the expected number of iterations that are needed to achieve the properties required in \cref{algostep:checkIterate_laminar}.

Replicating the analysis in the proof of \cref{thm:MCCST}, we see that \cref{thm:rounding,thm:neighborhood} imply that the tree $T$ obtained in \cref{algostep:roundY_laminar,algostep:alterateTree_laminar} of \cref{alg:extensionLaminarDP} has the desired properties with probability at least $\sfrac{1}{2|V|}$.
As all iterations are independent, the probability that we succeed precisely at iteration $\ell\in\mathbb{Z}_{>0}$, i.e., in time $\ell\cdot|V|^{\Oh(1)}$, is $(1-\sfrac{1}{2|V|})^{\ell-1}\cdot \sfrac{1}{2|V|}$.
Consequently, the expected running time is
\begin{equation*}
\sum_{\ell\geqslant 1} \ell\cdot|V|^{\Oh(1)}\cdot \left(1-\frac{1}{2|V|}\right)^{\ell-1}\cdot \frac{1}{2|V|} = |V|^{\Oh(1)}\enspace,
\end{equation*}
where we use that $\sum_{\ell\geqslant 1} \ell(1-x)^{\ell-1}=\frac1{x^2}$ for $x\in(0,1)$.
\end{proof}

Together with the bound $|\mathcal{K}|=|V|^{\Oh(\tau)}$ from \cref{prop:sizeK}, we are finally ready to prove \cref{thm:MLCST_randRunningTime}.

\begin{proof}[Proof of \cref{thm:MLCST_randRunningTime}]
We run a dynamic program that calculates trees $T_{\sfc}$ for all $\sfc\in\mathcal{K}$ starting from the initialization $T_{(\emptyset,\emptyset,\{\emptyset\})}=\emptyset$, and using \cref{alg:propagationLaminarDP} for propagation.
Note that $T_{(\emptyset,\emptyset,\{\emptyset\})}$ satisfies \cref{prop:laminarDP}, and hence by an inductive application of \cref{lem:treesSatisfyProperty}, all trees $T_{\sfc}$ satisfy \cref{prop:laminarDP}.
In particular, $T_{(V,\emptyset,\{\emptyset\})}$ is thus a tree satisfying the guarantees of \cref{thm:MLCST_randRunningTime}.

The running time is determined by the number of calls to \cref{alg:extensionLaminarDP}.
For every triple $\sfc\in\mathcal{K}$, when calculating $T_{\sfc}$, there is one call to \cref{alg:extensionLaminarDP} for every possible choice of triples $\sfc[1],\ldots,\sfc[w]\in\mathcal{K}$ with $S_1,\ldots,S_w\in\mathcal{L}_S$ having pairwise empty intersections.
Note that $\mathcal{L}$ has width $k$, so there are at most $k$ sets with pairwise empty intersection, i.e., $w\leqslant k$.
This implies that the number of calls to \cref{alg:extensionLaminarDP} is bounded from above by $|\mathcal{K}|^{k+1}$.
Consequently, the bound on the expected running time of $|V|^{\Oh(k\tau)}$ follows from combining \cref{prop:sizeK} and \cref{lem:extensionLaminar_expRunningTime}.\sloppy
\end{proof}
\section[Adaptations to TSP variants]{Adaptations to \TSP variants}\label{sec:approachPathTSP}

\newcommand{\joinset}{R}
\newcommand{\oddelement}{r}

We now expand on how our DP approach can be adapted to obtain approximation algorithms for \TSP variants, in particular \MSCJ, which is a natural generalization of \pathTSP.
This will lead to a proof of \cref{thm:TJoins}.
We remark that in order to keep notation unambiguous, we reserve the variable $T$ for trees and use $\joinset$ instead for $\joinset$-joins throughout the rest of this paper.

Note that we presented our techniques in the context of finding constrained spanning trees.
However, connected $\joinset$-joins are not necessarily spanning trees.
The following result by \textcite{cheriyan_2015_approximating} shows that there is always a shortest connected $\joinset$-join that is a spanning tree, thus linking the two problems.
\begin{theorem}[\textcite{cheriyan_2015_approximating}]\label{thm:optQJoins}
Let $G=(V,E)$ be a complete graph with metric edge lengths $\ell\colon E\to\mathbb{R}_{\geqslant 0}$, and let $\joinset\subseteq V$ be nonempty and of even cardinality.
Given a connected $\joinset$-join $J$, a spanning tree $T$ of $G$ with $\ell(T)\leqslant\ell(J)$ that is a connected $\joinset$-join can be found efficiently.
\end{theorem}
The proof of this theorem exploits the assumption that the instance is metric by using that shortcutting operations are not length-increasing.
The result is then obtained by proving that whenever a $\joinset$-join has cycles, a sequence of shortcutting operations can be applied to obtain a spanning tree.
Note that \cref{thm:optQJoins} implies that, indeed, \pathTSP is a special case of \MSCJ[\joinset], as for $\joinset=\{s,t\}$ with $s\neq t$, a spanning tree that is a connected $\joinset$-join is a Hamiltonian $s$-$t$~path.
Moreover, by \cref{thm:optQJoins}, we also see that it is enough for an approximate solution to compare well to optimal spanning trees, which is a crucial observation for simplifying our analysis.

Our proof of \cref{thm:TJoins}, our dynamic programming approach can be used to obtain a good spanning tree for a Christofides-Serdyukov-type algorithm, which starts with a spanning tree $T$ and does parity correction in a second step by adding further edges.
Thus, we start with a short recap of the approach by \textcite{christofides1976worst,serdyukov_1978_onekotorykh} (also see~\cite{christofides_1976_worst_new,vanBevern_2020_historical}), which has been used heavily for both \TSP and \pathTSP, in the context of the more general \MSCJ[\joinset].
More precisely, we walk through a polyhedral analysis of the Christofides-Serdyukov algorithm due to \textcite{wolsey1980heuristic}.
These ideas form the basis of essentially all improvements in approximation algorithms for \pathTSP over the last few years~\cite{an_2015_improving, gottschalk_2018_better, sebo_2019_salesman, traub_2019_approaching, vygen_2016_reassembling, traub_2021_reducing,karlin_2021_slightly,karlin_2022_deterministic}.
A key difference compared to chain-constrained and laminarly-constrained spanning trees that we considered previously is that, ideally, one would like that the tree $T$ satisfies parity constraints for a well-defined family of cuts.
Whereas we cannot directly impose parity constraints (one can even observe that this leads to an NP-hard problem), we present a proxy that is good enough by sacrificing an arbitrarily small constant error $\varepsilon$ in the approximation guarantee.

\subsection{The Christofides-Serdyukov algorithm and Wolsey's analysis}

The Christofides-Serdyukov algorithm for \TSP builds on the observation that a solution has to satisfy two properties, namely connectivity and correct degree parities.
Connectivity can be guaranteed by starting with a spanning tree $T$.
The degree parities of $T$ are wrong precisely at the vertices in $\joinset_T\coloneqq\odd(T)\symdiff \joinset$, which can be corrected by adding a $\joinset_T$-join $J$ to $T$.\footnote{By $\odd(T)$, we denote the set of odd-degree vertices in $T$.} The multiset obtained by combining $T$ and $J$ can be shortcut to a solution of the problem, and approximation guarantees follow by choosing $T$ and $J$ such that $\ell(T)$ and $\ell(J)$ can be bounded in terms of $\ell(\OPT)$, where $\OPT$ denotes an optimal solution of the problem.

While it is easy to observe that for a shortest spanning tree $T$, we have $\ell(T)\leqslant\ell(\OPT)$, bounds on $\ell(J)$ for a shortest $\joinset_T$-join $J$ can for example be obtained by exploiting polyhedral descriptions of $\joinset$-joins.
In particular, the dominant of the $\joinset$-join polytope\footnote{The dominant of the $\joinset$-join polytope is the set of all points $x\in\mathbb{R}^E$ such that there is a convex combination $y = \sum_{i=1}^k \lambda_i\chi^{J_i}$ of characteristic vectors $\chi^{J_i}\in\{0, 1\}^E$ of $\joinset$-joins $J_i$ with $\lambda_i>0$ for $i\in[k]$ such that $y\leqslant x$.} is given by
\begin{equation*}
\TJdom{\joinset}\coloneqq\left\{ x\in\mathbb{R}^E_{\geqslant 0} \,\middle|\, x(\delta(C))\geqslant 1\ \text{$\forall$\,$\joinset$-cuts $C\subseteq V$}  \right\}\enspace,
\end{equation*}
where a \emph{$\joinset$-cut} $C$ is a subset of $V$ with $|C\cap \joinset|$ odd (see~\cite[Section 29]{schrijver2003combinatorial}).
By integrality of this polytope, in order to prove a bound of the form $\ell(J)\leqslant \gamma\cdot\ell(\OPT)$ for a shortest $\joinset_T$-join $J$, it is sufficient to find a point $z\in\TJdom{\joinset_T}$ with $\ell^\top z\leqslant\gamma\cdot\ell(\OPT)$, and a $(1+\gamma)$-approximation follows.

In many approaches in this context, an important role for finding a suitable point $z$ is taken by a linear relaxation of the problem.
For \MSCJ[\joinset], we use the formulation
\begin{equation}\tag{$\mathrm{LP}_{\mathrm{HK}}$}\label{eq:QJoinRelaxation}
\begin{aligned}
\min \quad \ell^\top x \\
x(\delta(C))&\geqslant 2 && \forall C\subsetneq V,\ C\neq\emptyset,\ |C\cap \joinset|\ \text{even}\\
x(\delta(C))&\geqslant 1 && \forall C\subseteq V,\ |C\cap \joinset|\ \text{odd}\\
x &\in \mathrlap{\mathbb{R}^E_{\geqslant 0} \enspace,}
\end{aligned}
\end{equation}%
which is an adaptation of the well-known Held-Karp relaxation for \TSP.

If $x^*$ is an optimal solution of~\eqref{eq:QJoinRelaxation}, the point $z=\sfrac{x^*}{2}$ is a good candidate for a feasible point of $\TJdom{\joinset_T}$.
More precisely, observe that the constraint $z(\delta(C))\geqslant 1$ is violated precisely for cuts $C$ with $x^*(\delta(C))<2$.
As $\delta(C)=\delta(V\setminus C)$, we can fix a vertex $\oddelement \in \joinset$, and it is enough to consider the cuts in the family
\begin{equation*}
\mathcal{N}\coloneqq\{C\subseteq V\mid \oddelement \notin C,\ x^*(\delta(C))<2\}\enspace,
\end{equation*}
the so-called \emph{narrow} cuts of $x^*$.
By the first constraint in~\eqref{eq:QJoinRelaxation}, only cuts $C$ with $|C\cap \joinset|$ odd can be narrow.
Moreover, note that the constraints $z(\delta(C))\geqslant 1$ only appear in the description of $\TJdom{\joinset_T}$ if $C$ is a $\joinset_T$-cut, but this is not necessarily the case for all $C\in\mathcal{N}$.
If indeed, none of the narrow cuts are $\joinset_{T}$-cuts, we conclude that $z=\sfrac{x^*}{2}$ is feasible for $\TJdom{\joinset_T}$.
Using that $\ell^\top x^*\leqslant\ell(\OPT)$, we get
\begin{equation*}
\ell(T)+\ell(J)\leqslant \sfrac32\cdot\ell(\OPT)\enspace,
\end{equation*}
and hence a $\sfrac32$-approximation.
In particular, if we could obtain an optimal solution $x^*$ of~\eqref{eq:QJoinRelaxation} and a tree $T^*$ with $\ell(T^*)\leqslant\ell(\OPT)$ that has an odd number of edges in every narrow cut of $x^*$, we could achieve the above result.
To see this, consider a narrow cut $C$, and observe that
\begin{equation}\label{eq:degreeCount}
\sum_{v\in C} \deg_{T^*}(v) = 2\cdot |T^*\cap E[C]| + |T^*\cap\delta(C)|\enspace.
\end{equation}
The assumption that $|T^*\cap\delta(C)|$ is odd implies that $T^*$ has an odd number of odd-degree vertices in $C$.
As $|C\cap \joinset|$ is odd, we conclude that $|C\cap \joinset_{T^*}| = |C\cap (\odd(T^*)\symdiff \joinset)|$ is even, i.e., $C$ is indeed not a $\joinset_{T^*}$-cut.

Ideally, we would thus like to find a spanning tree $T^*$ that has an odd number of edges in each narrow cut $C\in\mathcal{N}$.
\textcite{cheriyan_2015_approximating} showed that the family $\mathcal{N}$ is in fact a laminar family.
We additionally observe that
\begin{equation}\label{eq:width_N}
\operatorname{width}(\mathcal{N}) \leq |\joinset| - 1 \enspace.
\end{equation}
Indeed, every cut $C\in\mathcal{N}$ has odd intersection with $\joinset$, i.e., it contains at least one element of $\joinset\setminus\{r\}$, so there can be at most $|\joinset|-1$ many mutually disjoint sets in $\mathcal{N}$.
However, even for a chain, one can see that it is $\NP$-hard to decide whether there is a spanning tree that is odd in each cut.\footnote{
$\NP$-hardness can for example be derived by a reduction from the Hamiltonian $s$-$t$~path problem.
To this end, consider an arbitrary numbering of the vertices $V=\{v_1,\ldots, v_n\}$ with $v_1=s$ and $v_n=t$, and consider the complete chain $S_1,\ldots, S_{n-1}$, where $S_i=\{v_1,\ldots, v_i\}$ for $i\in [n-1]$.
First, the cuts $S_1$ and $S_{n-1}$ ensure that a spanning tree that is odd in each cut must have odd degree at $v_1$ and $v_n$.
Moreover, for any $i\in\{2,\ldots, n-1\}$, the vertex $v_i$ must have even degree in the spanning tree due to the cut constraints on $S_{i-1}$ and $S_i$.
Because the degrees of a spanning tree on $n$ vertices sum up to $2(n-1)$, this implies that $v_1$ and $v_n$ must have degree $1$, and all other vertices degree $2$.
However, a spanning tree with these properties is a Hamiltonian $v_1$-$v_n$~path, and any Hamiltonian $v_1$-$v_n$~path is such a spanning tree.
}
As an alternative, we can also use a tree $T$ with slightly weaker properties, and instead modify the vector $z$ to obtain a feasible point for $\TJdom{\joinset_T}$.
In particular, observe the following.

\begin{observation}\label{obs:suffCond}
Assume that we are given an optimal solution $x^*$ of~\eqref{eq:QJoinRelaxation}, a tree $T$, and a point $y\in\mathbb{R}^E_{\geqslant 0}$ with the following properties.
\begin{enumerate}
\item\label{item:observationLarge} For all narrow cuts $C$ of $x^*$, either $|T\cap\delta(C)|$ is odd, or $y(\delta(C))\geqslant \sfrac1\varepsilon$.
\item\label{item:observationBounds} $\ell^\top y \leqslant \ell(\OPT)$, and $\ell(T)\leqslant \ell(\OPT)$.
\end{enumerate}
Then, shortcutting the multiunion of $T$ and a shortest $\joinset_T$-join $J$ gives a $(1.5+\varepsilon)$-approximate solution for \MSCJ[\joinset].
\end{observation}

\begin{proof}
We claim that $z\coloneqq\frac12(x^*+\varepsilon y)\in\mathbb{R}_{\geqslant 0}^E$ satisfies $z\in\TJdom{\joinset_T}$.
From the above discussion and the bounds in \cref{item:observationBounds}, this immediately implies the observation.
The constraints of the form $z(\delta(C))\geqslant 1$ are obviously satisfied for non-narrow cuts $C$ of $x^*$, and we saw through \eqref{eq:degreeCount} that they do not appear in the description of $\TJdom{\joinset_T}$ for narrow cuts $C$ if $|T\cap\delta(C)|$ is odd.
For the remaining narrow cuts $C$, \cref{item:observationLarge} implies $y(\delta(C))\geqslant \sfrac1\varepsilon$, and hence $z(\delta(C))\geqslant \frac12(1+\varepsilon\cdot\sfrac1\varepsilon)=1$.
\end{proof}

The close relation of $y$ and $T$ that is required in \cref{obs:suffCond} motivates studying $\tau$-odd solutions, which to some extent embrace properties of $y$ and $T$.

\begin{definition}[$\tau$-odd]
For $\tau\in\mathbb{Z}_{\geqslant0}$ such that $\tau$ is odd, and a family $\mathcal{F}$ of cuts, we say that a point $y\in\mathbb{R}^E$ is \emph{$\tau$-odd} (with respect to $\mathcal{F}$), if for each $C\in\mathcal{F}$, either
\begin{enumerate}
\item\label{item:oddSmallCut} $y(\delta(C))\leqslant \tau$, $y$ is integral on the edges in $\delta(C)$, and $y(\delta(C))$ is odd, or
\item\label{item:oddLargeCut} $y(\delta(C))\geqslant \tau+2$.
\end{enumerate}
We call the cuts $C$ satisfying \cref{item:oddSmallCut,item:oddLargeCut} the \emph{$y$-small} and \emph{$y$-large} cuts, respectively.
\end{definition}

Note that, for $\tau\approx\sfrac1\varepsilon$, given a short $\tau$-odd point $y\in\PST$ with respect to the narrow cuts $\mathcal{N}$ of an optimal solution $x^*$ of~\eqref{eq:QJoinRelaxation}, a tree with the properties needed in \cref{obs:suffCond} could be obtained as a minimum length spanning tree $T$ such that $\chi^T$ coincides with $y$ on the integral edges of $y$.
This reduces the problem to finding short $\tau$-odd points, which is where our dynamic programming approach can help.

\subsection[Obtaining \texorpdfstring{$\tau$}{tau}-odd points via our DP]{Obtaining \boldmath\texorpdfstring{$\tau$}{tau}\unboldmath-odd points via our DP}

We now discuss how our dynamic programming approach can be adjusted to compute $\tau$-odd points.
Note that $\tau$-odd points are by definition very similar to $\tau$-integral points: The difference is that a $\tau$-odd point can only have an odd number of edges in small cuts.
Our dynamic programming approach can easily handle this type of constraints, as edges in small cuts are always determined by the connectivity triples used.
Thus, if we define the set
\begin{equation*}
\mathcal{K}'\coloneqq\big\{\sfc\in\mathcal{K} \,\big|\, \text{$|F|$ odd} \big\}
\cup\big\{(\emptyset,\emptyset,\{\emptyset\}), (V,\emptyset,\{\emptyset\})\big\}
\end{equation*}
and run the dynamic program presented in \cref{sec:approachDPMCCST} with $\mathcal{K}'$ instead of $\mathcal{K}$ (and no lower or upper bounds on the cuts), we immediately obtain the following analogue of \cref{thm:MCCSTdPguarantee}.

\begin{theorem}\label{thm:tauOddDP}
Let $\mathcal{C}$ be a chain of cuts.
For any $\tau\in\mathbb{Z}_{\geqslant 0}$, there is an algorithm that returns in $|V|^{\Oh(\tau)}$ time a $\tau$-odd point $y\in\PST$ with respect to $\mathcal{C}$ such that $c^\top y\leqslant c(T)$ for every spanning tree $T$ that has an odd number of edges in every cut in $\mathcal{C}$.
\end{theorem}

Before continuing on \MSCJ[\joinset], we show how \cref{thm:tauOddDP} readily implies a $(1.5+\varepsilon)$-approximation for \pathTSP, replicating prior results for this problem.
(We recall that the currently best approximation for \pathTSP has a factor slightly below $1.5$~\cite{karlin_2021_slightly,karlin_2022_deterministic}; nevertheless, \pathTSP allows for nicely exemplifying our techniques in a simpler setting.)
For completeness, we recall that \pathTSP is formally defined as follows.
\begin{mdframed}[leftmargin=0.025\linewidth,rightmargin=0.025\linewidth]%
{\textbf{Shortest Hamiltonian {\boldmath$s$-$t$\unboldmath} Path Problem (\labeltarget{prb:pathTSP}{\pathTSP}):}}
Let $G=(V,E)$ be a complete graph with metric edge lengths $\ell\colon E\to\mathbb{R}_{\geqslant 0}$, and let $s,t\in V$ be two distinct vertices.
Find a path $P\subseteq E$ minimizing $\ell(P)\coloneqq\sum_{e\in P}\ell(e)$ among all Hamiltonian $s$-$t$ paths in $G$.
\end{mdframed}
We recall that \pathTSP is indeed a special case of \MSCJ[\joinset] by choosing $\joinset=\{s,t\}$.
This follows from \cref{thm:optQJoins}, because a spanning tree that is a connected $\{s,t\}$-join is a Hamiltonian $s$-$t$~path.

In the special case of \pathTSP, the family $\mathcal{N}$ of narrow cuts has width $1$, i.e., it is a chain.
This was observed by \textcite{an_2015_improving} and has been crucial in many recent improvements for \pathTSP.
In our case, it guarantees applicability of \cref{thm:tauOddDP}, and shows that \cref{alg:pathTSP} is well-defined, which, as we discuss now, is a $(1.5+\varepsilon)$-approximation for \pathTSP.

\begin{algorithm2e}[H]
\caption{Polynomial time $(\sfrac32+\varepsilon)$-approximation for Path TSP\strut}\label{alg:pathTSP}
\begin{algostepsarabic}%

\item\label{algostep:narrowChain} Let $x^*$ be an optimal solution of~\eqref{eq:QJoinRelaxation}\strut{}, and let $\mathcal{C}$ be the family of narrow cuts of $x^*$ not containing $t$.

\item\label{algostep:tauOddChain} Let $\tau\in\{\lfloor\sfrac1\varepsilon\rfloor,\lfloor\sfrac1\varepsilon\rfloor+1\}$ odd, and use the algorithm guaranteed by \cref{thm:tauOddDP} to find a $\tau$-odd point $y\in\PST$ with respect to $\mathcal{C}$.

\item\label{algostep:getTpathTSP} Let $T$ be the shortest spanning tree of $G$ such that $\chi^T$ coincides with $y$ on all integral edges of $y$.

\item\label{algostep:findJoin} Let $J$ be a shortest $\joinset_T$-join in $G$, and return the shortcutted multiunion of $J$ and $T$.
\end{algostepsarabic}
\end{algorithm2e}

\begin{proposition}\label{prop:pathTSP}
\cref{alg:pathTSP} is a $(1.5+\varepsilon)$-approximation for \pathTSP.
\end{proposition}
\begin{proof}
We claim that the pair $(y,T)$ generated in \cref{alg:pathTSP} has the properties listed in \cref{obs:suffCond}.
To see that \cref{item:observationLarge} holds, first note that by $\tau$-oddness of $y$ with respect to $\mathcal{C}$, we have that for every narrow cut $C$ of $x^*$, either $y(\delta(C))\geqslant \tau+2\geqslant \sfrac1\varepsilon$, or $y$ is integral on $\delta(C)$ and $y(\delta(C))$ is odd.
As $\chi^T$ coincides with $y$ on all integral edges of $y$, the latter case in particular implies that $\chi^T$ coincides with $y$ on $\delta(C)$, and consequently, $|T\cap \delta(C)|$ is odd.
Hence, \cref{item:observationLarge} is satisfied.
For \cref{item:observationBounds}, observe that by definition, $c(T)\leqslant c^\top y$, and \cref{thm:tauOddDP} implies $c^\top y\leqslant c(\OPT)$.
Together, this yields $c(T)\leqslant c(\OPT)$, as desired.

By \cref{obs:suffCond}, it follows that \cref{alg:pathTSP} is a $(1.5+\varepsilon)$-approximation for \pathTSP.
For a running time guarantee, we observe that all steps in \cref{alg:pathTSP} can be performed efficiently.
For \cref{algostep:narrowChain}, note that the minimum cut in $G$ with respect to weights $x^*$ has value at least $1$ by the lower bounds in the relaxation~\eqref{eq:QJoinRelaxation}, and enumerating all cuts with values that are within a factor of $2$ from the minimum cut can be done in time $\Oh(|V|^4|E|)$ (see~\cite{nagamochi1997computing}).
\cref{algostep:tauOddChain} can be done in time $|V|^{\Oh(\sfrac{1}{\varepsilon})}$ by \cref{thm:tauOddDP}.
Finally, the running time of finding minimum $\joinset$-joins in \cref{algostep:findJoin} is negligible compared to the running time of the second step of the algorithm.
(For an efficient algorithm to find minimum $\joinset$-joins, see~\cite{edmonds1973matching}, for example.)
Thus, the running time of \cref{alg:pathTSP} is $|V|^{\Oh(\sfrac{1}{\varepsilon})}$.
\end{proof}%

The above approach exploits that the $\tau$-odd point $y$ lies in the spanning tree polytope: This guarantees that a spanning tree $T$ coinciding with $y$ on integral edges can be found in \cref{algostep:getTpathTSP} of \cref{alg:pathTSP}.
In the more general case of connected $\joinset$-joins with $|\joinset|>2$, the narrow cuts $\mathcal{N}$ no longer form a chain, which causes additional challenges in directly extending the propagation step of our dynamic programming approach, as we highlighted in the generalization from \MCCST to \MLCST.
With minor modifications, ideas of this generalization also help for \MSCJ[\joinset].
In fact, a simple alteration of the extension step can be used to obtain a pair $(y,T)$ with the properties highlighted in the following theorem.

\begin{theorem}\label{thm:QJdPguarantee}
For any odd $\tau\in\mathbb{Z}_{\geqslant 0}$ and a laminar family $\mathcal{L}\subseteq 2^V$ of width $k$, there is an algorithm that returns in time $|V|^{\Oh(k\tau)}$  a point $y\in\mathbb{R}^E_{\geqslant 0}$ and a spanning tree $T$ of $G$ with the following properties:
\begin{enumerate}
\item $y$ is $\tau$-odd with respect to $\mathcal{L}$.
\item $\chi^T$ coincides with $y$ on all $y$-small cuts in $\mathcal{L}$.
\item $\ell(T)\leqslant \ell^\top y \leqslant \ell^\top x$ for every $\tau$-odd point $x\in\PST\cap\mathbb{Z}^E_{\geqslant 0}$.
\end{enumerate}
\end{theorem}

Observe that \cref{thm:QJdPguarantee} allows for bounding from above the length of the $\tau$-odd points $y$ by the length of $\tau$-odd spanning trees only, but by \cref{thm:optQJoins}, this is sufficient for comparing to optimal solutions of \MSCJ[\joinset], once we prove that connected $\joinset$-joins that are spanning trees are indeed $\tau$-odd.
From the above ingredients, we obtain our approximation algorithm for \MSCJ[\joinset], which is stated as \cref{alg:Qjoin} below.
We show that this algorithm implies \cref{thm:TJoins}.
Note that, when the algorithm invokes \cref{thm:QJdPguarantee}, it obtains a pair $(y,T)$, but only uses the spanning tree $T$.
The point $y$ is used only in the analysis of \cref{alg:Qjoin}.

\begin{algorithm2e}[H]
\caption{Polynomial $(\sfrac32+\varepsilon)$-approximation for \protect\MSCJ[\joinset]\strut}\label{alg:Qjoin}
\begin{algostepsarabic}

\item\label{algostep:narrowCuts} Let $x^*$ be an optimal solution of~\eqref{eq:QJoinRelaxation}\strut{}, and let $\mathcal{N}$ be the family of all narrow cuts of $x^*$ not containing a fixed element $\oddelement \in \joinset$.

\item\label{algostep:tauOdd} Let $\tau\in\{\lfloor\sfrac1\varepsilon\rfloor,\lfloor\sfrac1\varepsilon\rfloor+1\}$ be odd, and use the algorithm guaranteed by \cref{thm:QJdPguarantee} to find a pair $(y,T)$ of a $\tau$-odd point $y\in\mathbb{R}^E_{\geqslant 0}$ with respect to $\mathcal{N}$ and a spanning tree $T$ of $G$.

\item\label{algostep:mscjJoin} Let $J$ be a shortest $\joinset_T$-join in $G$, and return the shortcutted multiunion of $J$ and $T$.
\end{algostepsarabic}
\end{algorithm2e}

\begin{proof}[Proof of \cref{thm:TJoins}]
As in the proof of \cref{prop:pathTSP}, it is easy to see that the pair $(y,T)$ returned by \cref{alg:Qjoin} satisfies \cref{item:observationLarge} in \cref{obs:suffCond}.
To prove that \cref{item:observationBounds} holds, we claim that for any spanning tree $S$ that is a $\joinset$-join, $\chi^S$ is a $\tau$-odd solution of $\PST$ with respect to $\mathcal{N}$.
By \cref{thm:QJdPguarantee}, this implies that $\ell^\top y\leqslant \ell(S)$ for all spanning trees $S$ that are $\joinset$-joins.
By \cref{thm:optQJoins}, at least one such spanning tree is in fact an optimal solution to \MSCJ[\joinset], and thus $\ell^\top y\leqslant \ell(\OPT)$ follows.
Furthermore, \cref{thm:QJdPguarantee} also implies $\ell(T)\leqslant\ell^\top y\leqslant\ell(\OPT)$.

For concluding the approximation guarantee of $\sfrac32+\varepsilon$ with the help of \cref{obs:suffCond}, it is thus sufficient to prove the claim.
Thereto, consider a spanning tree $S$ that is a $\joinset$-join, and let $C\in\mathcal{N}$.
We show that $\chi^S(\delta(C))$ is odd, which implies the claim due to integrality of $\chi^S$.
Double counting the edges in $S\cap E[C]$, we get
\begin{equation*}
\chi^S(\delta(C)) = |S\cap\delta(C)| = \sum_{v\in C} \deg_{S}(v) - 2\cdot|S\cap E[C]|\enspace.
\end{equation*}
Consequently, $\chi^S(\delta(C))$ has the same parity as the number of odd-degree vertices in $C$ with respect to $S$.
Because $S$ is a $\joinset$-join, the latter number is equal to $|C\cap \joinset|$, and as we already observed previously, all $C\in\mathcal{N}$ satisfy that $|C\cap \joinset|$ is odd.
This finishes the proof of the claim.

Finally, in order to obtain a running time bound, note that \cref{algostep:narrowCuts} can be performed in time $\Oh(|V|^4|E|)$ (see~\cite{nagamochi1997computing}; we note that this step is analogous to \cref{algostep:narrowChain} of \cref{alg:pathTSP}, and hence can be analyzed as in the proof of \cref{prop:pathTSP}).
The second step has running time $|V|^{\Oh(\sfrac{k}{\varepsilon})}$ by \cref{thm:QJdPguarantee}, where $k$ is the width of the laminar family $\mathcal{N}$, which is at most $k-1$ by \eqref{eq:width_N}.
Moreover, the running time of finding minimum $\joinset$-joins in \cref{algostep:mscjJoin} is negligible compared to the running time of the second step of the algorithm.
(For an efficient algorithm to find minimum $\joinset$-joins, see~\cite{edmonds1973matching}, for example.)
Thus, the running time of \cref{alg:Qjoin} is bounded by $|V|^{\Oh(\sfrac{|\joinset|}{\varepsilon})}$.
This completes the proof of \cref{thm:TJoins}.
\end{proof}

It thus remains to give a proof of \cref{thm:QJdPguarantee}, which we outline in the remainder of this section.
We adopt the dynamic programming approach used for \MLCST, where in order to make sure that the DP guesses an odd number of edges in small cuts, we use $\mathcal{K}'$ instead of $\mathcal{K}$ throughout the procedure.
For every connectivity triple $\sfc\in\mathcal{K}'$, the dynamic programming approach will construct a pair $(y_{\sfc},T_{\sfc})\in\mathbb{R}_{\geqslant 0}^E\times 2^E$ with the following property.

\begin{property}\label{prop:TjPartialProps}\leavevmode
\begin{enumerate}
\item\label{propitem:yAlphaOdd} $\supp(y_{\sfc})\subseteq E[S]$, and
$y_{\sfc}$ is $\tau$-odd with respect to $\mathcal{L}_S$.
\item\label{propitem:Tgood} $\chi^{T_{\sfc}}\in\PST\sfc$, and
$\chi^{T_{\sfc}}$ coincides with $y_{\sfc}$ on all $y_{\sfc}$-small cuts in $\mathcal{L}_S$.
\item\label{propitem:cheapPair} For all $U\subseteq E[S]$ such that $\chi^U\in\PST\sfc$ and $\chi^U+\chi^F$ is $\tau$-odd with respect to $\mathcal{L}_S$, we have $\ell(T_{\sfc})\leqslant\ell^\top y_{\sfc}\leqslant \ell(U)$.
\end{enumerate}
\end{property}

It is clear that if we can construct such pairs for all triples $\sfc\in\mathcal{K}'$, then we are done, as $(y_{(V,\emptyset,\{\emptyset\})},T_{(V,\emptyset,\{\emptyset\})})$ satisfies the assumptions of \cref{obs:suffCond}.
To maintain pairs $(y,T)$ in the dynamic program, we replace the extension step (\cref{alg:extensionLaminarDP}) by the one presented in \cref{alg:extensionQjoins} below.

\begin{algorithm2e}[H]
\caption{Extending $(y_{\sfc[1]}, T_{\sfc[1]})$, $\ldots$\,, $(y_{\sfc[w]}, T_{\sfc[w]})$ to $\sfc$.\strut{}}\label{alg:extensionQjoins}

\begin{algostepsarabic}
\item Let $z$ be a minimizer of the linear program\strut{}
\begin{align}
\min \quad  \ell^\top z \tag{$\text{lamExLP}_{2}$}\label{eq:lamExLP2} \\[-0.3em]
z &\in\PST\sfc\nonumber\\
z(\delta(S'))+|F\cap\delta(S')| &\geqslant \tau+2 && \forall S\in\mathcal{L}\text{ such that } \exists i\in[w] \text{ with } S_i\subsetneq S'\subsetneq S \nonumber\\
z(e) &= \chi^{T_{\sfc[i]}}(e) && \forall e \in E[S_i],\ \forall i\in[w]\nonumber\\
z(e) &= \chi^{F_i}(e) && \forall e \in \delta(S_i)\setminus F,\ \forall i\in[w]\enspace.\nonumber
\end{align}

\item\label{algstep:combineY} Define $y\in\mathbb{R}^E_{\geqslant 0}$ by
$
y(e) = \begin{cases}
y_{\sfc[i]}(e) & \text{if $e\in E[S_i]$ for some $i\in[w]$,}\\
z(e) & \text{else.}
\end{cases}
$

\item\label{algostep:findT} Let $T\subseteq E[S]$ be of minimum length $\ell(T)$ such that
$$\abovedisplayskip=2pt\belowdisplayskip=0pt\text{
	\begin{enumerate*}[label=(\roman*)]
	\item $\chi^T\in \PST\sfc$,\,
	\item $T\cap E[S_i]=T_{\sfc[i]} \;\forall i\in [w]$, and\,
	\item $T\cap \delta(S_i) = F_i \;\forall i\in [w]$.
	\end{enumerate*}
}
$$

\item Output $(y, T)$.
\end{algostepsarabic}
\end{algorithm2e}

The idea of this modified extension step is to maintain both a fractional point and a spanning tree at every stage, where (as in the dynamic program for \MLCST) extension steps using the linear program are always based on trees.
A new fractional point is then obtained by combining the potentially fractional extension with the fractional solutions of the subproblems (see \cref{algstep:combineY}).
Opposed to the situation in \MLCST, the application here only requires spanning trees that have an odd number of edges in small cuts of the corresponding fractional point.
This is easily achieved by the construction of the trees in \cref{algostep:findT}.
This construction guarantees that every pair $(y,T)$ returned by \cref{alg:extensionQjoins} satisfies the first two points of \cref{prop:TjPartialProps} with respect to the connectivity triple $\sfc$, which we formally prove in \cref{lem:pairsSatisfyProperty}.

For propagation, we use \cref{alg:propagationQjoins}, which is an analogon of \cref{alg:propagationLaminarDP}.
It considers all potential candidates for $(y_{\sfc},T_{\sfc})$ that can be obtained through extension steps, and the shortest pair $(y,T)$ with respect to $\ell^\top y$ is returned.

\begin{algorithm2e}[H]
\caption{Propagation to $(y_{\sfc},T_{\sfc})$ from all $(y_{\sfcprime}, T_{\sfcprime})$ with $S'\subsetneq S$.\strut{}}\label{alg:propagationQjoins}

\begin{algostepsarabic}

\item For every choice of triples $\sfc[i]\in\mathcal{K}'$ for $i\in[w]$ where $S_1, \ldots, S_w \subsetneq S$\strut{} have pairwise empty intersections, apply \cref{alg:extensionQjoins} to extend $(y_{\sfc[1]}, T_{\sfc[1]})$,~$\ldots\,$, $(y_{\sfc[w]}, T_{\sfc[w]})$ to $\sfc$.
Collect all pairs $(y,T)$ obtained this way in the set $\mathcal{P}$.

\item Return $(y_{\sfc},T_{\sfc})\in\argmin_{(y,T)\in\mathcal{P}}\ell^\top y$.
\end{algostepsarabic}
\end{algorithm2e}

We now show that the pairs $(y,T)$ returned by \cref{alg:propagationQjoins} have \cref{prop:TjPartialProps}.
Adapting the approach pursued in the case of \MLCST (\cref{lem:treesSatisfyProperty}), we proceed by induction.
Clearly, the pair $(y_{(\emptyset,\emptyset,\{\emptyset\})}, T_{(\emptyset,\emptyset,\{\emptyset\})})=(0,\emptyset)$ satisfies \cref{prop:TjPartialProps}, and the inductive step is given by the following lemma.

\begin{lemma}\label{lem:pairsSatisfyProperty}
Let $\sfc\in\mathcal{K}'$ with $S\neq\emptyset$.
Assume that for all $\sfcprime\in\mathcal{K}'$ with $S'\subsetneq S$, we are given $(y_{\sfcprime},T_{\sfcprime})$ satisfying \cref{prop:TjPartialProps}, and let $(y_{\sfc},T_{\sfc})$ be obtained from \cref{alg:propagationQjoins}.
Then $(y_{\sfc},T_{\sfc})$ satisfies \cref{prop:TjPartialProps}, as well.
\end{lemma}

\begin{proof}
Let $\mathcal{P}$ be defined as in \cref{alg:propagationQjoins}.
We already remarked above that all pairs $(y,T)\in\mathcal{P}$ satisfy \cref{propitem:yAlphaOdd,propitem:Tgood} of \cref{prop:TjPartialProps} with respect to the connectivity triple $\sfc$.
Indeed, $\supp(y_{\sfc})\subseteq E[S]$ holds by definition, as well as $\chi^{T_{\sfc}}\in\PST\sfc$.
Additionally, the facts that $y_{\sfc}$ is $\tau$-odd with respect to $\mathcal{L}_{S}$ and that $\chi^{T_{\sfc}}$ coincides with $y_{\sfc}$ on all $y_{\sfc}$-small cuts in $\mathcal{L}_{S}$ are true by definition for cuts $L\in\mathcal{L}_{S}$ with $S_i\subsetneq L$ for some $i\in[w]$; and they are implied by the assumptions on $(y_{\sfc[i]},T_{\sfc[i]})$ through \cref{prop:TjPartialProps} for cuts $L\in\mathcal{L}_{S}$ with $L\subsetneq S_i$ for some $i\in[w]$.

Moreover, for any pair $(y,T)\in\mathcal{P}$, we also have $\ell(T)\leqslant\ell^\top y$: If $z$ is the solution of the linear program~\eqref{eq:lamExLP2} that was used to define $y$, then $\ell(T)\leqslant \ell^\top z$ because $\chi^T$ is in fact an optimal solution of the same linear program without the constraints $z(\delta(S))\geqslant\tau+2$.
As by assumption, $\ell(T_{\sfc[i]})\leqslant \ell^\top y_{\sfc[i]}$, we further see that $\ell^\top z\leqslant \ell^\top y$, and hence $\ell(T)\leqslant\ell^\top y$, which is the first statement in \cref{propitem:cheapPair}.

Thus, it remains to prove that there exists a pair $(y,T)\in\mathcal{P}$ such that $\ell^\top y\leqslant \ell(U)$ for every set $U\subseteq E[S]$ such that $\chi^{U}\in\PST\sfc$ and $\chi^{U}+\chi^{F}$ is $\tau$-odd with respect to $\mathcal{L}_{S}$.
Fix such a set $U$, and let $R\subseteq\binom{V\setminus S}{2}$ be any set of edges that is right-compatible with $\sfc$.
Then, $T\coloneqq U\cup F\cup R$ is in the spanning tree polytope of $(V,E\cup R)$.
Let $S_1,\ldots,S_w\subsetneq S$ be the maximal $\chi^T$-small cuts in $\mathcal{L}_{S}$, and let $F_i=T\cap\delta(S_i)$.
We define connectivity patterns $\mathcal{C}_1,\ldots,\mathcal{C}_w$ such that $\sfc[i]\in\mathcal{K}'$ iteratively as follows, where $T_0\coloneqq T$ and $i\in[w]$.
\begin{align}\label{eq:iterativeC_Qjoins}
\text{Let $\mathcal{C}_{i}$ s.t.~$T_{i-1}$ is compatible with $\sfc[i]$, and let $T_i\coloneqq \big(T_{i-1}\setminus E[S_i]\big)\cup T_{\sfc[i]}$.}
\end{align}
Assume that we extend $(y_{\sfc[1]},T_{\sfc[1]}),\ldots,(y_{\sfc[w]},T_{\sfc[w]})$ to $\sfc$ using \cref{alg:extensionQjoins}.
If the algorithm returns the pair $(y,T)$, then $\ell^\top y\leqslant \ell(U)$.
To this end, observe that if $z$ is the solution of the linear program~\eqref{eq:lamExLP2} used in this call to \cref{alg:extensionQjoins}, then we can write
\begin{equation}\label{eq:splitY}
\ell^\top y = \sum_{i\in[w]}\ell^\top y_{\sfc[i]} + \ell^\top z-\sum_{i\in[w]}\ell(T_{\sfc[i]})\enspace.
\end{equation}
We will bound the right-hand side by $\ell(U)$.
Thereto, we claim that by the construction in~\eqref{eq:iterativeC_Qjoins}, we have $\ell^\top y_{\sfc[i]}\leqslant \ell(U\cap E[S_i])$.
This follows from invoking \cref{propitem:cheapPair} of \cref{prop:TjPartialProps} for the pair $(y_{\sfc[i]},T_{\sfc[i]})$ (which is satisfied by assumption) with the edge set $U\cap E[S_i]$.
To this end, we have to show that \begin{enumerate*}[label=(\alph*)]\item\label{inlineitem:UinPST} $\chi^{U\cap E[S_i]}\in\PST\sfc[i]$, and \item\label{inlineitem:UandFOdd} $\chi^{U\cap E[S_i]}+\chi^{F_i}$ is $\tau$-odd with respect to $\mathcal{L}_{S_i}$.\end{enumerate*} Indeed, \cref{inlineitem:UinPST} follows from the fact that $T_{i-1}$ is compatible with $\sfc[i]$, hence $\chi^{T_{i-1}\cap E[S_i]}\in\PST\sfc[i]$, and $T_{i-1}\cap E[S_i]=U\cap E[S_i]$ by construction; \cref{inlineitem:UandFOdd} follows from $\chi^U+\chi^{F}$ being $\tau$-odd with respect to $\mathcal{L}_{S}$, as $S_i\subsetneq S$.
Furthermore, note that $T_w$ is compatible with $\sfc$, hence $T_w\cap E[S]$ is left-compatible with $\sfc$, i.e., $\chi^{T_w\cap E[S]}\in\PST\sfc$.
Moreover, we can write
\begin{equation*}\textstyle
T_w\cap E[S]=\left(U\setminus\bigcup_{i\in[w]} E[S_i]\right)\cup\textstyle\,\bigcup_{i\in[w]}T_{\sfc[i]}\enspace,
\end{equation*}
which, together with the fact that $S_1,\ldots,S_w$ are maximal $\chi^T$-small cuts in $\mathcal{L}_{S}$, implies that $T_w\cap E[S]$ is an integral solution of~\eqref{eq:lamExLP2}.
As $z$ is an optimal fractional solution of the same linear program, we have $\ell^\top z\leqslant \ell(T_w\cap E[S])$.
Applying the inequalities just obtained to~\eqref{eq:splitY}, we get
\begin{align*}
\ell^\top y &\leqslant \sum_{i\in[w]} \ell(U\cap E[S_i]) + \ell(T_w\cap E[S]) - \sum_{i\in[w]}\ell(T_{\sfc[i]})\\
&= \sum_{i\in[w]} \ell(U\cap E[S_i]) + \ell\left(U\setminus{\textstyle\bigcup_{i\in[w]} E[S_i]}\right)\\
&=\ell(U)\enspace,
\end{align*}
as desired.
\end{proof}

In order to bound the running time of our dynamic program, we need an upper bound on the number of connectivity triples in $\mathcal{K}'$, but this is easily obtained from the upper bound $|\mathcal{K}|\leqslant |V|^{\Oh(\tau)}$ in \cref{prop:sizeK}: We have $\mathcal{K}'\subseteq\mathcal{K}$, and thus also $|\mathcal{K}'|\leqslant |V|^{\Oh(\tau)}$.
With these ingredients, we are finally ready to formally prove \cref{thm:QJdPguarantee}.

\begin{proof}[Proof of \cref{thm:QJdPguarantee}]
We run a dynamic program that calculates pairs $(y_{\sfc},T_{\sfc})$ for all $\sfc\in\mathcal{K}'$ starting with the initialization $(y_{(\emptyset,\emptyset,\{\emptyset\})},T_{(\emptyset,\emptyset,\{\emptyset\})})=(0,\emptyset)$, and using \cref{alg:propagationQjoins} for propagation in an order such that $(y_{\sfcprime},T_{\sfcprime})$ is computed before $(y_{\sfc},T_{\sfc})$ if $S'\subsetneq S$.
Note that $(y_{(\emptyset,\emptyset,\{\emptyset\})},T_{(\emptyset,\emptyset,\{\emptyset\})})$ satisfies \cref{prop:TjPartialProps}, and hence by an inductive application of \cref{lem:pairsSatisfyProperty}, all pairs $(y_{\sfc},T_{\sfc})$ satisfy \cref{prop:TjPartialProps}.
In particular, $(y_{(V,\emptyset,\{\emptyset\})},T_{(V,\emptyset,\{\emptyset\})})$ is thus a pair satisfying the guarantees of \cref{thm:QJdPguarantee}.

In terms of running time, the dominating operation is repeatedly solving linear programs of the type~\eqref{eq:lamExLP2}.
The total number of linear programs that we have to solve during this procedure is bounded from above by $|\mathcal{K}'|^{k+1}$, and the running time of $|V|^{\Oh(k\tau)}$ thus follows from $\mathcal{K}'\subseteq\mathcal{K}$ and \cref{prop:sizeK}, as remarked above, and the fact that linear programs of the type~\eqref{eq:lamExLP2} can be solved in strongly polynomial time by using a compact extended formulation for the spanning tree polytope with small coefficients in the constraint matrix (one can, for example, use the one by \textcite{martin1991using}, which has coefficients that are bounded by $1$ in absolute value), and then applying the framework of \textcite{tardos_1986_strongly}.
\end{proof}

\begingroup
\setlength{\emergencystretch}{3em}
\hbadness 2772\relax
\printbibliography
\endgroup

\appendix
\section{Weakness of the natural relaxation}\label{sec:relaxationWeak}

In this section, we demonstrate that the natural relaxation of \MCCST, which is given by
\begin{equation*}
Q = \left\{ x\in \PST \,\middle|\, a_i \leqslant x(\delta(S_i)) \leqslant b_i \;\forall i\in [k] \right\}\enspace,
\end{equation*}
is too weak for allowing small bounds on constraint violation when comparing an integral solution to the optimal value of the (fractional) relaxation.
More precisely, we show the following.

\begin{theorem}\label{thm:relaxationWeak}
For every $\varepsilon>0$, there is an instance of \MCCST and a point $y\in Q$ such that for any tree $T$ satisfying the chain constraints, there is a cut $C$ in the chain such that
\begin{equation*}\label{eq:relaxationWeak}
|T\cap \delta(C)|\geqslant (2-\varepsilon)\cdot y(\delta(C))\enspace.
\end{equation*}
\end{theorem}

\begin{proof}
We construct a family of instances of the \MCCST problem depending on a parameter $k\in\mathbb{Z}_{>0}$, where each instance has the properties listed in \cref{thm:relaxationWeak}, but with a factor $2-\sfrac{2}{(k+2)}$ instead of $2-\varepsilon$.
This clearly implies the theorem because for a fixed $\varepsilon>0$ and large enough $k$, we obtain an instance with the desired properties.

For $k\in\mathbb{Z}_{>0}$, let $H_k$ be the graph obtained as follows.
Start with a path of length $2^{k-1}$ on vertices $v_0,v_1,\ldots,v_{2^{k-1}}$, and for all $i,j\in\{0,1,\ldots,2^{k-1}\}$ such that $j-i=2^\ell$ for some $\ell\in\mathbb{Z}_{\geqslant 0}$, add a path of length $2$ between $v_i$ and $v_j$.
More precisely, for any such $i$ and $j$, we add a new vertex $w_{i,j}$ as well as edges $\{v_i,w_{ij}\}$ and $\{w_{ij},v_j\}$.
Additionally, we define cuts $S_\ell$ for $\ell\in [2^{k}]$ by
\begin{equation*}
S_\ell \coloneqq \{v_i\mid 2i<\ell\}\cup\{w_{i,j}\mid i+j<\ell\}\enspace.
\end{equation*}
Note that for $\ell_1<\ell_2$, we have $S_{\ell_1}\subsetneq S_{\ell_2}$, and thus the family $\mathcal{S}_k\coloneqq\{S_1,S_2,\ldots,S_{2^{k}}\}$ is a chain.
Finally, we define arbitrary uniform edge costs.
An illustration of this construction for $k=3$ is given in \cref{fig:H_k}.
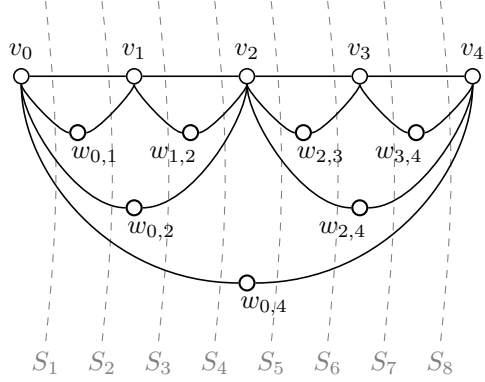
\begin{figure}[ht]
\centering
\begin{tikzpicture}[xscale=1.5]
\small
\pgfmathsetmacro{\k}{3}
\pgfmathsetmacro{\n}{int(2^(\k-1))} %

\foreach \x in {0,...,\n} {
  \node[draw, black, circle, inner sep=2, semithick, label={$v_{\x}$}] (v\x) at (\x,0) {};
}
\foreach \x [evaluate=\x as \y using \x-1] in {1,...,\n} {
  \draw[semithick] (v\y) -- (v\x);
}
\foreach [
  evaluate=\x as \first using int(2*\x-1),
  evaluate=\x as \second using int(2*\x)
] \x in {1,...,\n} {
  \draw[gray, dashed, xshift=-1] (\x-0.75,1) to [bend left=4] (\x-0.75,-\k-.55) node [below] {$S_{\first}$};
  \draw[gray, dashed, xshift=-1] (\x-0.25,1) to [bend left=4] (\x-0.25,-\k-.55) node [below] {$S_{\second}$};
}
\foreach [
    evaluate=\level as \nstep using int(2^(\k-\level)),
    evaluate=\level as \stepsize using int(2^(\level-1))
  ] \level in {1,...,\k}{
  \foreach [
    evaluate=\step as \origin using int((\step-1)*\stepsize),
    evaluate=\step as \target using int(\step*\stepsize)
  ] \step in {1,...,\nstep}{
    \pgfmathsetmacro{\middle}{0.5*(\target+\origin)}
    \pgfmathsetmacro{\ypos}{-\level+0.25}
    \ifodd\step
      \newcommand{\pos}{below right}
    \else
      \newcommand{\pos}{below left}
    \fi
    \node [draw, black, circle, inner sep=2, thick, label={[\pos=5pt and -6pt]:$w_{\origin,\target}$}] (w) at (\middle,\ypos) {};
    \pgfmathsetmacro{\loose}{0.3*\level}
    \draw[semithick] (v\origin) to[out=-90, in=180, looseness=\loose] (w);
    \draw[semithick] (w) to[out=0, in=-90, looseness=\loose] (v\target);
  }
}
\end{tikzpicture}  \caption{The graph $H_k$ and the family $\mathcal{S}_k=\{S_1,\ldots,S_{2^k}\}$ of cuts for $k=3$.}
\label{fig:H_k}
\end{figure}

To complete the instances of the \MCCST problem and to ensure feasibility, we can define $a_\ell=0$ and $b_\ell=k+1$ for all $\ell\in[2^k]$.
Note that this is equivalent to not putting any constraints on the sizes of the given cuts.
In fact our arguments are independent of the precise degree bounds (given feasibility).

We first observe that the corresponding relaxation $Q$ has solutions with small weight on all cuts in $\mathcal{S}_k$.
To this end, define $y\in \mathbb{R}^E_{\geqslant 0}$ by
\begin{equation*}
y(e) = \begin{cases}
1 & \text{if $e=\{v_{i-1},v_{i}\}$ for some $i\in[2^{k-1}]$}\\
\sfrac12 & \text{else}
\end{cases}\enspace.
\end{equation*}
Note that $y$ is indeed a point in $\PST$.
This can be seen by writing $y=\sfrac{\left(x_1+x_2\right)}{2}$, where $x_1,x_2\in\mathbb{R}^E_{\geqslant 0}$ are incidence vectors of spanning trees given by
\begin{equation*}
x_1(e) = \begin{cases}
1 & \text{if $e=\{v_{i-1},v_{i}\}$ for some $i$}\\
1 & \text{if $e=\{v_{i},w_{i,j}\}$ for some $i,j$}\\
0 & \text{if $e=\{w_{i,j},v_j\}$ for some $i,j$}
\end{cases}\quad\text{and}\quad
x_2(e) = \begin{cases}
1 & \text{if $e=\{v_{i-1},v_{i}\}$ for some $i$}\\
0 & \text{if $e=\{v_{i},w_{i,j}\}$ for some $i,j$}\\
1 & \text{if $e=\{w_{i,j},v_j\}$ for some $i,j$}
\end{cases}
\enspace.
\end{equation*}
Moreover, we observe that for every $i\in[2^k]$, we have
\begin{equation}\label{eq:weightFrac}
y(\delta(S_i)) = 1+\frac{k}{2}\enspace.
\end{equation}
The above implies that indeed, $y\in Q$.
Now, consider a spanning tree $T$ of $H_k$.
We claim that for every $k\in\mathbb{Z}_{>0}$, there exists $i\in[2^k]$ such that
\begin{equation}\label{eq:weightTree}
|T\cap\delta(S_i)| \geqslant k\enspace.
\end{equation}
Once we prove this, we can combine~\eqref{eq:weightFrac} and~\eqref{eq:weightTree} to obtain that there exists $i\in[2^k]$ such that
\begin{equation*}
|T\cap \delta(S_i)| \geqslant \frac{k}{1+\sfrac{k}{2}} \cdot y(\delta(S_i)) = \left(2-\frac{2}{k+2}\right)\cdot y(\delta(S_i))\enspace.
\end{equation*}
By choosing $k$ large enough such that $\varepsilon\geqslant\sfrac{2}{(k+2)}$, we thus obtain an instance satisfying the properties listed in \cref{thm:relaxationWeak}.
It remains to prove the claim.
To this end, we show the following stronger lemma.

\begin{adjustwidth}{1em}{}
\begin{lemma}\label{lem:claimInductive}
Let $F$ be a subset of the edges of $H_k$ such that any vertex $w_{i,j}$ of $H_k$ is incident to at least one edge of $F$.
Then there exists $i\in[2^k]$ such that $|F\cap\delta(S_i)|\geqslant k$.
\end{lemma}

\begin{proof}[Proof of \cref{lem:claimInductive}]
We proceed by induction on $k$.
The statement of the base case $k=1$ is directly implied by the assumption on $F$.
Indeed, at least one of the two edges $\{v_0,w_{0,1}\}$ and $\{w_{0,1},v_1\}$ is in $F$, and correspondingly, at least one of $|F\cap\delta(S_1)|\geqslant 1$ or $|F\cap\delta(S_2)|\geqslant 1$ holds.

For the inductive step, let $k\geqslant 2$ and consider the graph $H_k=(V,E)$.
By the assumption on $F$, at least one of the edges $\{v_0,w_{0,2^{k-1}}\}$ and $\{w_{0,2^{k-1}},v_{2^{k-1}}\}$ is in $F$.
By symmetry, we can assume without loss of generality that $\{v_0,w_{0,2^{k-1}}\}\in F$.
Observe that the subgraph of $H_k$ induced by the vertex set $V_0\coloneqq\{v_i\mid i \leqslant 2^{k-2}\}\cup\{w_{i,j}\mid i,j\leqslant 2^{k-2}\}$ is isomorphic to $H_{k-1}$, and $F\cap E[V_0]$ has an edge incident to every vertex $w_{i,j}$ of this copy of $H_{k-1}$.
Hence, by induction, there exists $i\in[2^{k-1}]$ such that $|(F\cap E[V_0])\cap\delta(S_i)|\geqslant k-1$.
As additionally, $\{v_0,w_{0,2^{k-1}}\}\in\delta(S_i)$, this implies $|F\cap\delta(S_i)|\geqslant k$.
\end{proof}%
\end{adjustwidth}
Finally, observe that by connectivity, every spanning tree $T$ of $H_k$ contains at least one edge incident to $w_{i,j}$, for all vertices $w_{i,j}$ of $H_k$.
Consequently, \cref{lem:claimInductive} does indeed imply existence of $i\in[2^k]$ such that $|T\cap\delta(S_i)|\geqslant k$.
This completes the proof of \cref{thm:relaxationWeak}.
\end{proof}
\section[Analyzing the DP by backtracing \texorpdfstring{$\OPT$}{OPT} fails in the general case]{Analyzing the DP by backtracing \boldmath\texorpdfstring{$\OPT$}{OPT}\unboldmath{} fails in the general case}\label{sec:exampleBacktracingOPT}

The aim of this section is to extend an example from \cref{sec:extensionMLCST} that showed why the analysis of our dynamic programming approach cannot be done in a straightforward classical way, i.e., by backtracing an optimal solution.
While the issues in the example constructed in \cref{sec:extensionMLCST} can be fixed by breaking ties in the right way, we now present a slightly more involved instance (see \cref{fig:fullInstance}) where the naive approach faces problems that cannot be avoided easily.

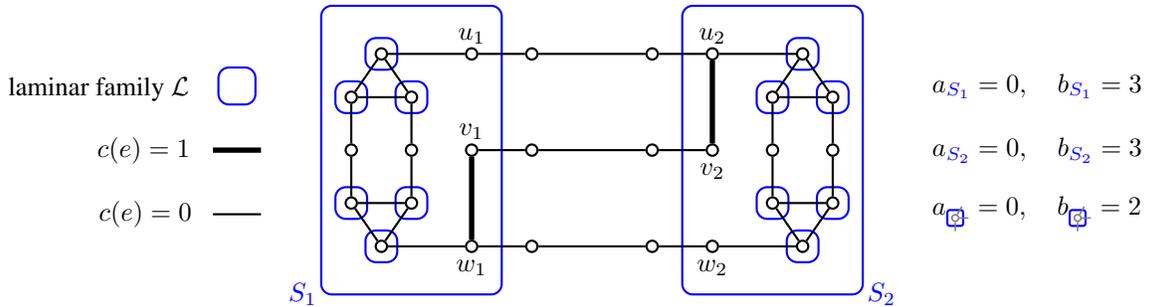
\begin{figure}[ht]
\centering
﻿\begin{tikzpicture}[yscale=1.6, scale=0.8]
\small

\pgfdeclarelayer{background}
\pgfsetlayers{background,main}

\clip (-8,-1.7) rectangle (12,1.7);

\begin{scope}[every node/.style={ns}]
\node (u1) at (0,1) {};
\node (v1) at (0,0) {};
\node (w1) at (0,-1) {};
\node (u1') at (1,1) {};
\node (v1') at (1,0) {};
\node (w1') at (1,-1) {};
\node (u2') at (3,1) {};
\node (v2') at (3,0) {};
\node (w2') at (3,-1) {};
\node (u2) at (4,1) {};
\node (v2) at (4,0) {};
\node (w2) at (4,-1) {};
\node (tl) at (-1.5,1) {};
\node (tl1) at (-2,0.55) {};%
\node (tl2) at (-1,0.55) {};%
\node (bl) at (-1.5,-1) {};
\node (bl1) at (-2,-0.55) {};%
\node (bl2) at (-1,-0.55) {};%
\node (ml1) at (-2,0) {};%
\node (ml2) at (-1,0) {};%
\node (tr) at (5.5,1) {};
\node (tr1) at (6,0.55) {};%
\node (tr2) at (5,0.55) {};%
\node (br) at (5.5,-1) {};
\node (br1) at (6,-0.55) {};%
\node (br2) at (5,-0.55) {};%
\node (mr1) at (6,0) {};%
\node (mr2) at (5,0) {};%
\end{scope}

\begin{scope}[]d
\node [above=-2pt of u1] {$u_1$};
\node [above=-2pt of v1] {$v_1$};
\node [below=-1pt of w1] {$w_1$};
\node [above=-2pt of u2] {$u_2$};
\node [below=-1pt of v2] {$v_2$};
\node [below=-1pt of w2] {$w_2$};
\end{scope}

\begin{scope}[thick]
\draw (tl2) -- (ml2) -- (bl2) -- (bl) -- (bl1) -- (ml1) -- (tl1) -- (tl) -- (u1) -- (u1') -- (u2') -- (u2) -- (v2) -- (v2') -- (v1') -- (v1) -- (w1) -- (w1') -- (w2') -- (w2) -- (br) -- (br1) -- (mr1) -- (tr1) -- (tr) -- (tr2) -- (mr2) -- (br2);
\draw (tl) -- (tl2) -- (tl1);
\draw (bl1) -- (bl2);
\draw (bl) -- (w1);
\draw (br) -- (br2) -- (br1);
\draw (tr1) -- (tr2);
\draw (tr) -- (u2);
\draw [line width=2pt] (v1) -- (w1);
\draw [line width=2pt] (v2) -- (u2);
\end{scope}

\begin{pgfonlayer}{background}
\begin{scope}[rounded corners, thick, blue]
\draw (0.5,1.5) rectangle (-2.5,-1.5) node[left=-2pt] {$S_1$};
\draw (3.5,1.5) rectangle (6.5,-1.5) node[right=-2pt] {$S_2$};
\end{scope}
\begin{scope}[every node/.style={ns, inner sep=.6em, draw=blue, rectangle, rounded corners}]
\node at (tl) {};
\node at (tl1) {};
\node at (tl2) {};
\node at (bl) {};
\node at (bl1) {};
\node at (bl2) {};
\node at (tr) {};
\node at (tr1) {};
\node at (tr2) {};
\node at (br) {};
\node at (br1) {};
\node at (br2) {};
\end{scope}
\end{pgfonlayer}

\node[thick, draw, blue, rectangle, rounded corners, inner sep=0.7em] (laminarLegend) at (-3.9,0.66666) {};
\node[left=0.25cm of laminarLegend] {laminar \smash{family} $\mathcal{L}$};
\draw[line width=2pt] (-3.5,0) --++ (-0.8,0) node[left=0.15cm] {$c(e)=1$};
\draw[thick] (-3.5,-0.66666) --++ (-0.8,0) node[left=0.15cm] {$c(e)=0$};

\node[anchor=west] at (7.5,0.66666) {$a_{\color{blue}S_1}=0,\quad b_{\color{blue}S_1}=3$};
\node[anchor=west] at (7.5,0) {$a_{\color{blue}S_2}=0, \quad b_{\color{blue}S_2}=3$};
\node[anchor=west] at (7.5,-0.66666) {$a_{%
\tikz[baseline=($(x)!0.5!(xx.south)$)]{%
\node[ns, semithick, draw=gray, inner sep=0pt, minimum size=2.5pt, anchor=center] (x) at (0,0) {};
\node[ns, inner sep=.3em, draw=blue, rectangle, rounded corners=1, anchor=center, fill=none] (xx) at (x) {};
\draw[semithick, draw=gray] (x) --++ (0:0.5em);
\draw[semithick, draw=gray] (x) --++ (60:0.5em);
\draw[semithick, draw=gray] (x) --++ (-90:0.5em);}
}=0,
\quad b_{
\tikz[baseline=($(x)!0.5!(xx.south)$)]{%
\node[ns, semithick, draw=gray, inner sep=0pt, minimum size=2.5pt, anchor=center] (x) at (0,0) {};
\node[ns, inner sep=.3em, draw=blue, rectangle, rounded corners=1, anchor=center, fill=none] (xx) at (x) {};
\draw[semithick, draw=gray] (x) --++ (0:0.5em);
\draw[semithick, draw=gray] (x) --++ (60:0.5em);
\draw[semithick, draw=gray] (x) --++ (-90:0.5em);}
}=2$};

\end{tikzpicture} \caption{An instance of the MLCST problem: The laminar family $\mathcal{L}$ is given by the blue sets, with lower and upper bounds indicated on the right.
Edge costs $c\colon E\to\mathbb{R}_{\geqslant 0}$ are $0$ on all edges except for $c((v_1,w_1))=c((u_2,v_2))=1$.}
\label{fig:fullInstance}
\end{figure}

The problem instance in \cref{fig:fullInstance} is very similar to the instance discussed in \cref{sec:extensionMLCST} (\cref{fig:incompatiblePatterns}).
While the latter had edges $(u_i,w_i)$ for $i\in\{1,2\}$, the vertices $u_i$ and $w_i$ are connected by an auxiliary graph in the new instance.
This auxiliary graph has the following two crucial properties:
\begin{enumerate}
\item\label{propItem:intInfeasibility} The auxiliary graph does not contain a spanning tree that satisfies the laminar constraints.
\item\label{propItem:fracFeasibility} The spanning tree polytope of the auxiliary graph contains a point that satisfies the laminar constraints.
\end{enumerate}
In other words, the two properties state that it is possible to ``fractionally connect'' $u_i$ and $w_i$ in the auxiliary graph, while an integral solution cannot connect $u_i$ and $w_i$ through the auxiliary graph.
In particular, this implies that any feasible integral solution will use both edges $(v_1,w_1)$ and $(u_2,v_2)$, and thus have cost at least~$2$.
One such integral solution is given in \cref{subfig:OPTsol}.
Observe that no matter how we choose an integral solution, the connectivity patterns induced on the sets $S_1$ and $S_2$ will always be $\sfc[1]$ and $\sfc[2]$, respectively, as indicated in \cref{subfig:OPTsol}.

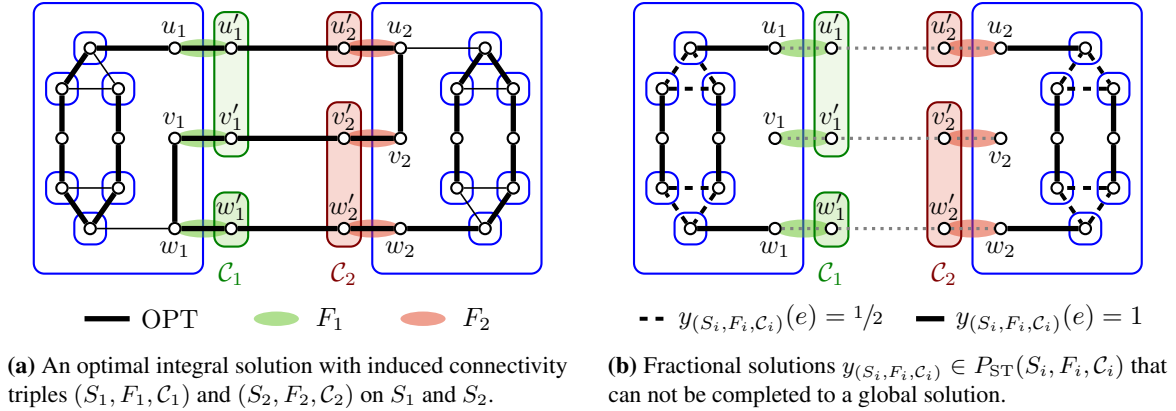
\begin{figure}[!ht]
\hfill
\begin{subfigure}{0.45\linewidth}
\centering
﻿\begin{tikzpicture}[yscale=1.6, scale=0.75]
\small

\pgfdeclarelayer{background}
\pgfsetlayers{background,main}

\clip (-2.6,-2.2) rectangle (6.6,1.55);

\begin{scope}[every node/.style={ns}]
\node (u1) at (0,1) {};
\node (v1) at (0,0) {};
\node (w1) at (0,-1) {};
\node (u1') at (1,1) {};
\node (v1') at (1,0) {};
\node (w1') at (1,-1) {};
\node (u2') at (3,1) {};
\node (v2') at (3,0) {};
\node (w2') at (3,-1) {};
\node (u2) at (4,1) {};
\node (v2) at (4,0) {};
\node (w2) at (4,-1) {};
\node (tl) at (-1.5,1) {};
\node (tl1) at (-2,0.55) {};%
\node (tl2) at (-1,0.55) {};%
\node (bl) at (-1.5,-1) {};
\node (bl1) at (-2,-0.55) {};%
\node (bl2) at (-1,-0.55) {};%
\node (ml1) at (-2,0) {};%
\node (ml2) at (-1,0) {};%
\node (tr) at (5.5,1) {};
\node (tr1) at (6,0.55) {};%
\node (tr2) at (5,0.55) {};%
\node (br) at (5.5,-1) {};
\node (br1) at (6,-0.55) {};%
\node (br2) at (5,-0.55) {};%
\node (mr1) at (6,0) {};%
\node (mr2) at (5,0) {};%
\end{scope}

\begin{scope}[]
\node [above=-2pt of u1] {$u_1$};
\node [above=-2pt of v1] {$v_1$};
\node [below=-1pt of w1] {$w_1$};
\node [above=-2pt of u2] {$u_2$};
\node [below=-1pt of v2] {$v_2$};
\node [below=-1pt of w2] {$w_2$};

\node [above=-3pt of u1'] {$u_1'$};
\node [above=-3pt of v1'] {$v_1'$};
\node [above=-3pt of w1'] {$w_1'$};
\node [above=-3pt of u2'] {$u_2'$};
\node [above=-3pt of v2'] {$v_2'$};
\node [above=-3pt of w2'] {$w_2'$};
\end{scope}

\begin{scope}[semithick]
\draw[line width=2pt] (tl2) -- (ml2) -- (bl2) -- (bl) -- (bl1) -- (ml1) -- (tl1) -- (tl) -- (u1) -- (u1') -- (u2') -- (u2) -- (v2) -- (v2') -- (v1') -- (v1) -- (w1) -- (w1') -- (w2') -- (w2) -- (br) -- (br1) -- (mr1) -- (tr1) -- (tr) -- (tr2) -- (mr2) -- (br2);
\draw (tl) -- (tl2) -- (tl1);
\draw (bl1) -- (bl2);
\draw (bl) -- (w1);
\draw (br) -- (br2) -- (br1);
\draw (tr1) -- (tr2);
\draw (tr) -- (u2);
\end{scope}

\begin{pgfonlayer}{background}
\begin{scope}[rounded corners, thick, blue]
\draw (0.5,1.5) rectangle (-2.5,-1.50);%
\draw (3.5,1.5) rectangle (6.5,-1.5);%
\end{scope}
\begin{scope}[every node/.style={ns, inner sep=.6em, draw=blue, rectangle, rounded corners}]
\node at (tl) {};
\node at (tl1) {};
\node at (tl2) {};
\node at (bl) {};
\node at (bl1) {};
\node at (bl2) {};
\node at (tr) {};
\node at (tr1) {};
\node at (tr2) {};
\node at (br) {};
\node at (br1) {};
\node at (br2) {};
\end{scope}
\begin{scope}[fill opacity=0.5]
\fill[fill=green!50!brown] (0.5,1) ellipse [x radius=0.5cm,y radius=0.1cm];
\fill[fill=green!50!brown] (0.5,0) ellipse [x radius=0.5cm,y radius=0.1cm];
\fill[fill=green!50!brown] (0.5,-1) ellipse [x radius=0.5cm,y radius=0.1cm];
\fill[fill=red!50!brown] (3.5,1) ellipse [x radius=0.5cm,y radius=0.1cm];
\fill[fill=red!50!brown] (3.5,0) ellipse [x radius=0.5cm,y radius=0.1cm];
\fill[fill=red!50!brown] (3.5,-1) ellipse [x radius=0.5cm,y radius=0.1cm];
\end{scope}
\begin{scope}[rounded corners, thick, fill opacity=0.2]
\filldraw[fill=green!50!brown!50, draw=green!50!black] (0.7, -0.2) rectangle (1.3,1.4);
\filldraw[fill=green!50!brown, draw=green!50!black] (0.7, -1.2) rectangle (1.3,-0.6);
\filldraw[fill=red!50!brown, draw=red!50!black] (2.7, 0.8) rectangle (3.3,1.4);
\filldraw[fill=red!50!brown, draw=red!50!black] (2.7, -1.2) rectangle (3.3,0.4);
\end{scope}
\end{pgfonlayer}

\node[below=7pt of w1', green!50!black] {$\mathcal{C}_1$};
\node[below=7pt of w2', red!50!black] {$\mathcal{C}_2$};

\draw[line width=2pt] (-1.6,-2) --++ (0.8,0) node[right] {$\OPT$};
\fill[fill=green!50!brown, fill opacity=0.5] (1.8,-2) ellipse [x radius=0.4cm,y radius=0.08cm] node[right=0.4cm, black, opacity=1] {$F_1$};
\fill[fill=red!50!brown, fill opacity=0.5] (4.4,-2) ellipse [x radius=0.4cm,y radius=0.08cm] node[right=0.4cm, black, opacity=1] {$F_2$};

\end{tikzpicture} \caption{An optimal integral solution with induced connectivity triples $\sfc[1]$ and $\sfc[2]$ on $S_1$ and $S_2$.}
\label{subfig:OPTsol}
\end{subfigure}
\hfill
\begin{subfigure}{0.45\linewidth}
\centering
\begin{tikzpicture}[yscale=1.6, scale=0.75]
\small

\pgfdeclarelayer{background}
\pgfsetlayers{background,main}

\clip (-2.6,-2.2) rectangle (6.6,1.55);

\begin{scope}[every node/.style={ns}]
\node (u1) at (0,1) {};
\node (v1) at (0,0) {};
\node (w1) at (0,-1) {};
\node (u1') at (1,1) {};
\node (v1') at (1,0) {};
\node (w1') at (1,-1) {};
\node (u2') at (3,1) {};
\node (v2') at (3,0) {};
\node (w2') at (3,-1) {};
\node (u2) at (4,1) {};
\node (v2) at (4,0) {};
\node (w2) at (4,-1) {};
\node (tl) at (-1.5,1) {};
\node (tl1) at (-2,0.55) {};%
\node (tl2) at (-1,0.55) {};%
\node (bl) at (-1.5,-1) {};
\node (bl1) at (-2,-0.55) {};%
\node (bl2) at (-1,-0.55) {};%
\node (ml1) at (-2,0) {};%
\node (ml2) at (-1,0) {};%
\node (tr) at (5.5,1) {};
\node (tr1) at (6,0.55) {};%
\node (tr2) at (5,0.55) {};%
\node (br) at (5.5,-1) {};
\node (br1) at (6,-0.55) {};%
\node (br2) at (5,-0.55) {};%
\node (mr1) at (6,0) {};%
\node (mr2) at (5,0) {};%
\end{scope}

\begin{scope}[]
\node [above=-2pt of u1] {$u_1$};
\node [above=-2pt of v1] {$v_1$};
\node [below=-1pt of w1] {$w_1$};
\node [above=-2pt of u2] {$u_2$};
\node [below=-1pt of v2] {$v_2$};
\node [below=-1pt of w2] {$w_2$};

\node [above=-3pt of u1'] {$u_1'$};
\node [above=-3pt of v1'] {$v_1'$};
\node [above=-3pt of w1'] {$w_1'$};
\node [above=-3pt of u2'] {$u_2'$};
\node [above=-3pt of v2'] {$v_2'$};
\node [above=-3pt of w2'] {$w_2'$};
\end{scope}

\begin{scope}[very thick, gray, dotted]
\draw (u1) -- (u1') -- (u2') -- (u2);
\draw (v2) -- (v2') -- (v1') -- (v1);
\draw (w1) -- (w1') -- (w2') -- (w2);
\end{scope}
\begin{scope}[line width=2pt]
\draw (u1) -- (tl);
\draw (tl1) -- (ml1) -- (bl1);
\draw (tl2) -- (ml2) -- (bl2);
\draw (w1) -- (bl);
\draw (u2) -- (tr);
\draw (tr1) -- (mr1) -- (br1);
\draw (tr2) -- (mr2) -- (br2);
\draw (w2) -- (br);
\end{scope}
\begin{scope}[line width=1.5pt, dashed]
\draw (tl) -- (tl1) -- (tl2) -- (tl);
\draw (bl) -- (bl1) -- (bl2) -- (bl);
\draw (tr) -- (tr1) -- (tr2) -- (tr);
\draw (br) -- (br1) -- (br2) -- (br);
\end{scope}

\begin{pgfonlayer}{background}
\begin{scope}[rounded corners, thick, blue]
\draw (0.5,1.5) rectangle (-2.5,-1.5);%
\draw (3.5,1.5) rectangle (6.5,-1.5);%
\end{scope}
\begin{scope}[every node/.style={ns, inner sep=.6em, draw=blue, rectangle, rounded corners}]
\node at (tl) {};
\node at (tl1) {};
\node at (tl2) {};
\node at (bl) {};
\node at (bl1) {};
\node at (bl2) {};
\node at (tr) {};
\node at (tr1) {};
\node at (tr2) {};
\node at (br) {};
\node at (br1) {};
\node at (br2) {};
\end{scope}
\begin{scope}[fill opacity=0.5]
\fill[fill=green!50!brown] (0.5,1) ellipse [x radius=0.5cm,y radius=0.1cm];
\fill[fill=green!50!brown] (0.5,0) ellipse [x radius=0.5cm,y radius=0.1cm];
\fill[fill=green!50!brown] (0.5,-1) ellipse [x radius=0.5cm,y radius=0.1cm];
\fill[fill=red!50!brown] (3.5,1) ellipse [x radius=0.5cm,y radius=0.1cm];
\fill[fill=red!50!brown] (3.5,0) ellipse [x radius=0.5cm,y radius=0.1cm];
\fill[fill=red!50!brown] (3.5,-1) ellipse [x radius=0.5cm,y radius=0.1cm];
\end{scope}
\begin{scope}[rounded corners, thick, fill opacity=0.2]
\filldraw[fill=green!50!brown!50, draw=green!50!black] (0.7, -0.2) rectangle (1.3,1.4);
\filldraw[fill=green!50!brown, draw=green!50!black] (0.7, -1.2) rectangle (1.3,-0.6);
\filldraw[fill=red!50!brown, draw=red!50!black] (2.7, 0.8) rectangle (3.3,1.4);
\filldraw[fill=red!50!brown, draw=red!50!black] (2.7, -1.2) rectangle (3.3,0.4);
\end{scope}
\end{pgfonlayer}

\node[below=7pt of w1', green!50!black] {$\mathcal{C}_1$};
\node[below=7pt of w2', red!50!black] {$\mathcal{C}_2$};

\draw[line width=1.5pt, dashed] (-2.4,-2) --++ (0.5,0) node[right] {$y_{\sfc[i]}(e) = \sfrac12$};
\draw[line width=2pt] (2.5,-2) --++ (0.5,0) node[right] {$y_{\sfc[i]}(e) = 1$};

\end{tikzpicture} \caption{Fractional solutions $y_{\sfc[i]}\in\PST\sfc[i]$ that can not be completed to a global solution.}
\label{subfig:fracSol}
\end{subfigure}
\hfill~
\caption{Optimal integral and partial fractional solutions connect vertices differently inside $S_1$ and $S_2$.}
\label{fig:OPTvsPartialFracSol}
\end{figure}

To analyze our DP approach by classical backtracing of an optimal solution, our goal would be to find partial solutions $y_{\sfc[i]}\in\PST\sfc[i]$ for $i\in\{1,2\}$, and show that a common extension of these solutions has smaller value than the actual optimal solution we started with.
In our example, however, the property in \cref{propItem:fracFeasibility} of the auxiliary graph allows for partial fractional solutions $y_{\sfc[i]}$ that differ substantially from integral solutions in terms of connectivity.
More precisely, the two fractional solutions $y_{\sfc[1]}$ and $y_{\sfc[2]}$ given in \cref{subfig:fracSol} are both of cost $0$ (and hence optimal), but there does not exist a common extension that is feasible for the natural linear relaxation of our problem instance at all.

\medskip

There is one last caveat that has to be addressed: Our dynamic program is designed to construct partial solutions inside all \emph{small} cuts of the laminar family for any choice of edges in the small cuts and corresponding connectivity patterns, and it always extends previously found solutions.
Above, the threshold $\tau$ for deciding whether a cut is small was implicitly assumed to be at least~$3$ so that both $S_1$ and $S_2$ are small cuts.
In the particular example, this implies that all the other cuts in $\mathcal{L}$ (which are precisely the singleton cuts) would be small cuts, as well, forcing our dynamic program to first construct partial solutions in these small cuts and only then extend to $S_1$ and $S_2$.
This would inevitably lead to integral partial solutions, hence we do need a setting where the singleton cuts are \emph{large} cuts.

To achieve this, we introduce dummy edges that increase the number of edges in the singleton cuts of $\mathcal{L}$.
More precisely, for any given threshold $\tau$ and every small singleton cut $\{x\}\in\mathcal{L}$, we can modify the problem instance as follows to turn $\{x\}$ into a large cut: Introduce new vertices $x_1,\ldots,x_\tau$ and edges $\{x, x_i\}$ for $i\in \{1,\dots, \tau\}$, and increase the bounds $a_{\{x\}}$ and $b_{\{x\}}$ by $\tau$.
Feasible solutions of the old and the new instance are in one-to-one correspondence and can be transformed into one another by adding or removing all the edges $\{x, x_i\}$, which are obviously part of any feasible solution of the new instance.

\medskip

To conclude, by introducing dummy edges in the graph given in \cref{fig:fullInstance} as described above, we obtain an instance where an analysis of our DP approach by backtracing an optimal solution in a classical way fails, which supports our novel approach.

\end{document}